\documentclass[11pt]{article}
\usepackage[margin=1in]{geometry}

\usepackage{etoolbox}
\usepackage{todonotes}
\usepackage{dsfont}
\newtoggle{DEBUG}
\togglefalse{DEBUG}
\newtoggle{COMMENTS}
\toggletrue{COMMENTS}

\usepackage{microtype}
\usepackage{graphicx}

\usepackage{subcaption}
\usepackage{booktabs} 

\usepackage{hyperref}

\usepackage{amsmath,bbm}
\usepackage{amssymb}
\usepackage{mathtools}
\usepackage{amsthm}
\usepackage{enumitem}

\usepackage{algorithm}
\usepackage{algorithmicx}
\usepackage{algpseudocode}
\usepackage{thm-restate}

\usepackage[capitalize,noabbrev]{cleveref}

\usepackage{tikz,enumitem}
\usepackage{footmisc} 
\usepackage{authblk} 
\usetikzlibrary{decorations.pathreplacing}
\usetikzlibrary{calc,angles,positioning,intersections,quotes,decorations.markings}

\newcommand{\etal}{{\em et al.}}

\definecolor{gray}{rgb}{0.5,0.5,0.5}
\newcommand{\ashish}[1]{\iftoggle{COMMENTS}{{\color{blue}[AG: \textsf{#1}]}}{}}
\newcommand{\zhihao}[1]{\iftoggle{COMMENTS}{{\color{brown}[ZZ: \textsf{#1}]}}{}}
\newcommand{\sss}[1]{ {\color{orange} [ Sahasrajit: {#1} ]  } }
\newcommand{\zhihaohalfcomment}[1]{\iftoggle{DEBUG}{{\color{brown}[ZZ: \textsf{#1}]}}{}}
\newcommand{\ak}[1]{\iftoggle{COMMENTS}{{\color{teal}[AK: \textsf{#1}]}}{}}

\newcommand{\vect}{\textbf}

\DeclareMathOperator{\expectation}{\mathbb{E}}
\DeclareMathOperator{\probability}{\mathbb{P}}

\DeclarePairedDelimiter\floor{\lfloor}{\rfloor}
\DeclarePairedDelimiterX\set[1]\lbrace\rbrace{\def\given{\;\delimsize\vert\;}#1}

\DeclareMathOperator*{\argmin}{arg\,min}

\newlist{pfparts}{description}{1}
\setlist[pfparts,1]{%
  font=\normalfont\textsf,
  itemindent=2pt,
  wide,
  itemsep=0pt,topsep=2pt,
  labelsep=0.75ex
}

\makeatletter
\def\namedlabel#1#2{\begingroup
    #2%
    \def\@currentlabel{#2}%
    \phantomsection\label{#1}\endgroup
}
\makeatother
\allowdisplaybreaks
\newcommand{\remove}[1]{}

\theoremstyle{plain}
\newtheorem{theorem}{Theorem}[section]

\newtheorem{lemma}[theorem]{Lemma}

\theoremstyle{definition}
\newtheorem{definition}[theorem]{Definition}
\newtheorem*{definition*}{Definition}
\newtheorem{assumption}[theorem]{Assumption}
\theoremstyle{remark}

\newtheorem{example}{Example}
\newtheorem{claim}{Claim}

\title{Differential Privacy with Multiple Selections\thanks{Authors are listed in alphabetical order.}}
\author{
\begin{minipage}[t]{0.3\textwidth}
        \centering
        \textbf{Ashish Goel}\\
        \texttt{Stanford University}\\
        \texttt{ashishg@stanford.edu}
    \end{minipage}\hfill
    \begin{minipage}[t]{0.3\textwidth}
        \centering
        \textbf{Zhihao Jiang}\\
        \texttt{Stanford University}\\
        \texttt{faebdc@stanford.edu}
    \end{minipage}\hfill
    \begin{minipage}[t]{0.3\textwidth}
        \centering
        \textbf{Aleksandra Korolova}\\
        \texttt{Princeton University}\\
        \texttt{korolova@princeton.edu}
    \end{minipage}\\
    \begin{minipage}[t]{0.45\textwidth}
        \centering
        \textbf{Kamesh Munagala}\\
        \texttt{Duke University}\\
        \texttt{kamesh@cs.duke.edu}
    \end{minipage}\hfill
    \begin{minipage}[t]{0.45\textwidth}
        \centering
        \textbf{Sahasrajit Sarmasarkar}\\
        \texttt{Stanford University}\\
        \texttt{sahasras@stanford.edu}
    \end{minipage}
}
\date{}

\begin{document}

\maketitle

\begin{abstract}
We consider the setting where a user with sensitive features wishes to obtain a recommendation from a server in a differentially private fashion. We propose a ``multi-selection'' architecture where the server can send back multiple recommendations and the user chooses one from these that matches best with their private features. When the user feature is one-dimensional -- on an infinite line -- and the accuracy measure is defined w.r.t some increasing function $\mathfrak{h}(.)$ of the distance on the line, we precisely characterize the optimal mechanism that satisfies differential privacy. The specification of the optimal mechanism includes both the distribution of the noise that the user adds to its private value, and the algorithm used by the server to determine the set of results to send back as a response and further show that Laplace is an optimal noise distribution. We further show that this optimal mechanism results in an error that is inversely proportional to the number of results returned when the function $\mathfrak{h}(.)$ is the identity function. 
\end{abstract}

\newtheorem{problem}{Problem}
\newtheorem{procedure}{Procedure}

\newtheorem*{lemma*}{Lemma}
\newtheorem*{theorem*}{Theorem}

\renewcommand{\etal}{\emph{et al.}}

\renewcommand{\Pr}[1]{\probability\left(#1\right)}
\newcommand{\Prxx}[2]{\probability_{#1}\left(#2\right)}
\newcommand{\E}[1]{\expectation\left[#1\right]}
\newcommand{\Ex}[1]{\expectation_{x}\left[#1\right]}
\newcommand{\Exx}[2]{\expectation_{#1}\left[#2\right]}
\newcommand{\Exxlimits}[2]{\mathop{\expectation}\limits_{#1}\left[#2\right]}
\newcommand{\Var}[1]{\mathbf{Var}\left(#1\right)}
\newcommand{\I}[1]{\mathbbm{I}\left\{#1\right\}}
\newcommand{\F}{\mathcal{F}}
\newcommand{\G}{\mathcal{G}}
\renewcommand{\H}{\mathcal{H}}
\newcommand{\Exp}[1]{\exp\left(#1\right)}

\newcommand{\D}{\mathcal{D}}
\newcommand{\C}{\mathcal{H}}
\newcommand{\A}{\mathcal{A}}
\newcommand{\tP}[1]{\tilde{\mathbf{P}}\left(#1\right)}

\newcommand{\ot}{\tilde{O}}

\newcommand{\alert}[1]{\textcolor{red}{#1}}

\newcommand{\gt}{\rightarrow}
\newcommand{\wgt}{\Rightarrow}
\newcommand{\natr}{\mathbb{N}}
\newcommand{\real}{\mathbb{R}}
\newcommand{\ratn}{\mathbb{Q}}
\newcommand{\limt}[3]{\lim\limits_{#1\rightarrow #2}#3}
\newcommand{\limtn}[1]{\lim\limits_{n\rightarrow \infty}#1}

\renewcommand{\P}[1]{\Pname\left(#1\right)}
\newcommand{\Px}[1]{\Pname_{x}\left(#1\right)}
\newcommand{\Pxx}[2]{\Pname_{#1}\left(#2\right)}
\newcommand{\Pname}{\mathbf{P}}

\newcommand{\Pp}[1]{\Ppname\left(#1\right)}
\newcommand{\Ppx}[1]{\Ppname_{x}\left(#1\right)}
\newcommand{\Ppxx}[2]{\tilde{P}_{#1}\left(#2\right)}

\newcommand{\Ppname}{\tilde{\mathbf{P}}}
\newcommand{\Probsname}{\mathcal{P}}
\newcommand{\Probsx}[2]{\Probsname^{(#1)}_{#2}}
\newcommand{\Probsxalt}[2]{\Probsname^{#1}_{#2}}
\newcommand{\Probs}[1]{\Probsx{\epsilon}{#1}}

\newcommand{\Probsloc}[1]{\Probsx{\epsilon-\text{loc}}{#1}}
\newcommand{\Probsgeo}[1]{\Probsx{\epsilon}{#1}}

\newcommand{\Probsgeog}[1]{\Probsxalt{\mathfrak{g}(.)}{#1}}
\newcommand{\Q}[1]{\Qname\left(#1\right)}
\newcommand{\Qx}[1]{\Qname_{x}\left(#1\right)}
\newcommand{\Qxx}[2]{\Qname_{#1}\left(#2\right)}
\newcommand{\Qname}{\mathbf{Q}}
\newcommand{\Qrobs}{\mathcal{Q}}

\newenvironment{hproof}{%
  \renewcommand{\proofname}{Proof Sketch}\proof}{\endproof}

\newcommand\XXX[1]{\begingroup \bfseries\color{red} #1\endgroup}

\newcommand{\R}[1]{\Rname\left(#1\right)}
\newcommand{\Rx}[1]{\Rname_{x}\left(#1\right)}
\newcommand{\Rxx}[2]{\Rname_{#1}\left(#2\right)}
\newcommand{\Rname}{\mathbf{R}}
\newcommand{\Rrobsname}{\mathcal{R}}
\newcommand{\Rrobsx}[1]{\Rrobsname^{(#1)}}
\newcommand{\Rrobs}{\Rrobsx{\epsilon}}

\newcommand{\TRrobsname}{\mathcal{\tilde{R}}}
\newcommand{\TRrobsx}[1]{\TRrobsname^{(#1)}}
\newcommand{\TRrobs}{\TRrobsx{\epsilon}}
\newcommand{\Reta}[1]{\Retaname\left(#1\right)}
\newcommand{\Retax}[1]{\Retaname_{x}\left(#1\right)}
\newcommand{\Retaxx}[2]{\Retaname_{#1}\left(#2\right)}
\newcommand{\Retaname}{\mathbf{R^{(\eta)}}}

\renewcommand{\S}[1]{\Sname\left(#1\right)}
\newcommand{\Sx}[1]{\Sname_{x}\left(#1\right)}
\newcommand{\Sxx}[2]{\Sname_{#1}\left(#2\right)}
\newcommand{\Sname}{\mathbf{S}}
\newcommand{\Srobsname}{\mathcal{S}}
\newcommand{\Srobsx}[1]{\Srobsname^{(#1)}}
\newcommand{\Srobs}{\Srobsx{\epsilon}}

\newcommand{\Lap}{\mathcal{L}_{\epsilon}(0)}
\newcommand{\T}[1]{\Tname\left(#1\right)}
\newcommand{\Tx}[1]{\Tname_{x}\left(#1\right)}
\newcommand{\Txx}[2]{\Tname_{#1}\left(#2\right)}
\newcommand{\Tname}{\mathbf{T}}
\newcommand{\Trobsname}{\mathcal{T}}
\newcommand{\Trobsx}[1]{\Trobsname^{(#1)}}
\newcommand{\Trobs}{\Trobsx{\epsilon}}

\newcommand{\defeq}{:=}
\newcommand{\Dist}{\D}
\newcommand{\Users}{U}
\newcommand{\Signals}{Z}
\newcommand{\Ads}{\Users}
\newcommand{\K}{k}
\newcommand{\Dim}{d}
\newcommand{\Kads}{\real^{\K}}
\newcommand{\rd}{\real^{\Dim}}
\newcommand{\rdk}{\left(\real^{\Dim}\right)^{\K}}
\newcommand{\Psym}{\Probs{sym}}
\newcommand{\Psymx}[1]{\Probsx{#1}{sym}}
\newcommand{\measure}[1]{\mathcal{F}_{#1}}
\newcommand{\ddis}[2]{\mathfrak{D}\left(#1, #2\right)}

\newcommand{\funf}[4]{f(#1,#3,#4)}
\newcommand{\hatfunf}[2]{\hat{f}(#1,#2)}
\newcommand{\funfgeo}[5]{f^{#5}(#3,#4)}
\newcommand{\funfgeoalt}[5]{f^{\text{alt},#5}(#3,#4)}
\newcommand{\hatfunfgeo}[4]{\hat{f}^{#4}(#2,#3)}
\newcommand{\funfloc}[5]{f^{\text{loc},#5}(#1,#3,#4)}
\newcommand{\funforiggeo}[5]{{f}_{\text{unrestricted}}^{#5}(#3,#4)}
\newcommand{\funfnewgeo}[5]{{f}_{\text{restricted}}^{#5}(#3,#4)}
\newcommand{\funforigloc}[5]{{f}_{\text{relaxed}}^{\text{loc},#5}(#1,#3,#4)}
\newcommand{\fung}[2]{g(#1,#2)}
\newcommand{\sig}[2]{\Pxx{#1}{#2}}
\newcommand{\sigq}[2]{\Qxx{#1}{#2}}
\newcommand{\dsig}[1]{P_{#1}}
\newcommand{\dsighat}[1]{\tilde{P}_{#1}}
\newcommand{\dsigq}[1]{Q_{#1}}
\newcommand{\dsigqprime}[1]{Q'_{#1}}
\newcommand{\defdis}[2]{  \left|#1-#2\right|}
\newcommand{\dis}[2]{\left|#1-#2\right|}
\newcommand{\disb}[2]{\left<#1,#2\right>}
\newcommand{\disbb}[2]{\arccos{\left(#1\cdot #2\right)}}
\newcommand{\disg}[2]{\digamma \left(#1,#2\right)}
\newcommand{\diszero}[1]{\left|#1\right|}
\newcommand{\disr}[2]{\left|#1-#2\right|}
\newcommand{\disrzero}[1]{\left|#1\right|}
\newcommand{\zero}{\mathbf{0}}
\newcommand{\zerob}{\mathbf{J}}
\newcommand{\diszerob}[1]{\disb{#1}{\zerob}}
\newcommand{\probm}[1]{\mathbb{D}\left(#1\right)}

\newcommand{\rhovaname}{\rho_{V}}
\newcommand{\rhovbname}{\rho_{VB}}
\newcommand{\rhova}[1]{\rhovaname\left( #1 \right)}
\newcommand{\loRiemannint}[2]{
  \underline{\int_{#1}^{#2}}
}
\newcommand{\rhovb}[1]{\rhovaname\left( #1 \right)}
\newcommand{\Aball}[2]{\mathcal{B}_{#1,#2}}
\newcommand{\Aballsize}[2]{\left| \Aball{#1}{#2} \right|}
\newcommand{\ballsize}[1]{\left| \mathcal{B}_{#1} \right|}

\newcommand{\Axa}{A^{(1)}}
\newcommand{\Sxa}{\Sname^{(1)}}
\newcommand{\Axb}{A^{(2)}}
\newcommand{\Sxb}{\Sname^{(2)}}
\newcommand{\Astar}{A^{*}}

\newcommand{\funextr}[3]{\text{h}_{#1,#2}(#3)}
\newcommand{\lar}{\emph{empty}}
\newcommand{\slar}{\emph{monotone~empty}}
\newcommand{\funmonname}{\text{NMonE}}
\newcommand{\funmon}[3]{\funmonname(#1,#2,#3)} 
\newcommand{\sma}{\emph{full}}
\newcommand{\ssma}{\emph{monotone~full}}
\newcommand{\funmonsmaname}{\text{NMonF}}
\newcommand{\funmonsma}[3]{\funmonsmaname(#1,#2,#3)}
\newcommand{\funnum}[2]{m(#1,#2)}
\newcommand{\sr}[3]{SR_{[#1,#2)}(#3)}
\newcommand{\sra}[3]{SR_{[#1,#2)}(#3)}
\newcommand{\srx}[3]{SR_{[#1,#2)}(#3)}

\newcommand{\sphere}{\mathbb{S}}
\newcommand{\spherex}[1]{\mathbb{S}^{#1}}
\newcommand{\dg}[1]{w\left(#1\right)}

\newcommand{\Disutility}{\text{Dis-util}}

\definecolor{csma}{rgb}{1,1,0.8}
\definecolor{clar}{rgb}{0.8,1,1}
\definecolor{cmsma}{rgb}{1,1,0.5}
\definecolor{cmlar}{rgb}{0.5,1,1}
\definecolor{cano}{rgb}{1,0.8,1}

\newenvironment{proofof}[1]{{\bf Proof of #1:}}{$\qed$\par}
\newenvironment{proofsketch}{{\sc{Proof Outline:}}}{$\qed$\par}


\def\lowint{\mkern3mu\underline{\vphantom{\intop}\mkern7mu}\mkern-10mu\int}
\def\upint{\mathchoice%
    {\mkern13mu\overline{\vphantom{\intop}\mkern7mu}\mkern-20mu}%
    {\mkern7mu\overline{\vphantom{\intop}\mkern7mu}\mkern-14mu}%
    {\mkern7mu\overline{\vphantom{\intop}\mkern7mu}\mkern-14mu}%
    {\mkern7mu\overline{\vphantom{\intop}\mkern7mu}\mkern-14mu}%
  \int}

\section{Introduction}{\label{sec:introduction_setup}}
Consider a user who wants to issue a query to an online server (e.g. to retrieve a search result or an advertisement), but the query itself reveals private information about the use. One commonly studied approach to protect user privacy from the server in this context is for the user to send a perturbed query, satisfying differential privacy under the local trust model~ \cite{bebensee2019local}. However, since the query itself is changed from the original, the server may not be able to return a result that is very accurate for the original query. 
Our key observation is that in many situations such as search or content recommendations, the server is free to return many results, and the user can choose the one that is the most appropriate, without revealing the choice to the server. 
In fact, if the server also returns a model for evaluating the quality of these results for the user, then this choice can be made by a software intermediary such as a client running on the user's device. This software intermediary can also be the one that acts as the user's privacy delegate and is the one ensuring the local differential privacy protection.



We call this, new for the differential privacy (DP) literature system architecture, the ``Multi-Selection'' approach to privacy, and the key question we ask is: {\em What is the precise trade-off that can be achieved between the number of returned results and quality under a fixed privacy goal}? Of course, had the server simply returned all possible results, there would have been no loss in quality since the client could choose the optimal result. However, this approach is infeasible due to computation and communication costs, as well as due to potential server interest in not revealing proprietary information. 
We therefore restrict the server to return $k$ results for small $k$, and study the trade-off between $k$ and the quality when the client sends privacy-preserving  queries. Our algorithmic design space consists of choosing the client's algorithm and the server's algorithm, as well as the space of signals they will be sending.

At a high level, in addition to the novel multi-selection framework for differential privacy, our main contributions are two-fold. First, under natural assumptions on the privacy model and the latent space of results and users, we show a tight trade-off between utility and number of returned results via a natural (yet {\em a priori} non-obvious) algorithm, with the error perceived by a user decreasing as $\Theta(1/k)$ for $k$ results. Secondly, at a technical level, our proof of optimality is via a dual fitting argument and is quite subtle, requiring us to develop a novel duality framework for  linear programs over infinite dimensional function spaces, with constraints on both derivatives and integrals of the variables. This framework may be of independent interest for other applications where such linear programs arise.

\subsection{System Architecture and Definitions}
We denote the set of users by $\mathbb{R}$ and when we refer to a user $u \in \mathbb{R}$, we mean a user with a query value $u \in \mathbb{R}$. 
Let $M$ denote the set of results and $\text{OPT}: \mathbb{R} \rightarrow M$ denotes the function
which maps user queries to optimal results. This function $\text{OPT}$ is available (known) to the server but is unknown to the users.  

\subsubsection{Privacy}
We adopt a well-studied notion of differential privacy under the local trust model~\cite{bebensee2019local}:
\begin{definition}[adapted from \cite{6686179,koufogiannis2015optimality}]
    Let $\epsilon>0$ be a desired level of privacy and let $\mathcal{U}$ be a set of input data and $\mathcal{Y}$ be the set of all possible responses and $\Delta(\mathcal{Y})$ be the set of all probability distributions (over a sufficiently rich $\sigma$-algebra of $\mathcal{Y}$ given by $\sigma(\mathcal{Y})$). A mechanism $Q: \mathcal{U} \rightarrow \Delta(\mathcal{Y})$ is $\epsilon$-differentially private if     for all $S \in \sigma(\mathcal{Y})$ and $u_1,u_2 \in \mathcal{U}$:
    $$ \mathbb{P}(Qu_1 \in S) \leq e^\epsilon \mathbb{P}(Qu_2\in S).$$
\end{definition}

A popular relaxation of differential privacy is geographic differential privacy~\cite{andres2013geo} (GDP), which allows the privacy guarantee to decay with the distance between user values. It reflects the intuition that the user is more interested in protecting the specifics of a medical query they are posing rather than protecting whether they are posing a medical query or an entertainment query, and is thus, appropriate in scenarios such as search.
We restate the formal definition from \cite{koufogiannis2015optimality} and use it in the rest of the work. 
\begin{definition}[adapted from \cite{koufogiannis2015optimality}]{\label{def:geo_DP}}
        Let $\epsilon>0$ be a desired level of privacy and let $\mathcal{U}$ be a set of input data and $\mathcal{Y}$ be the set of all possible responses and $\Delta(\mathcal{Y})$ be the set of all probability distributions (over a sufficiently rich $\sigma$-algebra of $\mathcal{Y}$ given by $\sigma(\mathcal{Y})$). A mechanism $Q: \mathcal{U} \rightarrow \Delta(\mathcal{Y})$ is $\epsilon$-geographic differentially private if for all $S \in \sigma(\mathcal{Y})$ and $u_1,u_2 \in \mathcal{U}$:
    $$ \mathbb{P}(Qu_1 \in S) \leq e^{\epsilon|u_1-u_2|} \mathbb{P}(Qu_2\in S).$$
\end{definition}

\subsubsection{``Multi-Selection" Architecture} \label{sec:simmodel}
Our ``multi-selection" system architecture (shown in \cref{fig:arch1}) relies on the server returning a small set of results in response to the privatized user input, with the on-device software intermediary deciding, unbeknownst to the server, which of these server responses to use.  

\definecolor{usr}{rgb}{0.7,0.85,1.0}
\definecolor{usrl}{rgb}{0.9,0.95,1.0}
\definecolor{usrline}{rgb}{0.8,0.9,1.0}
\definecolor{sev}{rgb}{1.0,0.85,0.7}
\definecolor{sevl}{rgb}{1.0,0.95,0.9}
\definecolor{sevline}{rgb}{1.0,0.9,0.8}

\begin{figure}[htbp]
      \centering
      \tikzset{global scale/.style={
            scale=#1,
            every node/.append style={scale=#1}
          }
        }
        \begin{tikzpicture}

         \node[fill=usrl,draw=usrline,rounded corners, minimum width=10cm, minimum height=3cm] at (2.5,2) {};
         
         \node at (2.5,3.8) {User Side};
        
         \node[fill=usr,draw,rounded corners, minimum width=3cm, minimum height=1.5cm] at (0,2) {\bf User};
         
         \node[fill=usr,draw,rounded corners, minimum width=3cm, minimum height=1.5cm,align=center] at (5,2) {\bf Agent \\ e.g. browser};

         \node[fill=sevl,draw=sevline,rounded corners, minimum width=5cm, minimum height=3cm] at (11.5,2) {};
         
         \node at (11.5,3.8) {Server Side};
         \node[fill=sev,draw,rounded corners, minimum width=4cm, minimum height=1.5cm] at (11.5,2) {\bf Result Computation};

         \draw[->] (1.5,2.25) -- (3.45,2.25);         \draw[<-] (1.5,1.75) -- (3.45,1.75);         \draw[->] (6.5,2.25) -- (9.5,2.25);         \draw[<-] (6.5,1.75) -- (9.5,1.75);

         \node at (2.5,2.5) {$u\in \real$};
         \node at (2.5,1.5) {best in $\textbf{a}$};
         \node at (8,2.5) {$s (\in Z) \sim P_u$};
         \node at (8,1.5) {$\textbf{a} (\in M^k)\sim Q_s$};

        \end{tikzpicture}
      \caption{Overall architecture for multi-selection.}    \label{fig:arch1}
  \end{figure}
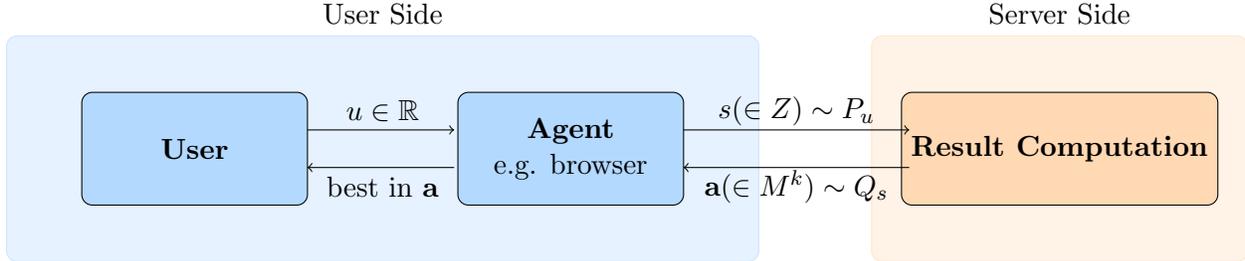


The space of mechanisms we consider in this new architecture consists of a triplet $(Z,\Pname, \Qname)$:
%
%
\begin{enumerate}
    \item A set of signals $Z$ that can be sent by users.
    \item  The actions of users, $\Pname$,  which involves a user sampling a signal from a distribution over signals. We use $\dsig{u}$ for $u\in \mathbb{R}$ to denote the distribution of the signals sent by user $u$ which is supported on $Z$. The set of actions over all users is given by $\Pname = \{\dsig{u}\}_{u \in \real}$.
    \item The distribution over actions of the server, $\Qname$. When the server receives a signal $s\in \Signals$, it responds with  $\dsigq{s}$, which characterizes the distribution of the $k$ results that the server sends (it is supported in $M^k$). We denote the set of all such server actions by $\Qname = \{\dsigq{s}\}_{s \in Z}$. \footnote{We treat this distribution to supported on $U^k$ instead of $k$-sized subset of $U$ for ease of mathematical typesetting. } 
\end{enumerate}

Our central question is to find the triplet over $(Z,\Pname,\Qname)$ that satisfies $\epsilon$-geographic differential privacy constraints on $\Pname$ while ensuring the best-possible utility or the smallest-possible disutility.



\subsubsection{The disutility model: Measuring the cost of privacy}

We now define the disutility of a user $u \in \mathbb{R}$ from a result $m \in M$. 
One natural definition would be to look at the difference between (or the ratio) of the cost of the optimum result for $u$ and the cost of the result $m$ returned by a privacy-preserving algorithm. However, we are looking for a general framework, and do not want to presume that this cost measure is known to the algorithm designer, or indeed, that it even exists. Hence, we will instead define the disutility of $u$ as the amount by which $u$ would have to be perturbed for the returned result $m$ to become optimum; this only requires a distance measure in the space in which $u$ resides, which is needed for the definition of the privacy guarantees anyway. For additional generality, we will also allow the disutility to be any increasing function of this perturbation, as defined below.

\begin{definition}{\label{def:disutility_user_result}}
The dis-utility of a user $u\in \mathbb{R}$ from a result $m\in M$ w.r.t some continuously increasing function $\mathfrak{h}(.)$ is given by \footnote{If no such $u'$ exists then the dis-utility is $\infty$ as infimum of a null set is $\infty$.}
    \begin{equation}{\label{eqn:disutil_user_result}}
        \Disutility^{\mathfrak{h}(.)}(u,m) := \inf\limits_{u'\in \mathbb{R}: OPT(u')=m} \mathfrak{h}(|u-u'|)
    \end{equation}
\end{definition}


We allow any function $\mathfrak{h}(.)$ that satisfies the following conditions:
\begin{align}
     \text{ $\mathfrak{h}(.)$ is a continuously increasing function satisfying $\mathfrak{h}(0)=0$. }\label{item:first-propertyh}\\
    \text{ There exists $\mathcal{B} \in \real^{+}$ s.t. $\log \mathfrak{h}(.)$ is Lipschitz continuous in $[\mathcal{B},\infty)$} \label{item:second-property-h} 
\end{align}
The first condition \eqref{item:first-propertyh} captures the intuition that dis-utility for the optimal result is zero. The second condition \eqref{item:second-property-h} which bounds the growth of $\mathfrak{h}(.)$ by an exponential function is (a not very restrictive condition) required for our mathematical analysis.

Quite surprisingly, to show that our multi-selection framework provides a good trade-off in the above model for every $\mathfrak{h}$ as defined above, we only need to consider the case where the $\mathfrak{h}$ is the identity function. The following example further motivates our choice of the disutility measure:

\begin{example}{\label{example:definition_first}}
Suppose, one assumes that the result set $M$ and the user set $\real$ are embedded in the same metric space $(d,M \cup \real)$. This setup is similar to the framework studied in the examination of metric distortion of ordinal rankings in social choice \cite{anshelevich2016randomized, CARAGIANNIS2022103802, gkatzelis2020resolving, kizilkaya2022plurality}. Using triangle inequality, one may argue that $d(u,m') -d(u,m) \leq 2d(u,u')$ where $m$ is the optimal result for user $u$ (i.e. $m = \argmin\limits_{m \in M} d(u,m)$) and $m'$ is the optimal result for user $u'$. \footnote{This follows since $d(u,m')-d(u,u') \leq d(u',m') \leq d(u',m) \leq d(u,u') + d(u,m)$} Thus, $2d(u,u')$ gives an upper bound on the disutility of user $u$ from result $\text{OPT}(u')$.
\end{example}

\subsubsection{Restricting users and results to the same set} {\label{subsec:user_results_restricted}} 
For ease of exposition, we study a simplified setup restricting the users and results to the same set $\real$. 
%
%
Specifically, since $\Disutility^{\mathfrak{h}(.)}(u,\text{OPT}(u')) \leq \mathfrak{h}(|u-u'|)$, our simplified setup restricts the users and results to the same set $\real$ where the dis-utility of user $u \in \real$ from a result $a \in \real$ is given by $\mathfrak{h}(|u-a|)$. Our results will extend to the model where the users and results lie in different sets, and we refer the reader to Appendix \ref{appendix-sec:restricted_unrestricted_setup} for the treatment. We note that even in the simplified setup, while we use $a\in \real$ to denote the result, what we mean is that the server sends back  $\text{OPT}(a) \in M$.



\subsubsection{Definition of the cost function in the simplified setup}{\label{sec:def_cost_function}}


We use $\texttt{Set}(\vect{a})$ to convert a vector $\vect{a} = (a_1,a_2,\ldots,a_k)^T\in \real^k$ to a set of at most $k$ elements, formally $\texttt{Set}(\vect{a}) = \{a_i: i \in [k]\}$. 


Recall from Section \ref{subsec:user_results_restricted}, the dis-utility of user $u\in \real$ from a result $a \in \real$ in the simplified setup may be written as
\begin{equation}{\label{eqn:disutil_user_result_restricted}}
    \Disutility^{\mathfrak{h}(.)}_{sim}(u,a) = \mathfrak{h}(|u-a|)
\end{equation}
Since we restrict the users and results to the same set, $Q_s$ is supported on $\real^k$ for every $s \in Z$.
%
Thus, the cost for a user $u$ from the mechanism $(Z,\Pname,\Qname)$ is given by 
\begin{align}{\label{eqn:cost_user_tuple_mechanism}}
    \text{cost}^{\mathfrak{h}(.)} (u, (Z,\Pname,\Qname)) & = \Exxlimits{s\sim \dsig{u}}{\Exxlimits{\vect{a}\sim \dsigq{s}}{ \min_{a \in \texttt{Set}(\vect{a})} \Disutility^{\mathfrak{h}(.)}_{sim}(u,a) } }\\ &= \Exxlimits{s\sim \dsig{u}}{\Exxlimits{\vect{a}\sim \dsigq{s}}{ \min_{a \in \texttt{Set}(\vect{a})} \mathfrak{h} (|u-a|)}},
\end{align}
where the expectation is taken over the randomness in the action of user and the server.

We now define the cost of a mechanism by supremizing\footnote{We write supremum instead of maximum as the maximum over an infinite set may not always be defined.} over all users $u \in \real$. Since, we refrain from making any distributional assumptions over the users, supremization rather than mean over the users is the logical choice.
%
%
\begin{equation}{\label{eqn:cost_tuple_mechanism}}
    \text{cost}^{\mathfrak{h}(.)}(Z,\Pname,\Qname) := 
    \sup_{u \in \mathbb{R}} \text{cost}^{\mathfrak{h}(.)} (u, (Z,\Pname,\Qname)) 
    = \sup_{u \in \mathbb{R}} \Exxlimits{s\sim \dsig{u}}{\Exxlimits{\vect{a}\sim \dsigq{s}}{ \min_{a \in \texttt{Set}(\vect{a})} \mathfrak{h} (|u-a|)}}
\end{equation}

\noindent We use $\mathds{1}(.)$ to denote the identity function, i.e. $\mathds{1}(x) = x$ for every $x \in \real$ and thus define the cost function when $\mathfrak{h}(.)$ is an identity function as follows: 
\begin{equation}{\label{eqn:identity_cost_tuple_mechanism}}
    \text{cost}^{\mathds{1}(.)}(Z,\Pname,\Qname) := 
    \sup_{u \in \mathbb{R}} \text{cost}^{\mathds{1}(.)} (u, (Z,\Pname,\Qname)) 
    = \sup_{u \in \mathbb{R}} \Exxlimits{s\sim \dsig{u}}{\Exxlimits{\vect{a}\sim \dsigq{s}}{ \min_{a \in \texttt{Set}(\vect{a})} |u-a|}}
\end{equation}




Recall our central question is to find the triplet over $(Z,\Pname,\Qname)$ that ensures the smallest possible disutility / cost while ensuring that  $\Pname$
satisfies $\epsilon$-geographic differential privacy.
We denote the value of this cost by $\funfgeo{(d,U)}{\Ads}{\epsilon}{\K}{\mathfrak{h}(.)}$ and it is formally defined as

\begin{align}{\label{eqn:cost_function_defn}}
    \funfgeo{(d,U)}{\Ads}{\epsilon}{\K}{\mathfrak{h}(.)}
\defeq & \inf_{\Signals}\inf_{\substack{\Pname\in \Probsgeo{\Signals}}} \inf_{\Qname\in \Qrobs_{ \Signals}} \left(\text{cost}^{\mathfrak{h}(.)}(Z,\Pname,\Qname)\right)\nonumber \\
        = & \inf_{\Signals}\inf_{\substack{\Pname\in \Probsgeo{\Signals}}} \inf_{\Qname\in \Qrobs_{ \Signals}} \left(\sup_{u \in \real} \Exxlimits{s\sim \dsig{u}}{\Exxlimits{\vect{a} \sim \dsigq{s}}{ \min_{a \in \texttt{Set}(\vect{a})} \mathfrak{h} (|u-a|)}} \right) \text{where}
\end{align}


$\Probsgeo{\Signals}:=\{ \Pname | \forall u\in \real, \dsig{u}\text{~is a distribution on $\Signals$,} 
    \text{~and $\Pname$ satisfies $\epsilon$-geographic differential privacy.} \}$.

$\Qrobs_{\Signals}:= \{ \Qname | \forall s\in \Signals, \dsigq{s}\text{~is a distribution on $\real^{\K}$} \}$.


\subsection{Our results and key technical contributions}{\label{sec:results_and_constributions}}





For any $\mathfrak{h}(.)$ satisfying equations \eqref{item:first-propertyh} and \eqref{item:second-property-h} when the privacy goal is $\epsilon$-geographic DP
our main results are: 
\begin{itemize}

    \item The optimal action $P_u$ for a user $u$, is to add Laplace noise\footnote{We use $\mathcal{L}_{\epsilon}(u)$ to denote a Laplace distribution of scale $\frac{1}{\epsilon}$ centred at $u$. Formally, a distribution $X \sim \mathcal{L}_{\epsilon}(u)$ has its probability density function given by $f_X(x) = \frac{\epsilon}{2} e^{-\epsilon|x-u|}$.} of scale $\frac{1}{\epsilon}$ to its value $u$ (result stated in Theorem \ref{thm:optimality-laplace-sim} and proof sketch described in Sections \ref{subsec-comb:proof-step1}, \ref{subsec-comb:primal-feasible} and \ref{subsec-comb:dual-soln-construction}). Further, we emphasize that the optimality of adding Laplace noise is far from obvious\footnote{In fact, only a few optimal DP mechanisms are known~\cite{Ghosh2012privacy, gupte2010universally, kairouz2014extremal, ding2017collecting}, and it is known that for certain scenarios, universally optimal mechanisms do not exist~\cite{brenner2014impossibility}.}. For instance, when users and results are located on a ring, Laplace noise is {\em not optimal} (see Appendix \ref{appendix-sec:non-opt-ex} for an analysis when $k=2$).


    \item 
    The optimal server response $\Qname$ could be different based on different $\mathfrak{h}$. We give a recursive construction of  $\Qname$ for a general $\mathfrak{h}$ (Section \ref{subsec-comb:step4}). Furthermore, when $\mathfrak{h}(t)=t$, we give an exact construction of $\Qname$ (sketched in Fig. \ref{fig:intro1} for $k=5$) and show that  $\funfgeo{(\ell_1,\real)}{\Ads}{\epsilon}{\K}{\mathds{1}(.)} = O(\frac{1}{\epsilon k})$ in Section \ref{subsec-comb:step4} and Appendix \ref{subsec-appendix:server_response}. 

\end{itemize}

Formally, our main results can be stated as: 
\begin{restatable}[corresponds to Theorem \ref{thm:optimality-laplace-sim} and Theorem \ref{cor:geofinal}]{theorem}{simrstepe} \label{thm:lap:sim}
For $\epsilon$-geographic differential privacy, adding Laplace noise, that is, user $u$ sends a signal drawn from distribution $\mathcal{L}_{\epsilon}(u)$, 
 is one of the optimal choices of $\Probs{\Signals}$ for users. Further, when $\mathfrak{h}(t)=t$, we have $\funfgeo{(\ell_1,\real)}{\Ads}{\epsilon}{\K} {\mathds{1}(.)} = O(\frac{1}{\epsilon k})$ and the optimal mechanism $(Z,\Pname,\Qname)$ (choice of actions of users and server) itself can be computed in closed form. For a generic $\mathfrak{h}(.)$, the optimal server action $\Qname$ may be computed recursively.
\end{restatable}  


In addition to our overall framework and the tightness of the above theorem, a key contribution of our work is in the techniques developed. At a high level, our proof proceeds via constructing an infinite dimensional linear program to encode the optimal algorithm under differential privacy constraints. We then use dual fitting to show the optimality of Laplace noise. Finally, the optimal set of results being computable by a simple dynamic program given such noise.

The technical hurdles arise because the linear program for encoding the optimal mechanism is over infinite-dimensional function spaces with linear constraints on both derivatives and integrals, since the privacy constraint translates to constraints on the derivative of the density encoding the optimal mechanism, while capturing the density itself requires an integral. We call it Differential Integral Linear Program (DILP); see Section~\ref{subsec-comb:primal-feasible}. However, there is no weak duality theory for such linear programs in infinite dimensional function spaces, such results only existing for linear programs with integral constraints~\cite{anderson1987linear}.

We therefore develop a weak duality theory for DILPs (see Section~\ref{subsec-comb:primal-feasible} with a detailed proof in Appendix~\ref{sec:weak_duality}), which to the best of our knowledge is novel. The proof of this result is quite technical and involves a careful application of Fatou's lemma \cite{rudin1976principles} and the monotone convergence theorem to interchange integrals with limits, and integration by parts to translate the derivative constraints on the primal variables to derivatives constraint on the dual variables. 

We believe our weak duality framework is of independent interest and has broader implications beyond differential privacy; see Appendix~\ref{sec-appendix:weak_duality_application} for two such applications.

\begin{figure}
    \centering

    \includegraphics[width=80mm]{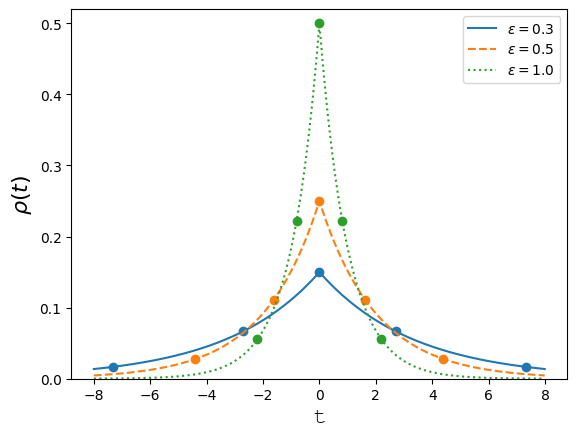}
      \caption{Optimal mechanisms in geographic differential privacy setting when $\K=5$ and $\epsilon\in\{0.3,0.5,1.0\}$. Suppose the user has a private value $u$. Then the user sends a signal $s$ drawn from distribution $\mathcal{L}_{\epsilon}(u)$ to the server, meaning the user sends $s=v+x$ where $x$ is drawn from the density function $\rho(t)$ in this figure. Suppose the server receives $s$. Then the server responds $\{s+a_1,...,s+a_5\}$, where the values of $a_1,a_2,...,a_5$ are the $t$-axis values of dots on the density functions.
      }   \label{fig:intro1}
  \end{figure}

\subsection{Related Work}


\subsubsection{Differential Privacy}
The notion of differential privacy in the trusted curator model is introduced in \cite{dwork2006calibrating}; see \cite{dwork2014algorithmic} for a survey of foundational results in this model. 
%
%
The idea of local differential privacy dates back to \cite{warner1965randomized}, and the algorithms for satisfying it have been studied extensively following the deployment of DP in this model by Google~\cite{RAPPOR} and Apple~\cite{Apple}; see, e.g.~\cite{bun2019heavy, chen2016private, wang2017locally, bassily2017practical} and Bebensee~\cite{bebensee2019local} for a survey.
Geographic differential privacy was introduced by Andr{\'e}s~\etal~\cite{andres2013geo} and has gained widespread adoption for anonymizing location data. 
Our use of geographic DP utilizes the definition of \cite{andres2013geo} with the trust assumptions of the local model, and is thus, only a slight relaxation of the traditional local model of differential privacy. 

\subsubsection{Multi-Selection}
%
An architecture for multi-selection, particularly with the goal of privacy-preserving advertising, was already introduced in \textit{Adnostic} by~\cite{toubiana2010adnostic}. Their proposal was to have a browser extension that would run the targeting and ad selection on the user's behalf, reporting to the server only click information using cryptographic techniques. Similarly, \textit{Privad} by~\cite{guha2011privad} propose to use an anonymizing proxy that operates between the client that sends broad interest categories to the proxy and the advertising broker, that transmits all ads matching the broad categories, with the client making appropriate selections from those ads locally. 
Although both Adnostic and Privad reason about the privacy properties of their proposed systems, unlike our work, neither provides DP guarantees.

Two lines of work in the DP literature can be seen as somewhat related to the multi-selection paradigm -- the exponential mechanism (see e.g.~\cite{mcsherry2007mechanism, blum2013learning, liu2019private}) and amplification by shuffling (see e.g. \cite{erlingsson2019amplification, cheu2019distributed, feldman2022hiding}). The exponential mechanism focuses on high-utility private selection from multiple alternatives and is usually deployed in the TCM model. Amplification by shuffling analyzes the improvement in the DP guarantees that can be achieved if the locally privatized data is shuffled by an entity before being shared with the server. As far as we are aware, neither of the results from these lines of work can be directly applied to our version of multi-selection, although combining them is an interesting avenue for future work.


More broadly, several additional directions within DP research can be viewed as exploring novel system architectures in order to improve privacy-utility trade-offs, e.g., using public data to augment private training~\cite{papernot2018scalable, bassily2020private}, combining data from TCM and LM models~\cite{blender, dubey2018power, beimel2020power}, and others. Our proposed architecture is distinct from all of these. Finally, our work is different from how privacy is applied in federated learning~\cite{geyer2019differentially} -- there, the goal is for a centralized entity to be able to build a machine learning model based on distributed data; whereas our goal is to enable personalized, privacy-preserving retrieval from a global ML model.

\subsubsection{Optimal DP mechanisms}{\label{sec:related_optimal_DP}}
To some extent, much of the work in DP can be viewed as searching for the optimal DP mechanism, i.e. one that would achieve the best possible utility given a fixed desired DP level. 
Only a few optimal mechanisms are known~\cite{Ghosh2012privacy, gupte2010universally, kairouz2014extremal, ding2017collecting}, and it is known that for certain scenarios, universally optimal mechanisms do not exist~\cite{brenner2014impossibility}.

Most closely related to our work is the foundational work of~\cite{Ghosh2012privacy} that derives the optimal mechanism for counting queries via a linear programming formulation; the optimal mechanism turns out to be the discrete version of the Laplace mechanism.
Its continuous version was studied in \cite{fernandes2021laplace}, where the Laplace mechanism was shown to be optimal. 
These works focused on the trusted curator model of differential privacy unlike the local trust model which we study. 

In the local model,  \cite{koufogiannis2015optimality} show Laplace noise to be optimal for $\epsilon$-geographic DP. Their proof relies on formulating an infinite dimensional linear program over noise distributions and looking at its dual. Although their proof technique bears a slight resemblance to ours, our proof is different and intricate since it involves the minimisation over the set of returned results in the cost function. 

A variation of local DP is considered in \cite{geng2015optimal}, in which DP constraints are imposed only when the distance between two users is smaller than a threshold. For that setting, the optimal noise is piece-wise constant, which is a similar outcome to our optimal Laplace noise distribution.
However, our setting of choosing from multiple options makes the problems very different.



\subsubsection{Related work on duality theory for infinite dimensional LPs}
{\label{sec:related_theory_dilp_optimization}}

Strong duality is known to hold for finite dimensional linear programs \cite{weakduality}. However, for infinite dimensional linear programs, strong duality may not always hold (see \cite{anderson1987linear} for a survey). Sufficient conditions for strong duality are presented and discussed in \cite{VINH20161,basu2015strong} for generalized vector spaces. A class of linear programs (called SCLPs) over bounded measurable function spaces have been studied in \cite{doi:10.1137/0331073,bellman2010dynamic} with integral constraints on the functions. However, these works do not consider the case with derivative constraints on the functional variables. In \cite[Equation 7]{koufogiannis2015optimality} a linear program with derivative and integral constraints (DILPs) is formulated to show the optimality of Laplace noise for geographic differential privacy. However, their duality result does not directly generalize to our case since our objective function and constraints are far more complex as it involves minimisation over a set of results.

We thus need to reprove the weak-duality theorem for our DILPs, the proof of which is technical and involves a careful application of integration by parts to translate the derivative constraint on the primal variable to a derivative constraint on the dual variable. Further, we require a careful application of 
Fatou's lemma \cite{rudin1976principles} and monotone convergence theorem to exchange limits and integrals. Our weak duality result generalizes beyond our specific example and is applicable to a broader class of DILPs. Furthermore, we discuss two problems (one from job scheduling \cite{anderson1981new} and one from control theory \cite{evans1983introduction}) in Appendix \ref{sec-appendix:weak_duality_application} which may be formulated as a DILP.

\section{Characterizing the Optimal Mechanism: Proof Sketch of Theorem~\ref{thm:lap:sim}}     \label{sec:simproof}


We now present a sketch of the proof of Theorem~\ref{thm:lap:sim}; the full proof involving the many technical details is presented in the Appendix. We first show that for the sake of analysis, the server can be removed by making the signal set coincide with result sets (Section~\ref{subsec-comb:proof-step1}) assuming that the function $\text{OPT}$  is publicly known both to the user and server.\footnote{One should note that this removal is just for analysis and the server is needed since the $\text{OPT}$ function is unknown to the user.} Then in Section~\ref{subsec-comb:primal-feasible} we construct a primal linear program $\mathcal{O}$ for encoding the optimal mechanism, and show that it falls in a class of infinite dimensional linear programs that we call DILPs, as defined below.

\begin{definition}{\label{def:DILP}}
 Differential-integral linear program (DILP) is a linear program over Riemann integrable function spaces involving constraints on both derivatives and integrals. 
\end{definition}

A simple example is given in Equation \eqref{eq:example_DILP}. Observe that in equation \eqref{eq:example_DILP}, we define $\mathcal{C}_1$ to be a continuously differentiable function.

\begin{equation}{\label{eq:example_DILP}}
        \text{\Large{$\tilde{\mathcal{O}}=$} }
        \left\{
        \begin{aligned}
            \inf_{g(.): \mathcal{C}^1(\mathbb{R} \rightarrow \mathbb{R}^{+})} & \int_{\mathbb{R}} |v| g(v) dv\\
            \textrm{s.t.} \quad & \int_{\mathbb{R}} g(v) dv =1\\
            & -\epsilon g(v) \leq g'(v) \leq \epsilon g(v) \text{ }\forall v \in \mathbb{R}
        \end{aligned}
        \right.
\end{equation}

We next construct a dual DILP formulation $\mathcal{E}$ in Section~\ref{subsec-comb:primal-feasible}, and show that the formulation satisfies weak duality. As mentioned before, this is the most technically intricate result since we need to develop a duality theory for DILPs. We relegate the details of the proof here to the Appendix. 

Next, in Section~\ref{subsec-comb:dual-soln-construction}, we show the optimality of the Laplace noise mechanism via dual-fitting, {\em i.e.}, by constructing a feasible solution to DILP $\mathcal{E}$ with objective identical to that of the Laplace noise mechanism. Finally, in Section~\ref{subsec-comb:step4}, we show how to find the optimal set of $k$ results given Laplace noise. We give a construction for general functions $\mathfrak{h}(.)$ as well as a closed form for the canonical case of $\mathfrak{h}(t)=t$. This establishes the error bound and concludes the proof of Theorem~\ref{thm:lap:sim}.

\subsection{Restricting the signal set $Z$ to $\mathbb{R}^k$} 
\label{subsec-comb:proof-step1}
We first show that it is sufficient to consider a more simplified setup where the user sends a signal set in $\real^k$ and the server sends back the results corresponding to the signal set. Since we assumed users and ads lie in the same set, for the purpose of analysis, this removes the server from the picture. To see this, note that for the setting  discussed in Section \ref{subsec:user_results_restricted}, the optimal result for user $u$ is the result $u$ itself, where when we refer to ``result $u$'', we  refer to the result $\text{OPT}(u) \in M$. 

Thus, this approach is used only for a simplified analysis as the $\text{OPT}$ function is not known to the user and our final mechanism will actually split the computation between the user and the server.


Therefore a user can draw a result set directly from the distribution over the server's action and send the set as the signal. 
The server returns the received signal, hence removing it from the picture. In other words, it is sufficient to consider mechanisms in $\Probs{\Kads}$, which are in the following form (corresponding to Theorem \ref{thm:singpoint:sim}).

\begin{enumerate}
    \item User $v\in \real$ reports $s$ that is drawn from a distribution $\dsig{v}$ over $\real^{\K}$.
    \item The server receives $s$ and returns $s$.
\end{enumerate}


We give an example to illustrate this statement below.


\begin{example}
Fix a user $u$ and let $\Signals=\{s_1,s_2\}$ and $\{\vect{a}^{(1)},\vect{a}^{(2)}\} \subseteq \mathbb{R}^k$ and consider a mechanism $\mathcal{M}_1$ where user $u$ sends $s_1$ and $s_2$ with equal probability. The server returns $\vect{a} \in \real^k$ on receiving signal $s$, with the following probability.
\begin{align*}
    &\Pr{\vect{a}=\vect{a}^{(1)} | s=s_1}=0.2,\quad
    \Pr{\vect{a}=\vect{a}^{(2)} | s=s_1}=0.8, \\
    &\Pr{\vect{a}=\vect{a}^{(1)} | s=s_2}=0.4,\quad
    \Pr{\vect{a}=\vect{a}^{(2)} | s=s_2}=0.6.
\end{align*}
Then the probability that $u$ receives $\vect{a}^{(1)}$ is 0.3 and it receives $\vect{a}^{(2)}$ is 0.7. Now consider another mechanism $\mathcal{M}_2$ with the same cost satisfying differential privacy constraints, where $\Signals=\{\vect{a}^{(1)},\vect{a}^{(2)}\}$, with user $u$ sending signal $\vect{a}^{(1)}$ and $\vect{a}^{(2)}$ with probabilities 0.3 and 0.7.
When the server receives $\vect{a}\in \Kads$, it returns $\vect{a}$.
\end{example}

We show the new mechanism $\mathcal{M}_2$ satisfies differential privacy assuming the original mechanism $\mathcal{M}_1$ satisfies it. As a result, we can assume $\Signals=\Kads$ when finding the optimal mechanism. 


The following theorem states that it is sufficient to study a setup removing the server from the picture and consider mechanisms in set of probability distributions supported on $\real^k$ satisfying $\epsilon$-geographic differential privacy ($\Probs{\real^{\K}}$ as defined in Section \ref{sec:def_cost_function}). 

\begin{restatable}
[detailed proof in Appendix \ref{sec-appendix:simproof}] {theorem}{simrstepa}\label{thm:singpoint:sim}

It is sufficient to remove the server $(\Qname)$ from the cost function $\funfgeo{(d,\Users)}{\Ads}{\epsilon}{\K}{\mathfrak{h}(.)}$ and pretend the user has full-information. Mathematically, it maybe stated as follows.
\begin{equation}{\label{eq:new_expr_f}} 
\funfgeo{(d,\Users)}{\Ads}{\epsilon}{\K}{\mathfrak{h}(.)} =\inf_{\Pname\in \Probs{\real^{\K}}} \sup_{u \in \real} \Exxlimits{\vect{a} \sim \dsig{u}}{ \min_{a\in \texttt{Set} (\vect{a})}\mathfrak{h}(|u-a|) }.
\end{equation}
\end{restatable}


\begin{proof}[Proof Sketch]

Fix $\Signals,\Pname\in \Probs{\Signals},\Qname\in \Qrobs_{\Signals}$. For $u\in \real$ and $S\subseteq \real^{\K}$, let $\Ppxx{u}{S}$ be the probability that the server returns a set in $S$ to user $u$. Then for any $u_1,u_2\in \real, S\subseteq \real^{\K}$, we can show that $\Ppxx{u_1}{S} \leq e^{\epsilon\cdot |u_1-u_2|} \Ppxx{u_2}{S}$ using post-processing theorem, and thus $\Ppname = \{\dsighat{u}\}_{u \in \real} \in \Probs{\real^{\K}}$ because $\dsighat{u}$ is a distribution on $\real^{\K}$ for any $u\in \real$. At the same time,


\begin{align*}
\Exxlimits{s\sim \Pname_u}{\Exxlimits{\vect{a}\sim \Qname_s}{ \min_{a\in \texttt{Set}(\vect{a})}\mathfrak{h}(|u-a|) }}
=\Exxlimits{\vect{a}\sim \dsighat{u}}{ \min_{a\in \texttt{Set}(\vect{a})}\mathfrak{h}(|u-a|)) }, \text{ so we have}
\end{align*}
\begin{align*}
\funfgeo{(d,\Users)}{\Ads}{\epsilon}{\K}{\mathfrak{h}(.)} = \inf_{\Pname\in \Probs{\real^{\K}}} \sup_{u \in \real} \Exxlimits{\vect{a} \sim \dsig{u}}{ \min_{a\in \texttt{Set} (\vect{a})}\mathfrak{h}(|u-a|) },
\end{align*}



\end{proof}


\subsection{Differential integral linear programs to represent $f(\epsilon,k)$ and a weak duality result}{\label{subsec-comb:primal-feasible}}

Recall the definition of DILP from Definition \ref{def:DILP}. In this section, we construct an infimizing DILP $\mathcal{O}$ to represent the constraints and the objective in the cost function $f(\epsilon,k)$.  We then construct a dual  DILP $\mathcal{E}$, and provide some intuition for this formulation. The proof of weak duality is our main technical result, and its proof is defered to the Appendix.

\subsubsection{Construction of DILP $\mathcal{O}$ to represent cost function  $f(\epsilon,k)$}

We now define the cost of a mechanism $\Pname$ which overloads the cost definition in Equation \ref{eqn:cost_tuple_mechanism} 

\begin{definition}{\label{defn:cost}}
    Cost of mechanism $\textbf{P} \in \mathcal{P}^{(\epsilon)}_{\mathbb{R}^k}$: We define the cost of mechanism $\textbf{P}$ as
    \begin{equation}{\label{eq:cost}}
        \textit{cost}(\textbf{P}) := \sup_{u \in \mathbb{R}} \underset{\vect{a} \sim \dsig{u}}{\mathbb{E}} \left[\min_{a \in \texttt{Set}(\vect{a})} \mathfrak{h}(|u-a|) \right]
    \end{equation}
\end{definition}

Observe that in Definition \ref{defn:cost} 
 we just use $\Pname$ instead of the tuple $(Z,\Pname,\Qname)$ as in Equation \eqref{eqn:cost_tuple_mechanism}. Observe that it is sufficient to consider $\Pname$ in the cost since $\Pname$ simulates the entire combined action of the user and the server as shown in Theorem \ref{eq:new_expr_f} in Section \ref{subsec-comb:proof-step1}.  We now define the notion of approximation using cost of mechanism by a sequence of mechanisms which is used in the construction of DILP $\mathcal{O}$.

\begin{definition}{\label{defn:cost-approx}}
    Arbitrary cost approximation: We call mechanisms $\textbf{P}^{(\eta)} \in \mathcal{P}^{(\epsilon)}_{\mathbb{R}^k}$ an arbitrary cost approximation of mechanisms $\textbf{P} \in \mathcal{P}^{(\epsilon)}_{\mathbb{R}^k}$ if $\lim\limits_{\eta \to 0} \textit{cost}(\textbf{P}^{(\eta)}) = \textit{cost}(\textbf{P})$  
\end{definition}

Now we define the DILP $\mathcal{O}$ to characterise ${f}(\epsilon,k)$  in Equation \eqref{eq:new_expr_f}. In this formulation, the variables are $g(.,.) : \mathcal{I}^B (\mathbb{R} \times \mathbb{R}^k \rightarrow \mathbb{R}^{+})$, which we assume are non-negative \textit{Riemann integrable} bounded functions. These variables capture the density function $P_u$. 



\begin{equation}{\label{orig_primal}}
    \text{ \Large{$\mathcal{O} =$}}
\left\{
\begin{aligned}
    \inf_{g(.,.): \mathcal{I}^B (\mathbb{R} \times \mathbb{R}^k \rightarrow \mathbb{R}^{+}), \kappa \in \mathbb{R}} \quad & \kappa  \\
    \textrm{s.t.} \quad &\kappa - \int\limits_{\vect{x} \in \mathbb{R}^k} \left[\min_{a \in \texttt{Set}(\vect{x})} \mathfrak{h}(|u-a|) \right] g(u,\vect{x}) d\left(\prod_{i=1}^{k} x_i\right) \geq 0 \text{ }\forall u\in \mathbb{R}\\
    & \int\limits_{\vect{x} \in \mathbb{R}^k} g(u,\vect{x}) d\left(\prod_{i=1}^{k} x_i\right) = 1 \text{ } \forall u\in \mathbb{R}\\
    &  \epsilon g(u,\vect{x}) + \underline{g}_{u}(u,\vect{x}) \geq 0; \text{ }\forall u\in \mathbb{R}; \vect{x} \in \mathbb{R}^k\\
    &  \epsilon g(u,\vect{x}) - \overline{g}_{u}(u,\vect{x}) \geq 0; \text{ }\forall u\in \mathbb{R}; \vect{x} \in \mathbb{R}^k
\end{aligned}
\right.
\end{equation}

In DILP $\mathcal{O}$, we define $\underline{g}_u(u,\vect{x})$ and $\overline{g}_u(u,\vect{x})$ to be the lower and upper partial derivative of $g(u,\vect{x})$ at $u$. Now observe that, we use lower and upper derivatives instead of directly using derivatives as the derivatives of a probability density function may not always be defined (for example, the left and right derivatives are unequal in the Laplace distribution at origin). 

Note that the DILP $\mathcal{O}$ involves integrals and thus requires mechanisms to have a valid probability density function, however not every distribution is continuous, and, as a result, may not have a density function (e.g. point mass distributions like $\hat{P}^{\mathcal{L}_{\epsilon}}_u$ defined in Definition \ref{def:laplace_noise_addition}). Using ideas from mollifier theory \cite{99c43556-16f5-3ed5-90f5-6c2bbf257ea3} we construct mechanisms $\textbf{P}^{(\eta)}$ with a valid probability density function that are an arbitrary good approximation of mechanism $\textbf{P}$ in Lemma \ref{lemma:arbitrary_aprroximation}, hence showing that it suffices to define $\mathcal{O}$ over bounded, non-negative Reimann integrable functions $g$.






We now prove that the DILP constructed above captures the optimal mechanism, in other words, $\text{opt}(\mathcal{O})= f(\epsilon, k)$.  

\begin{lemma}{\label{lemma:lp_formulation}}
        Let $\text{opt}(\mathcal{O})$ denote the optimal value of DILP \eqref{orig_primal}, then  
    $f(\epsilon,k) = \text{opt}(\mathcal{O})$.
\end{lemma}

To prove this lemma, we show Lemma~\ref{lemma:diff-privacy-condition}, which relates the last two constraints of the DILP $\mathcal{O}$ to $\epsilon$-geographic differential privacy, and Lemma~\ref{lemma:arbitrary_aprroximation}, which shows that an arbitrary cost approximation of mechanism $\textbf{P}$ can be constructed with valid probability density functions.



\begin{lemma}{\label{lemma:diff-privacy-condition}}
    Let $P_u$ have a probability density function given by $g(u,.): \mathbb{R}^k \rightarrow \mathbb{R}$ for every $u \in \mathbb{R}$. Then, $\Pname$ satisfies $\epsilon$-geographic differential privacy iff $\max(|\overline{g}_u(u,\vect{x})|, |\underline{g}_u(u,\vect{x})|) \leq \epsilon g(u,\vect{x}) \text{ } \forall u \in \mathbb{R};\forall \vect{x} \in \mathbb{R}^k$ \footnote{$\underline{g}_u(u,\vect{x})$, $\overline{g}_u(u,\vect{x})$ denote the lower and upper partial derivative w.r.t. $u$}   
\end{lemma} 

The proof of this lemma (Appendix~\ref{sec:diff-privacy-condition-proof}) proceeds by showing that $\epsilon$-geographic differential privacy is equivalent to Lipschitz continuity of $\log g(u,\vect{x})$ in $u$.\footnote{We handle the case when the $\log$ is not defined as the density is zero at some point separately in the proof.}

\begin{lemma}{\label{lemma:arbitrary_aprroximation}} (Proven in Appendix \ref{sec:proof_arbitrary_approximation})
    Given any mechanism $\textbf{P} \in \mathcal{P}^{(\epsilon)}_{\textbf{R}^k}$ (satisfying $\epsilon$-geographic differential privacy), we can construct a sequence of mechanisms $\textbf{P}^{(\eta)}$ with bounded probability density functions such that $\textbf{P}^{(\eta)}$ is an arbitrary cost approximation of mechanism $\textbf{P} \in \mathcal{P}^{(\epsilon)}_{\textbf{R}^k}$.
\end{lemma}

 
 Using Lemmas \ref{lemma:arbitrary_aprroximation} and \ref{lemma:diff-privacy-condition}, we give the proof of Lemma \ref{lemma:lp_formulation}.

\begin{proof}[Proof of Lemma \ref{lemma:lp_formulation}]
    Consider any $\zeta>0$. As established in Lemma \ref{lemma:arbitrary_aprroximation}, it follows for every mechanism $\textbf{P} \in \mathcal{P}^{(\epsilon)}_{\textbf{R}^k}$, we can construct another mechanism $\textbf{P}^{(\eta)}$ with bounded probability density functions whose cost is a $\zeta$ approximation of the cost of mechanism $\textbf{P}$.
   Thus, we can use Lemma \ref{lemma:diff-privacy-condition} to conclude that the optimum value of DILP $\mathcal{O}$ is precisely $f(\epsilon,k)$.
\end{proof}



\subsubsection{Dual DILP $\mathcal{E}$ and statement of weak duality theorem}
Now, we write the \textit{dual} of the DILP $\mathcal{O}$ as the DILP $\mathcal{E}$ in Equation \eqref{orig_dual}. Observe that, we have the constraint that $\delta(.)$ and $\lambda(.)$ is non-negative, $\mathcal{C}^0$ (continuous) and $\nu(.,.)$
is a $\mathcal{C}^{1}$ function i.e. $\nu(r,\vect{v})$ is \textit{continuously differentiable} in $r$ and \textit{continuous} in $\vect{v}$.
Thus, we may rewrite the equations as 
\begin{equation}{\label{orig_dual}}
    \text{\Large{$\mathcal{E}=$} }
    \left\{
    \begin{aligned}
    \sup_{\substack{\delta(.),\lambda(.): \mathcal{C}^{0}(\mathbb{R}\rightarrow \mathbb{R}^+); \\ \nu(.,.): \mathcal{C}^1(\mathbb{R} \times (\mathbb{R})^k\rightarrow \mathbb{R})} } & \int_{r \in \mathbb{R}} \lambda(r) dr\\
    \textrm {s.t. } \quad & \int_{r \in \mathbb{R}} \delta(r) dr \leq 1\\
     & \hspace{- 3em} -\left[\min_{a \in \texttt{Set}(\vect{v})} \mathfrak{h}(|r-a|) \right] \delta(r) + \lambda(r) + \nu_r(r,\vect{v}) + \epsilon|\nu(r,\vect{v})| \leq 0 \text{ }\forall  r \in \mathbb{R}; \vect{v} \in (\mathbb{R})^k\\
    & \exists U: \mathcal{C}^0(\mathbb{R}^k \rightarrow \mathbb{R})  \textrm { s.t. }\nu(r,\vect{v}) \geq 0 \text{ }\forall r\geq U(\vect{v}) \text{ } \forall \vect{v} \in (\mathbb{R})^k \\
    & \exists L: \mathcal{C}^0(\mathbb{R}^k \rightarrow \mathbb{R}) \textrm { s.t. } \nu(r,\vect{v}) \leq 0 \text{ } \forall r\leq L(\vect{v}) \text{ } \forall \vect{v} \in (\mathbb{R})^k 
\end{aligned}
\right.
\end{equation}


To get intuition behind the construction of our dual DILP $\mathcal{E}$, relate the equations in DILP $\mathcal{O}$ to the dual variables of DILP $\mathcal{E}$ as follows. The first equation denoted by $\{\delta(r)\}_{r \in \mathbb{R}}$, second equation denoted by $\{\lambda(r)\}_{r \in \mathbb{R}}$ and the last two equations are \textit{jointly} denoted by $\{\nu(r,\vect{v})\}_{r \in \mathbb{R}; \vect{v} \in (\mathbb{R})^k}$ 
\footnote{Note that the variable $\nu(r,\vect{v})$ is constructed from the difference of two non-negative variables corresponding to third and fourth equations, respectively. The detailed proof is in Appendix \ref{sec:weak_duality},}.\\ The last two terms in the second constraint of DILP $\mathcal{E}$ are a consequence of the last two equations on DILP $\mathcal{O}$ and observe that it involves a derivative of the dual variable $\nu(u,\vect{v})$. The linear constraint on the derivative of the primal variable translates to a derivative constraint on the dual variable by a careful application of integration by parts, discussed in detail in Appendix \ref{sec:weak_duality}. 


Observe that in our framework we have to prove the weak-duality result as, to the best of our knowledge, existing duality of linear programs in infinite dimensional spaces work for cases involving just integrals. 
The proof of this Theorem \ref{theorem:weak_duality_result} is technical and we defer the details to Appendix \ref{sec:weak_duality}. 


\begin{theorem}{\label{theorem:weak_duality_result}}
    $\text{opt}(\mathcal{O}) \geq \text{opt}(\mathcal{E})$.
\end{theorem}

\subsection{Dual fitting to show the optimality of Laplace noise addition}\label{subsec-comb:dual-soln-construction}

Before starting this section, we first define a function $\hatfunf{\epsilon}{\K}$ which characterises the optimal placement of $k$ points in $\real$ to minimise the expected minimum dis-utility among these $k$ points measured with respect to some user $u$ sampled from a Laplace distribution. As we shall prove in Theorem \ref{lemma:bounding_f} that it bounds the cost of the Laplace noise addition mechanism.

\begin{equation}{\label{eq:hat_f_defn}}
    \hatfunf{\epsilon}{\K} = \min_{\vect{a} \in \mathbb{R}^k} \Exxlimits{y \sim \mathcal{L}_{\epsilon}(0)}{\min_{a \in \texttt{Set}(\vect{a})} \mathfrak{h} (|y -a |)}
\end{equation}

In this section, we first define a mechanism in Definition \ref{def:laplace_noise_addition} which simulates the action of the server corresponding to the Laplace noise addition mechanism in Section \ref{subsec:sim_laplace_noise} and show that the cost of Laplace noise addition mechanism is $\hat{f}(\epsilon,k)$. We finally show the optimality of Laplace noise addition mechanism via dual fitting i.e. constructing a feasible solution to the dual DILP $\mathcal{E}$ with an objective function $\hat{f}(\epsilon,k)$ in Section \ref{subsec:feasible_soln}.




%

\subsubsection{Bounding cost function $f(\epsilon,k)$ by the cost of Laplace noise adding mechanism}{\label{subsec:sim_laplace_noise}}

We now define the mechanism $\hat{\textbf{P}}^{\mathcal{L}_{\epsilon}} = \{\hat{P}_u^{\mathcal{L}_{\epsilon}}\}_{u \in \real}$ which corresponds to simulating the action of the server on receiving signal $S_u\sim \mathcal{L}_{\epsilon}(u)$ from user $u$. We often call this in short as the Laplace noise addition mechanism. 

\begin{definition}{\label{def:laplace_noise_addition}}
The distribution $\hat{P}^{\mathcal{L}_{\epsilon}}_u$ is defined as follows for every $u \in \real$.
    \begin{align}{\label{hat_P_defn}}
    \hat{\vect{a}} \sim \hat{P}^{\mathcal{L}_{\epsilon}}_u  \iff \hat{\vect{a}} = & \argmin_{\vect{a} \in \mathbb{R}^k} \Exxlimits{y \sim \mathcal{L}_{\epsilon}(S_u)}{\min_{a \in \texttt{Set}(\vect{a})} \mathfrak{h} (|y -a |)} \text{ where } S_u \sim \mathcal{L}_{\epsilon}(u)\\
    \overset{(a)}{=} & \argmin_{\vect{a} \in \mathbb{R}^k} \Exxlimits{y \sim \mathcal{L}_{\epsilon}(0)}{\min_{a \in \texttt{Set}(\vect{a})} \mathfrak{h} (|y -a |)} + S_u \text{\footnotemark where } S_u \sim \mathcal{L}_{\epsilon}(u)
\end{align}

 Equality $(a)$ follows from the fact that $y \sim \mathcal{L}_{\epsilon}(z) \implies y-z\sim \mathcal{L}_{\epsilon}(0)$ for every $z \in \mathbb{R}$.
\footnotetext{Observe that we choose a deterministic tie-breaking rule amongst all vectors minimising this objective.} 
\end{definition}

Observe that the server responds with set of points $\texttt{Set}(\vect{a})$ for some $\vect{a} \in \mathbb{R}^k$ so as to minimise the expected cost with respect to some user sampled from a Laplace distribution centred at $S_u$.


We show that the following lemma which states that $\hat{P}^{\mathcal{L}_{\epsilon}}$ satisfies $\epsilon$-geographic differential privacy constraints and bound $f(\epsilon,k)$ by $\hat{f}(\epsilon,k)$.

\begin{restatable} 
[detailed proof in Appendix \ref{sec-appendix:primal-feasible}] {lemma}{simrstepa}\label{lemma:bounding_f}
    $\hat{\textbf{P}}^{\mathcal{L}_{\epsilon}}$  satisfies $\epsilon$-geographic differential-privacy constraints i.e. $\hat{\textbf{P}}^{\mathcal{L}_{\epsilon}} \in \mathcal{P}^{(\epsilon)}_{\mathbb{R}^k}$ and thus, we have $f(\epsilon,k) \leq \textit{cost}(\hat{{P}}^{\mathcal{L}_{\epsilon}}) = \hat{f}(\epsilon,k)$
     
\end{restatable}

     

\begin{proof}[Proof Sketch]

     Observe that $\hat{P}^{\mathcal{L}_{\epsilon}} \in \mathcal{P}^{(\epsilon)}_{\mathbb{R}^k}$ from the post processing theorem, refer to \cite{diff-privacy-book} since $S_u \sim \mathcal{L}_{\epsilon}(u)$ satisfies $\epsilon$-geographic differential privacy constraints.\footnote{Post processing theorem can be proven even for $\epsilon$-geographic differential privacy similarly}. Thus, we prove $f(\epsilon,k) \leq \text{cost}(\hat{{P}}^{\mathcal{L}_{\epsilon}})$. The equality is fully proven in Appendix \ref{sec-appendix:primal-feasible}.

\end{proof}

\subsubsection{Obtaining a feasible solution to DILP $\mathcal{E}$}{\label{subsec:feasible_soln}}

We now construct feasible solutions to DILP $\mathcal{E}$. For some $\zeta>0$ and $\hat{\lambda}>0$, we define  

\begin{equation}{\label{eqn:nu_delta_defn}}
    \delta^{(c)}(r) = (\zeta/2) e^{-\zeta|r|} \text{ and } \lambda^{(c)}(r) = \hat{\lambda}\cdot (\zeta/2) e^{-\zeta|r|} \text{ } \forall r \in \mathbb{R}
\end{equation}


Now define $v_{med} = \text{Median}(\texttt{Set}(\vect{v}))$ and for every $\vect{v} \in \mathbbm{R}^k$, we consider the following Differential Equation \eqref{eqn:diff_eqn_nu} in $\hat{\nu}(.)$. 

\begin{equation}{\label{eqn:diff_eqn_nu}}
    -\left[\min_{a \in \texttt{Set}(\vect{v})} \mathfrak{h}(|r-a|)\right] \delta^{(c)}(r) + \lambda^{(c)}(r) + \frac{d\hat{\nu}(r)}{dr} + \epsilon |\hat{\nu}(r)| = 0 ; \text{ with  $\hat{\nu}(v_{med}) = 0$}  
\end{equation}

Observe that this equation precisely corresponds to the second constraint of DILP $\mathcal{E}$ (inequality replaced by equality) with an initial value. We now show that a solution $\hat{\nu}(.)$ to differential equation \eqref{eqn:diff_eqn_nu} exists such that $\hat{\nu}(r)$ is non-negative for sufficiently large $r$ and non-positive for sufficiently small $r$ to satisfy the last two constraints of DILP $\mathcal{E}$ in Lemma \ref{lemma:diff_eqn_soln}.

Observe that the structure of our differential equation is similar to that in \cite[Equation 19]{koufogiannis2015optimality}. However, our differential equation has significantly more complexity since we are minimising over a set of points $\vect{v} \in \mathbb{R}^k$ and also our equation has to be solved for every $\vect{v} \in \mathbb{R}^k$ making it more complex.  


\begin{restatable} 
[Proof in Appendix \ref{subsec-appendix:diff_eqn_soln}] {lemma}{simrstepa}\label{lemma:diff_eqn_soln}
     Choose $\zeta<\epsilon$ and $0< \hat{\lambda} \leq \frac{\epsilon - \zeta}{\epsilon+\zeta} \hat{f}(\epsilon+ \zeta,k)$, then equation $\eqref{eqn:diff_eqn_nu}$ has a unique $\mathcal{C}^1$ solution $\nu^{(c)}(.)$ and there exists $U, L \in \mathbb{R}$ satisfying $\nu^{(c)}(r) \geq 0 \text{ }\forall r \geq U$ and $\nu^{(c)}(r) \leq 0 \text{ } \forall r \leq L$.
     
\end{restatable}

\begin{proof}[Intuitive explanation]
    We just give an intuition for this proof for the case where $\hat{\lambda}$ exceeds $\frac{\epsilon-\zeta}{\epsilon+\zeta} \hat{f}(\epsilon,k)$ by showing two plots in Figure \ref{fig:plotlambda40} and \ref{fig:plotlambda46} for the two cases where $\hat{\lambda} < \frac{\epsilon-\zeta}{\epsilon+\zeta} \hat{f}(\epsilon,k)$ and $\hat{\lambda} > \frac{\epsilon-\zeta}{\epsilon+\zeta} \hat{f}(\epsilon,k)$ respectively. In the first case, $\nu^{(c)}(r)$ is positive for sufficiently large $r$ and in second case, it goes negative for large $r$ demonstrating the requirement of the bound $\frac{\epsilon-\zeta}{\epsilon+\zeta} \hat{f}(\epsilon,k)$ on $\hat{\lambda}$. 


    The two plots are for the case when $\epsilon=1$, $\zeta=0.1$, $\mathfrak{h}(z)=z$ and thus $\frac{\epsilon-\zeta}{\epsilon+\zeta}\hat{f}(\epsilon,k)$ may be approximately by $\frac{9}{11}\times \frac{1}{2}=0.44$ as shown in Section \ref{subsec-comb:step4}. For the purpose of the plots, we choose $\vect{v} = [- \log 4; \text{ }0; \text{ }\log 4]^T$\footnote{We choose this vector since it minimises equation \eqref{eq:hat_f_defn} and a detailed calculation is given in Section \ref{subsec-comb:step4}. 
    }
    and demonstrate the point in the Lemma. 

            \begin{figure}[!htp]
        \centering
        \begin{subfigure}{.5\textwidth}
          \centering
          \includegraphics[width= 0.85\linewidth]{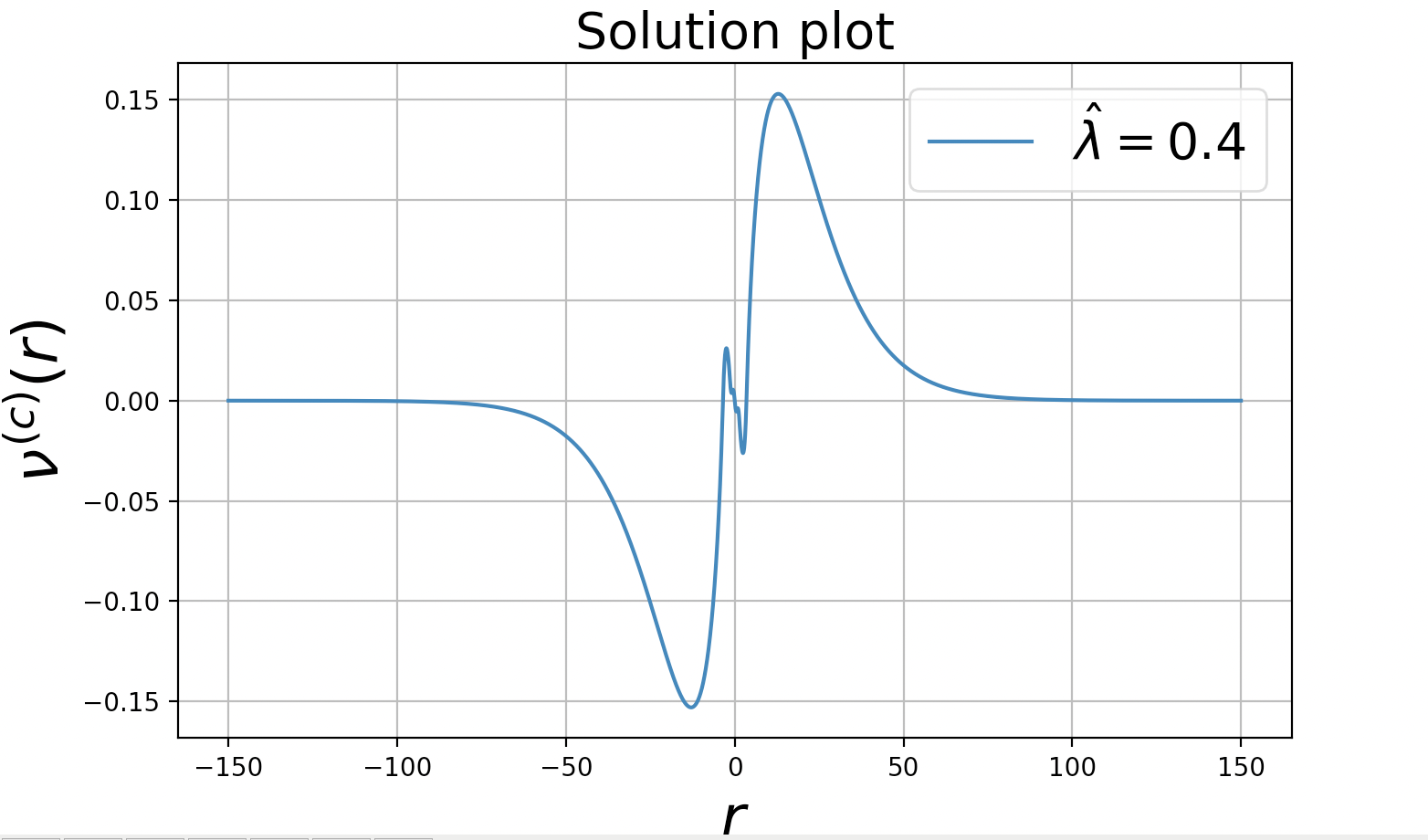}
          \caption{Solution for $\hat{\lambda} = 0.40$}
          \label{fig:plotlambda40}
        \end{subfigure}%
        \begin{subfigure}{.5\textwidth}
          \centering
          \includegraphics[width= 0.85\linewidth]{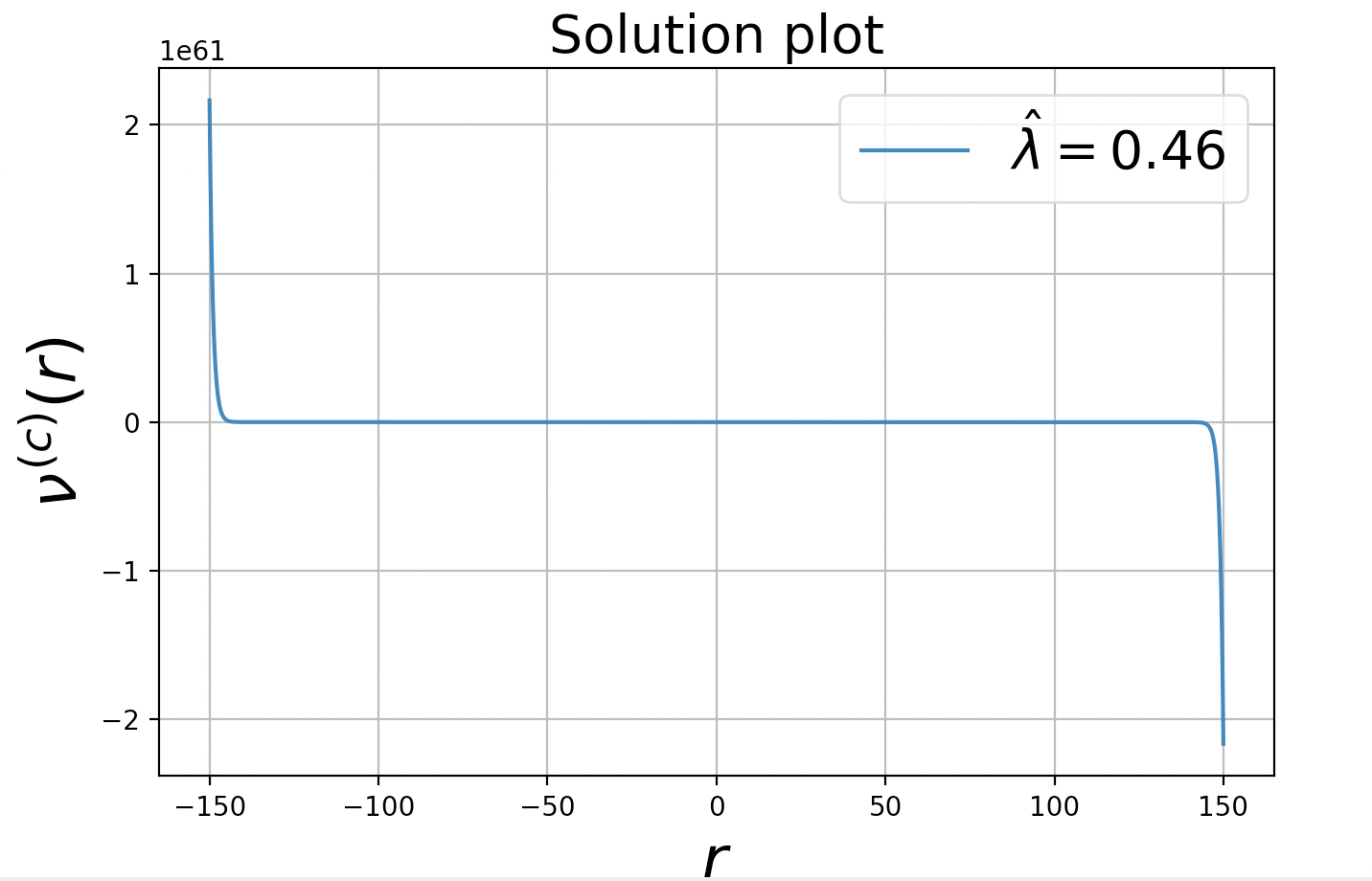}
          \caption{Solution for $\hat{\lambda} = 0.46$}
          \label{fig:plotlambda46}
        \end{subfigure}
        \caption{Solutions for Differential Equation \eqref{eqn:diff_eqn_nu} for $\vect{v} = [- \log 4; \text{ }0; \text{ }\log 4]^T$}
        \end{figure}
These spikes in the solution may be observed due to the selection of $\vect{v} \in \mathbb{R}^3$ due to the term $\left[\min\limits_{a \in \texttt{Set}(\vect{v})} \mathfrak{h}(|r-a|)\right]$ in the differential equation.
\end{proof}

\begin{restatable} 
[detailed proof in Appendix \ref{subsec-appendix:diff_eqn_soln}] {lemma}{simrstepa}\label{lemma:dual_achievable}
    $\text{opt}(\mathcal{E}) \geq \hat{f}(\epsilon,k)$.
\end{restatable}

We present a proof sketch where we do not explicitly show the continuity of the bounds $U(.)$ and $L(.)$. In Appendix \ref{sec:continuous_bounds_claim}, we prove a claim showing the existence of such continuous bounds.


\begin{proof}[Proof Sketch]
    Recall the functions $\lambda^{(c)}(.), \delta^{(c)}(.)$ defined in \eqref{eqn:nu_delta_defn}. Also for every $\vect{v} \in \mathbb{R}^k$, we obtain a function $\nu^{(c)}(.,\vect{v})$ [solution of Equation \eqref{eqn:diff_eqn_nu}] with bounds $U(\vect{v})$ and $L(\vect{v})$ satisfying $\nu^{(c)}(r,\vect{v}) \geq 0 \text{ }\forall u\geq U(\vect{v})$ and $\nu^{(c)}(r,\vect{v}) \leq 0 \text{ }\forall u\leq L(\vect{v})$ and this solution is feasible.


    The objective value of this feasible solution is $\hat{\lambda}$ and the constructed solution is feasible for any $\hat{\lambda} \leq \frac{\epsilon-\zeta}{\epsilon+\zeta} \hat{f}(\epsilon+\zeta,k)$ and $\zeta>0$. Now, since $\hat{f}(\epsilon,k)$ is continuous in $\epsilon$, choosing $\zeta$ to be arbitrarily small enables us to obtain the objective value of the solution arbitrarily close to $\hat{f}(\epsilon,k)$ and thus, $\text{opt}(\mathcal{E}) \geq \hat{f}(\epsilon,k)$. 
\end{proof}


Observe that although we defined the Laplace noise addition mechanism $(\hat{\textbf{P}}^{\mathcal{L}_{\epsilon}})$ (see Definition \ref{def:laplace_noise_addition}) entirely in terms of the user's action, we can consider an alternate mechanism splitting $(\hat{\textbf{P}}^{\mathcal{L}_{\epsilon}})$ into user's action and server's response attaining the same cost: 
\begin{itemize}
    \item User $u$ sends $S_u \sim \mathcal{L}_{\epsilon}(u)$ to the server.

    \item The server on receiving $S_u$ responds with a vector $\vect{a}  = \argmin\limits_{\vect{a} \in \real^k} \Exxlimits{y \sim \mathcal{L}_{\epsilon}(S_u)}{\min\limits_{a \in \texttt{Set}(\vect{a})} \mathfrak{h} (|y -a |)}$.

\end{itemize}

\begin{theorem}{\label{thm:optimality-laplace-sim}}
    For $\epsilon$-geographic differential privacy, sending Laplace noise, that is, user $u$ sends a signal drawn from distribution $\mathcal{L}_{\epsilon}(u)$ is one of the optimal choices of $\Probs{\Signals}$ for users, and in this case $\funfgeo{(\ell_1,\real)}{\Ads}{\epsilon}{\K} {\mathfrak{h}(.)} = \hatfunf{\epsilon}{\K}$.
\end{theorem}

\begin{proof}
    Combining the results in Lemmas \ref{lemma:bounding_f},  \ref{lemma:lp_formulation}, \ref{lemma:dual_achievable} and Theorems  \ref{theorem:weak_duality_result} and Theorem \ref{thm:singpoint:sim}, we obtain $\hat{f}(\epsilon,k) \leq \text{opt}(\mathcal{E}) \leq \text{opt}(\mathcal{O}) \leq \hat{f}(\epsilon,k)$ where $\hat{f}(\epsilon,k)$ denotes the cost of Laplace noise addition mechanism $\hat{P}^{\mathcal{L}_{\epsilon}}$ i.e. $\textit{cost}(\hat{P}^{\mathcal{L}_{\epsilon}}) = \hat{f}(\epsilon,k)$.

\end{proof}

\subsection{Server response given the user sends Laplace Noise}    
\label{subsec-comb:step4}
\newcommand{\kk}{b}


Recall that we proved in Theorem \ref{thm:optimality-laplace-sim} that the Laplace noise addition mechanism is an optimal action for the users. We now focus on the construction of an optimal server action on receiving the signal $s$ from an user.


\begin{enumerate}
    \item User with value $v\in \real$ reports $s$ after adding Laplace noise of scale $\frac{1}{\epsilon}$. More formally, $s$ is drawn from Laplace distribution $\mathcal{L}_{\epsilon}(v)$. 
    \item The server receives $s$ and respond $(s+a_1,\dots,s+a_{\K})$, where $a_1,\dots,a_{\K}$ are fixed real values.
\end{enumerate}

For the case of $\mathfrak{h}(t)=t$, the optimal mechanism is simple enough that the values $a_1, a_2, \ldots, a_k$ can be computed by dynamic programming, as we sketch in the following subsection, and this concludes the proof of Theorem~\ref{thm:lap:sim}. 
For other general increasing functions, the optimal solution for $\{a_i\}_{i=1}^{k}$ may not always be written in closed form, however we can always write a recursive expression to compute the points. 

\subsubsection{Sketch of server response for odd $k$}
We now show the optimal choice of $A$ to optimize cost function $\hat{f}(\epsilon,k)$ [in Equation \eqref{eq:hat_f_defn}]. Specifically, we assume odd $k$ in this section. The solution for even $k$ (refer Theorem \ref{cor:geofinal}) can be constructed using a similar induction where the base case for $k=2$ can be directly optimized. 



Assuming the symmetry of $A$, let $A=\{-y_{\kk-1}, \ldots -y_1, 0, y_1, \ldots, y_{\kk-1}\}$, where $y_1,\ldots,y_{\kk-1}$ are positive numbers in increasing order. We will construct the set $y_1, \ldots, y_{\kk-1}$ inductively.
 Let $x$ be a random variable drawn from Laplace distribution $\Lap$ with parameter $\epsilon$, and the goal is to minimize $D_{\kk}=\Exx{x\sim \Lap}{\min\limits_{a\in A}|x-a|}$. Since the density function of $\Lap$ satisfies $\rho_{\Lap}(x)=\rho_{\Lap}(-x)$, we have 
$$D_{\kk}=\Exx{x\sim \Lap}{\min_{a\in A}\mathfrak{h}(|x-a|) \bigg{|} x} >0,$$ 
i.e. the user has a positive private value. Under this conditioning, the variable $x$ is an exponential random variable of mean 1. In this case, the search result being used by the server will be one of $y_0, y_1$, \ldots, $y_{\kk-1}$. Clearly, $D_1 = 1$. To compute $D_{\kk+1}$, let $s = y_1$. 

Then using the memorylessness property of exponential random variables, we get the recurrence
\begin{align*}
    D_{\kk+1} & = \int_{t=0}^s \min\{\mathfrak{h}(t), \mathfrak{h}(s-t)\}e^{-t}dt +e^{-s}D_\kk\\ & = \int_{t=0}^{s/2} \mathfrak{h}(t)e^{-t}dt + s \int_{t=s/2}^{s} e^{-t}dt - \int_{t=s/2}^{s} \mathfrak{h}(t) e^{-t}dt  +e^{-s}D_\kk.
\end{align*}

The optimal $D_{b+1}$ given $D_b$ can be computed by minimising over all $s \in \mathbb{R}$. However, for the case where $\mathfrak{h}(.)$ is an identity function, we may give a closed form expression below.

\begin{align*}
 D_{\kk+1} 
=& \int_{t=0}^{s/2} te^{-t}dt + s \int_{t=s/2}^{s} e^{-t}dt 
 -  \int_{t=s/2}^{s} t e^{-t}dt  +e^{-s}D_\kk\\
 =& \left(1 - (s/2) e^{-s/2} -  e^{-s/2}\right) + s \left(e^{-s/2}-e^{-s}\right) 
  - \left((s/2) e^{-s/2} +  e^{-s/2} - se^{-s} - e^{-s}\right) + e^{-s}D_\kk\\
 =& 1 - 2e^{-s/2} + e^{-s} + e^{-s}D_\kk \\
=& \left(1-e^{-s/2}\right)^2 + \left(e^{-s/2}\right)^2D_\kk
\end{align*}
Setting $\gamma = e^{-s/2}$, and minimizing by taking derivatives, we get $-2(1-\gamma) + 2\gamma D_\kk = 0$ which in turn gives 
$$\gamma = \frac{1}{D_\kk + 1} \qquad \mbox{and} \qquad D_{\kk+1} = \frac{D_\kk}{D_\kk+1}.$$ 
Plugging in the inductive hypothesis of $D_\kk = 1/\kk$, we get $D_{\kk+1} = 1/(\kk+1)$. Further, we get $s = 2\ln (1+1/\kk)$. Thus, by returning $\K=2\kk-1$ results, the expected ``cost of privacy" can be reduced by a factor of $\kk$. To obtain the actual positions $y_1, .., y_{\kk-1}$ we have to unroll the induction. For $i=1, \ldots, \kk-1$, the position  $y_i$ is given by:
$$y_i =  y_{i-1} + 2\ln (1 + 1/(\kk-i)).$$

Based on the above arguments in the four sections, we have the main theorem \ref{thm:lap:sim}.

\section{Further extensions}{\label{sec:diff-privacy-extensions}}

We describe some additional results below.

\begin{itemize}
    \item When the user is not able to perform the optimal action, we show in Appendix \ref{sec-appendix:mhr} that $\text{cost}^{\mathds{1}(.)}(Z,\Pname,\Qname) = O(\frac{\log k}{k\epsilon})$ for an appropriate server response $\Qname$\footnote{The server's action $\Qname$ involves sampling from the posterior of the noise distribution.} if the user's action $\Pname$ consists of adding symmetric noise whose distribution satisfies log-concave property\footnote{If the random noise with log-concave distribution $g$ is given by $Y$, then we have $\mathbb{E}[Y^{+} \cup \{0\}] = \frac{1}{\epsilon}$ and $g(y) = g(-y)$.}. Observe that this property is satisfied by most natural distributions like Exponential and Gaussian. 

    \item  In practice, the set of users may not belong to $\real$ but in many cases may have a feature vector embedding in $\real^d$. In this scenario, a server could employ dimensionality reduction techniques such as Principal Component Analysis (PCA) to create a small number $d'$ of dimensions which have the strongest correlation to the dis-utility of a hypothetical user with features identical to the received signal. The server may project the received signal only along these dimensions to select the set of $k$ results. In this case, we show that $\text{cost}^{\mathds{1}(.)}(Z,\Pname,\Qname)= O\left(\frac{1}{\epsilon k^{1/d'}}\right)$ under some assumptions as discussed Appendix~\ref{sec-appendix:high-dim-space} when the user's action $\Pname$ consists of adding independent Gaussian noise to every feature.

    \item We further show that Laplace noise continues to be an optimal noise distribution for the user even under a relaxed definition of geographic differential privacy (defined in Definition \ref{def:alt_geo_DP})  
    in Section \ref{sec:diff-privacy-gen}. This definition captures cases when privacy guarantees are imposed only when the distance between users is below some threshold (recall from Section \ref{sec:related_optimal_DP} that such a  setup was studied in \cite{geng2015optimal}).
    
\end{itemize}

\subsection{A generalization of Geographic differential privacy}{\label{sec:diff-privacy-gen}}

Here we consider a generalization of $\epsilon$-geographic differential privacy and define $\mathfrak{g}(.)$-geographic differential privacy for an increasing convex function $\mathfrak{g}(.)$ satisfying Assumption \ref{assn:g_defn}.

\begin{assumption}{\label{assn:g_defn}}
    $\mathfrak{g}(.)$ is a increasing convex function satisfying $\mathfrak{g}(0)=0$ and $\mathfrak{g}(.)$ is differentiable at 0 with $\mathfrak{g}'(0)\neq 0$.
\end{assumption}

\begin{definition}[alternate definition of geo-DP]{\label{def:alt_geo_DP}}
        Let $\epsilon>0$ be a desired level of privacy and let $\mathcal{U}$ be a set of input data and $\mathcal{Y}$ be the set of all possible responses and $\Delta(\mathcal{Y})$ be the set of all probability distributions (over a sufficiently rich $\sigma$-algebra of $\mathcal{Y}$ given by $\sigma(\mathcal{Y})$). For any $\mathfrak{g}(.)$ satisfying Assumption \ref{assn:g_defn} a mechanism $Q: u \rightarrow \Delta(\mathcal{Y})$ is $\mathfrak{g}(.)$-geographic differentially private if 
    for all $S \in \sigma(\mathcal{Y})$ and $u_1,u_2 \in \mathcal{U}$:
    $$ \mathbb{P}(Qu_1 \in S) \leq e^{\mathfrak{g}(|u_1-u_2|)} \mathbb{P}(Qu_2\in S).$$
\end{definition}

Since this definition allows the privacy guarantee to decay non-linearly with the distance between the user values, it is a relaxation of $\epsilon$-geographic DP as defined in Definition \ref{def:geo_DP}. Observe that this definition captures cases where the privacy guarantees exist only when the distance between users is below some threshold by defining $\mathfrak{g}(t)$ to be $\infty$ if $t>T_0$ for some threshold $T_0$. 


Under this notion of differential privacy, we may redefine cost function $f^{\text{alt},\mathfrak{h}(.)}(\epsilon,k)$ as follows. 

\[  \funfgeoalt{U}{\Ads}{\mathfrak{g}(.)}{\K}{\mathfrak{h}(.)}
\defeq \inf_{\Signals}\inf_{\Pname\in \Probsgeog{\Signals}} \inf_{\Qname\in \Qrobs_{ \Signals}} \sup_{u \in \real} \Exxlimits{s\sim \Pname_u}{ \Exxlimits{\vect{a} \sim \Qname_s}{ \min_{a\in \texttt{Set}(\vect{a})}\mathfrak{h}\left(|u-a|\right) }}, \]
where $\Probsgeog{\Signals}:=\{ \Pname | \forall u\in \real, \dsig{u}\text{~is a distribution on $\Signals$,} 
    \text{~and $\mathfrak{g}(.)$-geographic differential privacy is satisfied} \}$. 
The definition of $\Qrobs_{\Signals}$ are similar to that in Section \ref{sec:simmodel}.


We now show that adding Laplace noise continues to remain an optimal action for the users even under this relaxed model of geographic differential privacy. 

\begin{theorem} \label{thm:lap:sim-alt}
For $\mathfrak{g}(.)$-geographic differential privacy, adding Laplace noise, whose density function is $\rho(x)=\frac{\mathfrak{g}'(0)}{2} \cdot e^{- \mathfrak{g}'(0)\disrzero{x}} $, is one of the optimal choices of $\Probsgeog{\Signals}$ for users. Further, when $\mathfrak{h}(z)=z$, we have $\funfgeoalt{\real}{\Ads}{\mathfrak{g}(.)}{\K} {\mathfrak{h}(.)} = O\left(\frac{1}{\mathfrak{g}'(0) k}\right)$ and the optimal mechanism (choice of actions of users and server) itself can be computed in closed form.
\end{theorem}

\begin{proof}
    The proof of this theorem follows identically to that of Theorem \ref{thm:lap:sim}. However, we require a slight modification of Lemma \ref{lemma:diff-privacy-condition} to prove it as stated and proven in Lemma \ref{lemma:diff-privacy-condition_gen}.
\end{proof}

\begin{lemma}{\label{lemma:diff-privacy-condition_gen}}
    Suppose, $P_u$ has a probability density function given by $g(u,.): \mathbb{R}^k \rightarrow \mathbb{R}$ for every $u \in \mathbb{R}$. Then, $P$ satisfies $\mathfrak{g}(.)$-geographic differential privacy iff $\max(|\overline{g}_u(u,\vect{x})|, |\underline{g}_u(u,\vect{x})|) \leq \mathfrak{g}'(0) g(u,\vect{x}) \text{ } \forall u \in \mathbb{R};\forall x \in \mathbb{R}^k$ whenever $\mathfrak{g}(.)$ satisfies Assumption \ref{assn:g_defn}.   
\end{lemma}

The proof of this Lemma is very similar to that of Lemma \ref{lemma:diff-privacy-condition} and proven in Section \ref{sec:diff-privacy-condition_gen_proof}.

\section{Conclusion}
We have defined a new architecture for differential privacy with a small number of multiple selections,
and shown in a stylized model that significant improvements in the privacy-accuracy tradeoffs are indeed possible. Our model ignores some practical considerations, namely, the client's request lives in a high dimensional feature space (and not in one-dimension), and the server has a machine learning model to evaluate the quality of a result that it needs to convey to the client in some compressed form. Addressing these issues while preserving the privacy-accuracy trade-off either theoretically or empirically, will be the focus of future work.



\bibliography{refs}
\appendix

\renewcommand{\funf}[4]{f(#1,#3,#4)}

\renewcommand{\Users}{U}
\renewcommand{\Ads}{\Users}
\renewcommand{\Kads}{\mathcal{A}}
\renewcommand{\diszero}[1]{\left\|#1\right\|_2}

\renewcommand{\sphere}{\mathbb{S}^{\Dim}}

\renewcommand{\dis}[2]{\left\|#1-#2\right\|_2}

\renewcommand{\defdis}[2]{ \theta \left(#1,#2\right)}
\section{Extensions to the models and more detailed model description}

\subsection{Model and Overview (With More Details)}


\subsubsection{Notations}{\label{appendix-sec:notations}}


We use $[m]$ to denote the set $\{1,2,\dots,m\}$, for any non-negative integer $m$. $\natr$ is the set of positive integers, $\ratn$ is the set of rational real numbers, and $\real$ is the set of real numbers. We denote $\mathcal{A} = \real^k$



$\E{V}$ denotes the expectation of real random variable $V$. $\Pr{E}$ denotes the probability of an event $E$. $\I{E}$ is the indicator, which has value 1 when the event $E$ happens and has value 0 when $E$ does not happen, and thus we have $\Pr{E}=\E{\I{E}}$.

For any set $S$, $\probm{S}$ is defined as the set of probability measures on $S$.

We sometimes have integral operations. Since we are optimizing the objective over all possible mechanisms, some functions may not be continuous, and some distributions may not have density functions, so we use the Lebesgue integral.


We use $\Exx{X\sim \mu}{V(X)}$ to denote the expectation of $V(X)$ when the probability measure of $X$ is $\mu$, so
\begin{align*}
    \Exx{X\sim \mu}{V(X)} \defeq \int_{x} V(x)\mu(dx).
\end{align*}

Similarly, we use $\Prxx{X\sim \mu}{E(X)}$ to denote the probability that event $E(X)$ happens when the probability measure of $X$ is $\mu$, so
\begin{align*}
    \Prxx{X\sim \mu}{E(X)} \defeq \Exx{X\sim \mu}{\I{E(X)}} = \int_{x} \I{E(x)}\mu(dx).
\end{align*}

When a probability measure $\mu$ has a probability density (or in other words is continuous), we use $\rho_{\mu}(x)$ to denote its probability density at $x$ and the Lebesgue integral may be replaced by a Reimann integral.

Similarly, we use $\texttt{Set}$ to convert a vector in $\real^k$ to a set. More formally, for any $\vect{a} \in \real^k$, $\texttt{Set}(\vect{a})$ is denoted by $\{a_i: i \in [k]\}$.

\subsection{Problem setting for the restricted setup (results \iffalse\ak{can we not use ads here and be consistent with phrasing in the set-up?}\fi and users lie in the same space)}{\label{appendix-sec:formal-description}}


Users send a signal to the server, and the server sends back $k$ results $A$ which we denote as a vector in $\real^k$. In this subsection, we give a measure-theoretic view of the mechanism $(Z,\Pname,\Qname)$.




We aim to determine a mechanism with the following ingredients.
\begin{enumerate}
\item A set of signals $\Signals$.  
\item Action of users $\dsig{u}(\C)$, denoting the probability that user $u$ sends signal $s\in \C$, for $u\in \Users$ and $\C\in \measure{\Signals}$, where $\measure{\Signals}$ is the set of all the measurable subsets of $\Signals$. 
\item Action of server $\dsigq{s}(\C)$, denoting the probability that server sends back results $A\in \C$ when receiving signal $s$, for $s\in \Signals$ and $\C\in \measure{\Kads}$, where $\measure{\Signals}$ is all the measurable subset of $\Signals$. 
\end{enumerate}

The system should have geographic differential privacy. For any $\C\in \measure{\Signals}$, for any $u_1,u_2\in \Users$, $\dsig{u_1}(\C)\leq e^{\epsilon\cdot |u_1-u_2|} \dsig{u_2}(\C)$, 

%

We want to minimize the distance between the user and the closest ad that the user receives, in the worst case with respect to the distribution of users. Formally, we want to compute
\[  \funfgeo{(d,\Users)}{\Ads}{\epsilon}{\K}{\mathfrak{h}(.)} \defeq  \inf_{\Signals}\inf_{\Pname\in \Probs{\Signals},\Qname\in \Qrobs_{\Signals}} \sup_{u \in \real} \Exxlimits{s\sim \dsig{u}}{\Exxlimits{ \vect{a}\sim \dsigq{s}}{ \min_{a\in \texttt{Set}(\vect{a})}\mathfrak{h}(|u-a|) }}, \]
where
\begin{enumerate}
    \item $\Signals$ can be any set.
    \item  When we are interested in geographic differential privacy on $\real$. We call $\Exxlimits{s\sim \dsig{u}, \vect{a}\sim \dsigq{s}}{ \min\limits_{a\in \texttt{Set}(\vect{a})}\mathfrak{h}(|u-a|) }$ the cost of user $u$ from mechanism $(Z,\Pname,\Qname)$
    

    \item $\Probs{\Signals}$ is the set of mechanism satisfying geographic differential privacy.
    
    \item $\Qrobs_{\Signals}:=\{ \Qname | \forall s\in \Signals, \dsigq{s}({\cdot})\in \probm{\Kads} \}$, which is the set of available actions of the server.
\end{enumerate}

\zhihaohalfcomment{need assumption, e.g. $P_{x}(y)$ is measurable w.r.t. $(x,y)$. Otherwise the construction of new mechanism is not well defined}

\subsection{Summary of Results}


In either model, one of the optimal mechanisms satisfies $H=\Kads$, which means the signals sent by users and servers can be drawn from the same set. In fact, the server will directly return the signal it receives. Furthermore, we do not need to consider all the distribution over users, we only care about the user that has the largest error. 











\begin{theorem}\label{thm:main1}

For $\epsilon$-geographic differential privacy adding Laplace noise i.e. user $u$ draws a signal from $\mathcal{L}_{\epsilon}(u)$ is one of the optimal choices for users, and in this case $\funfgeo{\Users}{\Ads}{\epsilon}{\K}{\mathfrak{h}(.)}$ is $O(1/(\epsilon \K))$ when $\mathfrak{h}(t)=t$. Furthermore, when $\mathfrak{h}(t)=t$, the optimal mechanism can be computed in closed form (Theorem \ref{cor:geofinal}).

\end{theorem}

\subsection{Calculation of optimal mechanism on a ring for the case of $k=2$}    \label{appendix-sec:non-opt-ex}


We calculate the optimal mechanism in geographic differential privacy setting, on a unit ring, when $\epsilon=3/8$, and the number of results is $\K=2$. In this section, we define $d(u,a) = \disb{u}{a}$


We use real numbers in $[-\pi,\pi)$ to denote users and results on a unit ring, and $\disb{x}{a}$ denotes $|x-a|$. Figure \ref{fig:intro3} illustrates the optimal mechanism under geographic DP for $k=2$. This mechanism uses noise that is a piece-wise composition of Laplace noises; we obtain a cost of 0.72 \remove{n improvement of a factor of 0.558 with this optimal mechanism}whereas Laplace noise gives a cost of 0.75. \remove{ an improvement of only 0.585.} 

  
  \begin{figure}
      \centering

    \includegraphics[width = 80mm]{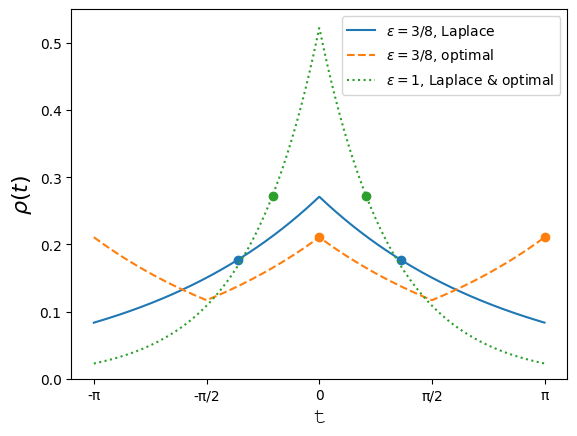}
      \caption{Geographic differential privacy setting when users and results are located on a unit ring, for $\K=2$ and $\epsilon\in \{3/8,1\}$, showing the stark difference between Laplace noise and the optimal noise. Suppose the user has a private value $u$. Then the user sends $u+x$ to the server, where $x$ is drawn from a noise distribution with density $\rho(t)$, depicted here for both Laplace noise and the optimal noise. Suppose the server receives $s$. Then the server's optimal response is $s+a_1$ and $s+a_2$, where the values of $a_1,a_2$ are the $t$-axis values of dots on the density functions, again assuming both Laplace noise and the optimal noise. Laplace is not optimal when $\epsilon=3/8$, while Laplace is optimal when $\epsilon=1$.}

      {
    \color{cyan}

      }
      
      \label{fig:intro3}
  \end{figure}

To find the optimal mechanism for the case of the ring, we solve the DILP $\mathcal{O}$ using a linear program solver and obtain the plot shown in Figure \ref{fig:intro3} with cost of 0.72. However, when the user sends Laplace noise, the server on receiving signal $z$ responds with two points $z+ a_1$ and $z+a_2$ which maybe calculated by the following problem.

\begin{align*}
    \min_{a_1\in [-\pi,\pi), a_2\in [-\pi,\pi)} \int_{-\pi}^{\pi} \min\{\disb{x}{a_1},\disb{x}{a_2}\} \rho(x) dx,
\end{align*}
where $\rho(x)$ is a density function for the Laplace distribution, 
\subsection{Restricted and Unrestricted Setup of the Multi-Selection model}{\label{appendix-sec:restricted_unrestricted_setup}} 

Recall the setup in Section \ref{sec:introduction_setup} where the users and results belonged to different sets $\real$ and $M$ with the definition of dis-utility in Definition \ref{def:disutility_user_result}. In section \ref{subsec:user_results_restricted}, we considered an alternate setup where the users and results belonged to the same set $\real$ and the optimal result for an user $u$ was the result $u$ itself. In this section, we call these setups unrestricted and restricted respectively and define our ``multi-selection'' model separately for both these setups. Finally, we bound the cost function in the unrestricted setup by the cost function in the restricted setup in Theorem \ref{thm:upper-bound-relaxation} thus, showing that it is sufficient to consider the cost function in the restricted setup.



\subsubsection{Unrestricted setup}
Recall that results and users are located in sets ${M}$ and $\real$ respectively and function $\text{OPT}: \real \rightarrow M$ maps every user to its optimal result(ad). Recall that the dis-utility of an user $u$ from a result $m$ is defined in Definition \ref{def:disutility_user_result}.

\subsubsection{Restricted setup}

This setup is very similar to the setup described except the fact that users and results(ads) lie on the same set $\real$. Recall from the description in Section \ref{subsec:user_results_restricted}, the dis-utility of an user $u \in \real$ from a result $a \in \real$ is given by $\mathfrak{h}(|u-a|)$ for some function $\mathfrak{h}(.)$ satisfying equations \eqref{item:first-propertyh} and \eqref{item:second-property-h}.

\subsubsection{The space of server/user actions}
Recall that the goal is to determine a mechanism that has the following ingredients:


\begin{enumerate}
\item A set of signals $\Signals$.
\item The action of users, which involves choosing a signal from a distribution over signals. We use $\dsig{u}$ for $u\in \real$ to denote the distribution of the signals sent by user $u$. This distribution is supported on $\Signals$.
\item The distribution over actions of the server, $\dsigq{s}$ when it receives $s\in \Signals$. This distribution denoting the distribution of the result set returned by the server given signal $s$ may be supported on either $\real^k$ or $M^k$, for the restricted setup and unrestricted setup respectively.
\end{enumerate}

The optimal mechanism is computed by jointly  optimizing over the tuple $(Z,\Pname,\Qname)$.

And thus, we define the set of server responses by $\Qrobs_{\text{unrestricted},\Signals}$ and $\Qrobs_{\text{restricted},\Signals}$ for unrestricted and restricted setup respectively.
\begin{itemize}
    \item $\Qrobs_{\text{unrestricted},\Signals}:=\{ \Qname | \forall s\in \Signals, \dsigq{s}\text{~is a distribution on $M^{\K}$} \}$. 
    \item $\Qrobs_{\text{restricted},\Signals}:=\{ \Qname | \forall s\in \Signals, \dsigq{s}\text{~is a distribution on $\real^{\K}$} \}$.
\end{itemize}

In any {\em feasible geographic DP mechanism}, the  user behavior should satisfy $\epsilon$-geographic differential privacy: for any $u_1,u_2\in \real$, it should hold that 
$P_{u_1}(S) \leq P_{u_2}(S) e^{\epsilon|u_1-u_2|}$ where $S$ is any measurable subset of $Z$.
For any fixed response size $k$, in order to maximize utility while ensuring the specified level of privacy, the goal is to minimize the disutility of the user from the result that the gives the user minimum dis-utility where the minimisation is the worst case user $u$ in ${\real}$.


\subsubsection{Cost functions in both the setups}{\label{appendix-sec:objective}} 

For the unrestricted and restricted setups, we define the cost functions $\funforiggeo{(d,U)}{\Ads}{\epsilon}{\K}{\mathfrak{h}(.)}$ and $\funforiggeo{(d,U)}{\Ads}{\epsilon}{\K}{\mathfrak{h}(.)}$ respectively below. Recall that $Z$ may denote any set.  

\[  \funforiggeo{(d,U)}{\Ads}{\epsilon}{\K}{\mathfrak{h}(.)}
\defeq \inf_{\Signals}\inf_{\substack{\Pname\in \Probsgeo{\Signals}\\ \Qname\in \Qrobs_{\text{unrestricted},\Signals}}} \sup_{u \in \real} \Exxlimits{s\sim \dsig{u}, \vect{a}\sim \dsigq{s}}{ \min_{a\in \texttt{Set}(\vect{a})}\left(\Disutility^{\mathfrak{h}(.)}(u,a)\right) }, \]

\[  \funfnewgeo{(d,U)}{\Ads}{\epsilon}{\K}{\mathfrak{h}(.)}
\defeq \inf_{\Signals}\inf_{\substack{\Pname\in \Probsgeo{\Signals}\\ \Qname\in \Qrobs_{\text{restricted}, \Signals}}} \sup_{u \in \real} \Exxlimits{s\sim \dsig{u}, \vect{a}\sim \dsigq{s}}{ \min_{a\in \texttt{Set}(\vect{a})}\mathfrak{h}\left(|u-a|\right) }, \text{ where }\]

$\Probsgeo{\Signals}:=\{ \Pname | \forall u\in \real, \dsig{u}\text{~is a distribution on $\Signals$,} 
    \text{~and $\epsilon$-geographic differential privacy is satisfied} \}$.

We state a theorem upper bounding $\funforiggeo{U}{\Ads}{\epsilon}{\K}{\mathfrak{h}(.)}$ by $\funfnewgeo{U}{\Ads}{\epsilon}{\K}{\mathfrak{h}(.)}$.

\begin{theorem}{\label{thm:upper-bound-relaxation}}
    For any $\mathfrak{h}(.)$ satisfying equation \eqref{item:first-propertyh}, we have $\funforiggeo{(d,U)}{\Ads}{\epsilon}{\K}{\mathfrak{h}(.)} \leq \funfnewgeo{(d,U)}{\Ads}{\epsilon}{\K}{\mathfrak{h}(.)}$.
\end{theorem}

\begin{proof}

    Recall that
    \[  \funforiggeo{(d,U)}{\Ads}{\epsilon}{\K}{\mathfrak{h}(.)}
    \defeq \inf_{\Signals}\inf_{\substack{\Pname\in \Probsgeo{\Signals}\\ \Qname\in \Qrobs_{\text{unrestricted},\Signals}}} \sup_{u \in \real} \Exxlimits{s\sim \dsig{u}, \vect{a}\sim \dsigq{s}}{ \min_{a\in \texttt{Set}(\vect{a})}\left(\Disutility^{\mathfrak{h}(.)}(u,a)\right) },\]
    \[  \funfgeo{(d,U)}{\Ads}{\epsilon}{\K}{\mathfrak{h}(.)}
    \defeq \inf_{\Signals}\inf_{\substack{\Pname\in \Probsgeo{\Signals}\\ \Qname\in \Qrobs_{\text{restricted},\Signals}}} \sup_{u \in \real} \Exxlimits{s\sim \dsig{u}, \vect{a}\sim \dsigq{s}}{ \min_{a\in \texttt{Set}(\vect{a})}\mathfrak{h}\left(|u-a|\right) }, \]
    and
    \[\Disutility^{\mathfrak{h}(.)}(u,m) := \inf\limits_{u': OPT(u')=m} \mathfrak{h}(|u-u'|).\]

    So we need to prove
    \begin{align*}
    & \inf_{\Signals}\inf_{\substack{\Pname\in \Probsgeo{\Signals}\\ \Qname\in \Qrobs_{\text{unrestricted},\Signals}}} \sup_{u \in \real} \Exxlimits{ s\sim \dsig{u}, \vect{a}\sim \dsigq{s}}{ \min_{a\in \texttt{Set}(\vect{a})}\left(\inf\limits_{u': OPT(u')=a} \mathfrak{h}(|u-u'|)\right) } \\
    \leq &
    \inf_{\Signals}\inf_{\substack{\Pname\in \Probsgeo{\Signals}\\ \Qname\in \Qrobs_{\text{restricted},\Signals}}} \sup_{u \in \real} \Exxlimits{ s\sim \dsig{u}, \vect{a}\sim \dsigq{s}}{ \min_{a\in \texttt{Set}(\vect{a})}\mathfrak{h}\left(|u-a|\right) }
    \end{align*}

    It is sufficient to show, for any $\Signals$ and $\Pname\in \Probsgeo{\Signals}$, 
    
    \begin{align}
    & \inf_{\Qname\in \Qrobs_{\text{unrestricted},\Signals}}  \sup_{u \in \real}\Exxlimits{s\sim \dsig{u}, \vect{a}\sim \dsigq{s}}{ \min_{a\in \texttt{Set}(\vect{a})}\left(\inf\limits_{u': OPT(u')=a} \mathfrak{h}(|u-u'|)\right) }  \nonumber \\
    \leq &
    \inf_{\Qname\in \Qrobs_{\text{restricted},\Signals}} \sup_{u \in \real} \Exxlimits{ s\sim \dsig{u}, \vect{a}\sim \dsigq{s}}{ \min_{a\in \texttt{Set}(\vect{a})}\mathfrak{h}\left(|u-a|\right) }  \label{sc:eq1} 
    \end{align}

    For $\Qname\in \Qrobs_{\text{restricted},\Signals}$, since $\Disutility^{\mathfrak{h}(.)}(u,m) := \inf\limits_{u': OPT(u')=m} \mathfrak{h}(|u-u'|)$, we have
    \begin{align}\label{sc:eq2}
     & \Exxlimits{s\sim \dsig{u}, \vect{a}\sim \dsigq{s}}{ \min_{a\in \texttt{Set}(\vect{a})}\mathfrak{h}\left(|u-a|\right) }   \geq  
    \Exxlimits{ s\sim \dsig{u}, \vect{a}\sim \dsigq{s}}{ \min_{a\in \texttt{Set}(\vect{a})}\Disutility^{\mathfrak{h}(.)}(u,\text{OPT}(a)) }.\nonumber\\
    &\hspace{-3 em}\implies \sup_{u \in \real} \Exxlimits{s\sim \dsig{u}, \vect{a}\sim \dsigq{s}}{ \min_{a\in \texttt{Set}(\vect{a})}\mathfrak{h}\left(|u-a|\right) }   \geq  \sup_{u \in \real} \Exxlimits{s \sim \dsig{u}, \vect{a}\sim \dsigq{s}}{ \min_{a\in \texttt{Set}(\vect{a})}\Disutility^{\mathfrak{h}(.)}(u,\text{OPT}(a))}.
    \end{align}

    Given $\Qname\in \Qrobs_{\text{restricted},\Signals}$, we draw $\vect{a}$ from $\Qname$, and let $\vect{b}=[\text{OPT}(a_1), \text{OPT}(a_2), \ldots, \text{OPT}(a_k)]^T$. 
    Suppose the distribution of $\vect{b}$ is $\Qname'$, and we have
    \begin{align}\label{sc:eq3}
    \Exxlimits{s\sim \dsig{u}, \vect{a}\sim \dsigq{s}}{ \min_{a\in \texttt{Set}(\vect{a})}\Disutility^{\mathfrak{h}(.)}(u,\text{OPT}(a)) }
    =
    \Exxlimits{s\sim \dsig{u}, \vect{b} \sim \dsigqprime{s}}{ \min_{\vect{b} \in \texttt{Set}(\vect{b})}\Disutility^{\mathfrak{h}(.)}(u,b) }.
    \end{align}

    Note that $\Qname' \in \Qrobs_{\text{unrestricted},\Signals}$, so combining Inequality \ref{sc:eq2} and Equality \ref{sc:eq3}, we have Equation \ref{sc:eq1}, which finishes the proof.

\end{proof}

And thus, it is sufficient to study $\funfnewgeo{(d,U)}{\Ads}{\epsilon}{\K}{\mathfrak{h}(.)}$ (defined as $\funfgeo{(d,U)}{\Ads}{\epsilon}{\K}{\mathfrak{h}(.)}$ in Section \ref{sec:simmodel}). 
\subsection{Noise satisfying Monotone Hazard Rate property}  \label{sec-appendix:mhr}


Let $Y$ denote the random noise with density $g$. We assume $Y$ is symmetric about the origin, and let $X = Y^{+} \cup \{0\}$. Let $f$ denote the density function of $X$ (so that $f(x) = 2 g(x)$ for $x \ge 0$), and let $F(x) = \Pr{X \ge x}$. We assume that $\E{X} = \frac{1}{\epsilon}$.\footnote{This implies that a large $\epsilon$ is equivalent to the magnitude of noise being smaller and vice-versa. Although this distribution does not satisfy $\epsilon$ geographic differential privacy, this follows a similar trend w.r.t $\epsilon$. } We assume $f$ is continuously differentiable and log-concave. By~\cite{BagnoliB}, we have $F$ is also log-concave.  Note that several natural distributions such as Exponential (Laplace noise) and Gaussian are log-concave. 

We are interested in choosing $2K-1$ values $S = \{-a_{K-1}, -a_{K-2}, \ldots, - a_1, 0, a_1, \ldots, a_{K-1}\}$ such that for a random draw $y \sim Y$, the expected error in approximating $y$ by its closest point in $S$ is small. For $i = 0,1,2,\ldots,K-1$, we will choose $a_i = F^{-1} \left(1 - \frac{i}{K}\right)$. Let $\phi = \Exx{x \sim X}{\min\limits_{v \in S} |v - x|}$. Note that the error of $Y$ with respect to $S$ is exactly $\phi$.

Our main result is the following theorem:
\begin{theorem}{\label{theorem:mhr}}
$ \phi = O\left(\frac{\log K}{K\epsilon}\right)$.
\end{theorem}
\begin{proof}
Let $G(z) = F^{-1}(z)$ for $z \in [0,1]$. For upper bounding $\phi$, we map each $x \sim X$ to the immediately smaller value in $S$. If we draw $z \in [0,1]$ uniformly at random, the error is upper bounded as:
$$ \phi \le \int_0^1 G(z) dz  - \sum_{i=1}^K \frac{1}{K} G\left(\frac{i}{K}\right) \le \int_0^{1/K} \left( G(z) - G(1/K) \right) dz + \frac{1}{K} G(1/K).$$
Let $q = G(1/K).$ Then the above can be rewritten as:
$$ \phi \le \frac{q}{K} + \int_q^{\infty} F(x) dx. $$ 
Next, it follows from~\cite{BagnoliB} that if $F$ is log-concave, then so is $\int_r^{\infty} F(x) dx$. This means the function 
$$ \ell(r) = \frac{\int_r^{\infty} F(x) dx}{F(r)} $$
is non-increasing in $r$. Therefore, 
$$ \int_q^{\infty} F(x) dx \le F(q) \int_0^{\infty} F(x) dx = \frac{1}{K} \E{X} = \frac{1}{\epsilon K}.$$

Let $h(x) = -\log F(x)$. Then, $h$ is convex and increasing. Further, $h(q) = \log K$. Let $s = F^{-1}(1/e)$ so that $h(s) = 1$. Since
$$ h(q) - h(s) \ge (q-s) h'(s),$$ 
we have 
$$ q - s \le \frac{\log K}{h'(s)}.$$
Further, $h(s) \le s h'(s)$ so that 
$$ \frac{1}{h'(s)} \le s.$$
Since $F(s) = 1/e$ and $\E{X} = \frac{1}{\epsilon} \ge \int_0^s F(x) dx$, we have $s \le \frac{e}{\epsilon}$. Therefore, $h'(s) \ge 1/{e\epsilon}$. Putting this together,
$$ q \le \frac{s}{\epsilon} + \frac{\log K}{h'(s)} \le \frac{e}{\epsilon}(1+ \log K) = O(\frac{\log K}{\epsilon}).$$

Therefore,
$$ \phi \le \frac{q}{K} + \int_q^{\infty} F(x) dx = O\left(\frac{\log K}{K\epsilon}\right) + \frac{1}{K\epsilon}$$
completing the proof.
\end{proof}

\subsection{User Features lying on some high dimensional space}{\label{sec-appendix:high-dim-space}}

In this subsection, we relax the assumption that users lying on an infinite dimensional line, rather for every user $u$ belongs to some set $U$ and every movie $m$ belongs to some set $M$ and consider some function $\kappa: U \rightarrow \mathbb{R}^d$ mapping every user to some user feature vector and denote the dis-utility of an user $u$ from a movie $m$ by $d(\kappa(u),m)$. Observe that this could be any complex function in high dimensions which could be captured by some machine learning model.
We now make the following assumption for movies lying in some vicinity of user $u$. Informally, this means that the the user does not get dis-utility from every possible feature but only gets a dis-utility only from some $d'$ directions. 

\begin{assumption}{\label{ass:high_dimension}}
    Fix any user $u \in U$. There exists a subset of movies $M_u \subseteq M$ (call it movies in neighbourhood of user $u$), a matrix $P_u \in \mathbb{R}^{d' \times d}$  function $\lambda_u: M_u \rightarrow \mathbb{R}^{d'}$ satisfying the following properties
    \begin{itemize}
        \item The disutility of users $u'\in U$ lying in a neighbourhood of $u$ i.e. $|\kappa(u')-\kappa(u)|_1<\delta$ from movies $m \in M_u$ may be approximated by $d(u', m) \approx ||P_u\cdot \kappa(u') - \lambda_u(m)||_1$.
        \item The function $\lambda_u$ is invertible i.e. the set of movies $M_u$ is densely populated and movies exist along most directions. 
    \end{itemize}
\end{assumption}

These assumptions imply that the functions described above hold not just for user $u$ but is also true for users $u'$ with $\kappa(u')$ lying in a neighbourhood of $\kappa(u)$. The function $\lambda_u$ denotes a map from the movie space to the feature space for an user $u \in U$. $P_u$ effectively captures the linear combination of the features that play a role in the dis-utility of user $u$.


Note that some of these assumptions a bit tight and may not be true in reality but in this section we analyse the server response and user action under this mechanism under the Gaussian noise addition mechanism. We analyse the Gaussian mechanism as it is easier to analyse since additions of Gaussian is also a Gaussian. Also we make an assumption that noise parameter $1/ \epsilon << \delta$.

\textbf{Notations:} For some $Y \sim \mathcal{N}(0,\sigma^2)$ and suppose $X= |Y|$. Now denote $F_{\sigma}(x) = \text{Pr}(X \geq x)$ for some $x>0$ and observe that $F_{\sigma}$ is invertible from $(0,1] \rightarrow [0,\infty)$. We also extend its definition to $(0,2)$ by defining $F^{-1}_{\sigma}(2-x) = - F^{-1}_{\sigma}(x)$ for $x>0$

Let us denote the sum of all squares of all entries in row $i$ of matrix $P$ is given by $P^{(2)}[i,:]$ and denote all integers from $a$ to $b$ by $[a \ldots b]$ and $[a \ldots b]^k$ denotes the set of all $k$ dimensional vectors in with each component taking an integral value from $a$ to $b$.



\textbf{User's Action:} 

An user $u \in U$ adds Gaussian noise with parameter $\epsilon$ i.e. $\mathcal{N}(0,\frac{1}{\epsilon^2} {I})$ to its feature $\kappa(u)$ and sends it to the server, let us call it distribution $\mathcal{F}_u$.

\textbf{Server's Action:} 

Suppose the server receives a signal $\hat{f}$. The signal sends back $k= (2k'+1)^{d'}$ along each of the the $d'$ dimensions which we describe below. Further consider some user $u'$ whose feature vector $\kappa(u')$ lies in $1/ \epsilon$ vicinity of signal $\hat{f}$. Let us index the movies by $\mathcal{M}_{\hat{f}} = [\mathfrak{m}_\vect{i}]_{\vect{i} \in {[-k' \ldots k']}^{d'}}$.

For every $\vect{i} \in {[-k' \ldots k']}^{d'}$, define $\mathfrak{m}_\vect{i}$ as follows:

\begin{align}
    m_{\vect{i}} := \lambda_{u'}^{-1}(P_{u'} \hat{f} + \vect{v}) \text{ for }& \vect{v} \in \mathbbm{R}^{d'} \\
    & \text{ where } v_j := F^{-1}_{ \frac{P^{(2)}_{u'}[j,:]}{\epsilon^2} }(1-i_j/k') \text{ } \forall j \in [1,\ldots,d']
\end{align}


\begin{lemma}
    $\mathbbm{E}_{\hat{f} \sim \mathcal{F}_u} \left[\min\limits_{m \in \mathcal{M}_{\hat{f}}} d(u,m)\right] \overset{(a)}{\approx} O\left(d' \frac{\log k'}{\epsilon k'}\right) = O\left(\frac{\log k}{\epsilon k^{1/d'}}\right)$
\end{lemma}

\begin{proof}
    This follows since $f'$ is a Gaussian Random vector and thus $j^{th}$ component of $P_{u'} \hat{f}$ is given by Gaussian random variable with variance $\frac{P^{(2)}_{u'}[j,:]}{\epsilon^2}$. The first point in Assumption \ref{ass:high_dimension} implies the disutility may be approximated by $||P_{u'}\cdot \kappa(u) - \lambda_{u'}(m)||_1$ since the 

    Now apply Theorem \ref{theorem:mhr} and since the placement of points is identical to that in Section \ref{sec-appendix:mhr}, we get the desired equality in $(a)$.
\end{proof}

\subsection{Additional problems that may be be formulated as a DILP}{\label{sec-appendix:weak_duality_application}}
 We now discuss two scenarios where a problem maybe solved by formulation as an optimization problem over function spaces involving constraints on both derivatives and integrals.

    \subsubsection{A problem from job scheduling}

    Job scheduling has been a well-studied problem and has attracted a lot of attention in theory of operations research starting from \cite{conway1967theory,baker1974introduction}. In \cite{anderson1981new}, a new continuous time model was proposed for the same described formally in \cite[Section 2]{anderson1981new} and the optimization variable functions are given by $\{x_i(.)\}_{i=1}^{N}: \mathbb{R} \rightarrow \mathbb{R}^{+}, \{y_{i,j}(.),u_{i,j}(.)\}_{j=1,i=1}^{N_i,N}: \mathbb{R} \rightarrow \mathbb{R}^{+}$. In this model, we need to produce multiple items with each item goes through multiple job shops for an operation to be performed. and there is a continuous demand for each item. We aim to minimise the cumulative backlog for items in continuous time.

    We now introduce some notations for the same by restating the problem in \cite{anderson1981new}.

    \begin{itemize}
        \item $N$ : number of parts
        \item $n_i$: number of operations required by part $i$ 
        \item $M$: number of machines
        \item $m_{i,j}$: machine on which the $j^{th}$ operation is performed for part $i$
        \item $r_{i,j}$: Maximum rate of procession on $m_{i,j}$ of the $j^{th}$ operation on part $i$
        \item $s_{i,j}$ Initial stock of part $i$ between the $j^{th}$ and $(j+1)^{th}$ operation.
        \item $b_i$: Initial backlog of part $i$
        \item $d_i(t)$: Rate of demand of part $i$ at time $t$
        \item $k_i$: cost per unit time of unit backlog of part $i$
        \item $T$ time horizon
    \end{itemize}

    Observe that the optimization variables in this problem are $x_i(t):[0,1] \rightarrow \mathbb{R}^{+}$ for $i \in [n]$ and $y_{i,j}(t):[0,1] \rightarrow \mathbb{R}^{+}$ for $j \in [n_i], i \in [n]$

    As discussed in \cite{anderson1981new}, we now formulate our optimization problem below.

    \begin{equation}{\label{eqn:optim_formulation}}
   \left\{
    \begin{aligned}
    \min & \int_{t \in 0}^{T} \sum_{i=1}^{n} k_i(t)x_i(t) du\\
    \textrm {s.t. } \quad & x'_i(t) = d_i(t) - r_{i,n_i}u_{i,n_i}(t) \text{ }\forall j \in [n_i]; i \in [n];\\
     & y'_{i,j}(t) = r_{i,j}u_{i,j}(t) - r_{i,j+1} u_{i,j+1}(t) \text{ }\forall j \in [n_i]; i \in [n];\\
    & x_i(t) \geq 0; y_{i,j}(t) \geq 0;  u_{i,j}(t) \geq 0\text{ } \forall j \in [n_i]; i \in [n]; t \in [0,T]\\
    & \sum_{i,j: m_{i,j}=k} u_{i,j}(t) \leq 1 \text{ } \forall k \in [M] \\
    & x_i(0)=b_i ; y_{i,j}(0) = s_{i,j} ; \forall j \in [n_i]; i \in [n];
    \end{aligned}
    \right.
    \end{equation}

    Observe that since it involves constraints on both the derivatives and integrals, this maybe modelled as a differential integral linear program (DILP) and our weak duality framework may be a useful tool in analysing the same. 


    \subsubsection{A problem from control theory}

    In control theory \cite{evans1983introduction}, we often have a state $\vect{x}(t)$ whose evolution is given by a differential equation. We redefine the following equation from \cite[Chapter 4]{evans1983introduction} below to characterise its evolution.

    \begin{equation}{\label{eqn:ode_equation}}
            \left\{
            \begin{aligned}
                & \vect{x}(t) = \vect{f}(\vect{x}(t),{\alpha}(t))\\
                & \vect{x}(0) = x^{0}
            \end{aligned}
            \right.
    \end{equation}

    Here $x^0 \in \mathbb{R}^n$, $A \in \mathbb{R}^m$, $\vect{f} : \mathbb{R}^n \times A \rightarrow \mathbb{R}, \alpha : [0, \infty) \rightarrow A$ is the control, and $\vect{x} : [0, \infty) \rightarrow \mathbb{R}^n$ is the response of the system. Denote the set of admissible controls by $\mathcal{A} = \{\alpha(.) : [0, \infty) \rightarrow A | \alpha(.) \text{ is measurable}\}$ 

    We also introduce the payoff functional $(P)$ 
    $$P[\alpha(.)] =  \int_{t=0}^{T} r(x(t), \alpha(t)) dt + g(x(T))$$,

where the terminal time $T > 0$, running payoff $r : \mathbb{R}^n \times A \rightarrow \mathbb{R}$ and terminal payoff $g : \mathbb{R}^n \rightarrow \mathbb{R}$ are given. Observe that if the function $f(.)$ is linear and the running payoff is linear, we may represent it as a differential integral linear program.

It may also be typical to have constraints on the the control input $\alpha$, for example there could be some constraints on its rate of growth or some constraints on its Lipschitz continuity.

\section{Proofs of skipped lemmas and theorems}

\subsection{Simulating the server: Detailed proof of Theorem \ref{thm:singpoint:sim}}\label{sec-appendix:simproof}



We first show it is sufficient to consider mechanisms in which servers directly return the received signal from the user, thus removing the server from the picture.

Recall that $\Probs{\Signals}$ is the set of differential private mechanisms that adopt signal set $\Signals$. Then $\Probs{\Kads}$ is the set of differential private mechanisms in which users pick signals from $\Kads= \real^k$.


\begin{theorem*}[Restatement of Theorem \ref{thm:singpoint:sim}]   

We have
\begin{align*}
\funfgeo{(d,\Users)}{\Ads}{\epsilon}{\K}{\mathfrak{h}(.)} =\inf_{\Pname\in \Probs{\real^{\K}}} \sup_{u \in \real} \Exx{\vect{a} \sim \dsig{u}}{ \min_{a\in \texttt{Set} (\vect{a})}\mathfrak{h}(|u-a|) }.
\end{align*}
\end{theorem*}

\begin{proof}

Fix $\Signals,\Pname\in \Probs{\Signals},\Qname\in \Qrobs_{\Signals}$ and recall that that $\mathcal{A} = \real^k$ from Appendix \ref{appendix-sec:notations}. For any $\C\in \measure{\Kads}$, let
\begin{align*}
    \Ppxx{u}{\C}=\Prxx{ s\sim \dsig{u}, \vect{a}\sim \dsigq{s}}{\vect{a}\in \C},
\end{align*}
which is the probability that user $u$ receives a result set in $\C$.

Then for any $u_1,u_2\in \real, \C\in \measure{\Kads}$, we have
\begin{align*}
    \Ppxx{u_1}{\C}
    =\Prxx{ s\sim \dsig{u_1}, \vect{a}\sim \dsigq{s}}{\vect{a}\in \C}  =&\Exx{s\sim \dsig{u_1}}{\Prxx{ \vect{a}\sim \dsigq{s}}{\vect{a}\in \C}}    \\
    =& \int_{s\in \Signals}{\left(\Prxx{ \vect{a}\sim \dsigq{s}}{\vect{a}\in \C}\cdot \dsig{u_1}(ds) \right)}    \\
    \leq & e^{\epsilon\cdot|u_1-u_2|}\int_{s\in \Signals}{\left(\Prxx{ \vect{a}\sim \dsigq{s}}{\vect{a}\in \C}\cdot \dsig{u_2}(ds) \right)}    \\
    = & e^{\epsilon\cdot|u_1-u_2|}\Exx{s\sim \dsig{u_2}}{\Prxx{ \vect{a}\sim \dsigq{s}}{\vect{a}\in \C} }    \\
    = & e^{\epsilon\cdot |u_1-u_2|}\Exx{s\sim \dsig{u_2},\vect{a}\sim \dsigq{s}}{\I{\vect{a}\in \C} }    \\
    = & e^{\epsilon\cdot |u_1-u_2|} \Ppxx{u_2}{\C},
\end{align*}
so $\Ppname\in \Probs{\Kads}$ because $\Ppxx{u}{\cdot}\in \probm{\Kads}$ for any $u\in \real$.

At the same time,
\begin{align*}
\Exx{s\sim \dsig{u}, \vect{a}\sim \dsigq{s}}{ \min_{a\in \texttt{Set}(\vect{a})}\mathfrak{h}(|u-a|) }
=\Exx{\vect{a}\sim \dsighat{u}}{ \min_{a\in \texttt{Set}(\vect{a})}\mathfrak{h}(|u-a|) } \text { 
 so we have},
\end{align*}
\begin{align*}
\funfgeo{(d,\Users)}{\Ads}{\epsilon}{\K}{\mathfrak{h}(.)} = \inf_{\Pname\in \Probs{\Kads}} \sup_{u \in {\real}} \Exx{\vect{a}\sim \dsig{u}}{ \min_{a\in \texttt{Set}(\vect{a})}\mathfrak{h}(|u-a|) }.
\end{align*}

\end{proof}







It is sufficient to assume the server directly sends back the received signal, and we care about the user with largest cost

\subsection{Dual and primal DILPs: Proof of Lemmas \ref{lemma:diff-privacy-condition}, Lemma \ref{lemma:diff-privacy-condition_gen} and \ref{lemma:arbitrary_aprroximation}}{\label{appendix-sec:primal-feasible}}
\subsubsection{Proof of Lemma \ref{lemma:diff-privacy-condition}}{\label{sec:diff-privacy-condition-proof}}

\begin{lemma*}
    Suppose, $P_u$ has a probability density function given by $g(u,.): \mathbb{R}^k \rightarrow \mathbb{R}$ for every $u \in \mathbb{R}$. Then, $P$ satisfies $\epsilon$-geographic differential privacy iff $\max(|\overline{g}_u(u,\vect{x})|, |\underline{g}_u(u,\vect{x})|) \leq \epsilon g(u,\vect{x}) \text{ } \forall u \in \mathbb{R};\forall x \in \mathbb{R}^k$ \footnote{$\underline{g}_u(u,\vect{x})$, $\overline{g}_u(u,\vect{x})$ denote the lower and upper partial derivative w.r.t. $u$}   
\end{lemma*} 

\begin{proof}

\textbf{We first prove the only if statement.}

Here assuming $\epsilon$- geographic differential privacy, we prove the desired constraint on $g(u,.)$.

    Observe that since, $\{P_u\}_{u \in \mathbb{R}}$ satisfies $\epsilon$ - geographic-differential privacy, we must have $g(u_1, \vect{x}) \leq e^{\epsilon|u_1-u_2|} g(u_2, \vect{x})$.

  Thus, applying $\log$ on both sides we get [assuming $g(u,x)\neq 0$], 

  \begin{align*}
      & \log g(u + \delta u, \vect{x}) - \log g(u, \vect{x}) \leq \epsilon|\delta u|\\
      \implies & \frac{\max(|\overline{g}_u(u,\vect{x})|, |\underline{g}_u(u,\vect{x})|)} {g(u,\vect{x})} \leq \epsilon
  \end{align*}
The last implication follows on limiting $\delta u$ towards zero and applying the lower and upper limits respectively and thus we prove the only if statement.

However, if $g(u,x) = 0$, we get the following 

\begin{align*}
      & g(u + \delta u, \vect{x}) - g(u,x)\leq 0\\
      \implies & {\max(|\overline{g}_u(u,\vect{x})|, |\underline{g}_u(u,\vect{x})|)} \leq 0
\end{align*}

The last inequality follows on applying $\delta u$ tending to 0 and applying upper and lower limits respectively.

\textbf{We now prove the if statement.}

Here assuming, $\max(|\overline{g}_u(u,\vect{x})|, |\underline{g}_u(u,\vect{x})|) \leq \epsilon g(u,\vect{x})$ we prove $\epsilon$-geographic differential privacy.

We first show that $g(u,\vect{x})$ is continuous in $u$. Observe that $\liminf\limits_{\delta u \to 0} g(u+\delta u) - g(u) = \lim\limits_{\delta u \to 0} \delta u\times \liminf\limits_{\delta u \to 0} \frac{ g(u+\delta u) - g(u)}{\delta u} = 0$ as $|\underline{g}_u (u,\vect{x})|= |\liminf\limits_{\delta u \to 0} \frac{ g(u+\delta u) - g(u)}{\delta u}| \leq \epsilon g(u,\vect{x})$. Similarly, using the bound on upper derivative, we show that $\limsup\limits_{\delta u \to 0} g(u+\delta u) - g(u) = 0$ and thus, we have continuity of $g(u,\vect{x})$ at $u$ for every $u \in \mathbb{R}^k$.

We first prove assuming $g(u,x) \neq 0 \forall u \in \mathbbm{R}$.
Now consider $u_1 >u_2$ for some $u_1,u_2 \in \mathbb{R}$.\\

\begin{align*}
    \log g(u_1,\vect{x}) - \log g(u_2, \vect{x}) \leq \int_{u = u_2}^{u_1} \frac{\overline{g}_u(u,\vect{x})}{g(u,\vect{x})} \leq \epsilon (u_1 -u_2) 
\end{align*}

The first inequality follows since we use an upper derivative.

\begin{align*}
    \log g(u_1,\vect{x}) - \log g(u_2, \vect{x}) \geq \int_{u = u_2}^{u_1} \frac{\underline{g}_u(u,\vect{x})}{g(u,\vect{x})} \geq -\epsilon (u_1 -u_2) 
\end{align*}

Thus, we get $g(u_1, \vect{x}) \leq e^{\epsilon|u_1-u_2|} g(u_2,\vect{x})$ on taking exponent on both sides 


We now consider the other case that $g(u,\vect{x})=0$
for some $u \in \mathbbm{R}$. We now show that this implies $g(u,\vect{x})=0 \text{ } \forall u \in \mathbbm{R}$. Suppose not and there exists some $u_1$ s.t. $g(u_1,\vect{x})>0$ and consider the largest $u<u_1$ s.t. $g(u,\vect{x})=0$ (call it $u_0$)

Now, observe the following observe that we define the limits over extended reals i.e. $\mathbbm{R} \cup \{-\infty,\infty\}$. Observe that\\

\begin{align*}
    \lim\limits_{u \to u_1} g(u,\vect{x}) & = \lim\limits_{u \to u_0} g(u,\vect{x}) + \lim\limits_{\substack{d^l \to u_0\\d^u \to u_1}} (\log g(d^u,\vect{x}) - \log g(d^l,\vect{x}))\\
    & \leq \lim\limits_{u \to u_0} g(u,\vect{x}) + \int_{u = u_0}^{u_1} \frac{\overline{g}_u(u,\vect{x})}{g(u,\vect{x})}\\
    & {\leq} \lim\limits_{u \to u_0} g(u,\vect{x}) + \epsilon (u_1-u_0)\\
    & \overset{(a)}{=} -\infty + \epsilon(u_1-u_0) = -\infty
\end{align*}

$(a)$ follows from the fact that $g(u,\vect{x})$ is continuous in $u$ and $g(u_0,\vect{x})=0$.

Since, $g(u,\vect{x})$ is a continuous function, this implies that $g(u_1,\vect{x})=0$. Since we can make this argument for every $u_1 \in \mathbbm{R}$, we argue that $g(u,\vect{x})=0 \text{ }\forall u \in \mathbbm{R}$ and thus satisfy the condition of $\epsilon$-geographic DP.



\end{proof}

\subsubsection{Proof of Lemma \ref{lemma:diff-privacy-condition_gen}}{\label{sec:diff-privacy-condition_gen_proof}}

\begin{lemma*}
    Suppose, $P_u$ has a probability density function given by $g(u,.): \mathbb{R}^k \rightarrow \mathbb{R}$ for every $u \in \mathbb{R}$. Then, $P$ satisfies $\mathfrak{g}(.)$-geographic differential privacy iff $\max(|\overline{g}_u(u,\vect{x})|, |\underline{g}_u(u,\vect{x})|) \leq \mathfrak{g}'(0) g(u,\vect{x}) \text{ } \forall u \in \mathbb{R};\forall x \in \mathbb{R}^k$ \footnote{$\underline{g}_u(u,\vect{x})$, $\overline{g}_u(u,\vect{x})$ denote the lower and upper partial derivative w.r.t. $u$} whenever $\mathfrak{g}(.)$ satisfies Assumption \ref{assn:g_defn}.   
\end{lemma*}


\begin{proof}

\textbf{We first prove the only if statement.}

Assuming $\mathfrak{g}(.)$ geographic differential privacy we prove the desired constraint on $g(u,.)$. 

    Observe that since, $\{P_u\}_{u \in \mathbb{R}}$ satisfies $\mathfrak{g}(.)$ geographic-differential privacy, we must have $g(u_1, \vect{x}) \leq e^{\mathfrak{g}(|u_1-u_2|)} g(u_2, \vect{x})$.

  Thus, applying $\log$ on both sides we get, 

  \begin{align*}
      & \log g(u + \delta u, \vect{x}) - \log g(u, \vect{x}) \leq \mathfrak{g}(|\delta u|)\\
      \implies & \frac{\max(|\overline{g}_u(u,\vect{x})|, |\underline{g}_u(u,\vect{x})|)} {g(u,\vect{x})} \leq \mathfrak{g}'(0)
  \end{align*}
The last implication follows on limiting $\delta u$ towards zero and $\mathfrak{g}(0)=0$ thus, applying the lower and upper limits respectively and thus we prove the only if statement.

However, if $g(u,x) = 0$, we get the following 

\begin{align*}
      & g(u + \delta u, \vect{x}) - g(u,x)\leq 0\\
      \implies & {\max(|\overline{g}_u(u,\vect{x})|, |\underline{g}_u(u,\vect{x})|)} \leq 0
\end{align*}

The last inequality follows on applying $\delta u$ tending to 0 and applying upper and lower limits respectively.

\textbf{We now prove the if statement.}

Here assuming $\max(|\overline{g}_u(u,\vect{x})|, |\underline{g}_u(u,\vect{x})|) \leq \mathfrak{g}'(0) g(u,\vect{x})$, we prove $\mathfrak{g}(.)$-geographic differential privacy.

We first show that $g(u,\vect{x})$ is continuous in $u$. Observe that $\liminf\limits_{\delta u \to 0} g(u+\delta u) - g(u) = \lim\limits_{\delta u \to 0} \delta u\times \liminf\limits_{\delta u \to 0} \frac{ g(u+\delta u) - g(u)}{\delta u} = 0$ as $|\underline{g}_u (u,\vect{x})|= |\liminf\limits_{\delta u \to 0} \frac{ g(u+\delta u) - g(u)}{\delta u}| \leq \epsilon g(u,\vect{x})$. Similarly, using the bound on upper derivative, we show that $\limsup\limits_{\delta u \to 0} g(u+\delta u) - g(u) = 0$ and thus, we have continuity of $g(u,\vect{x})$ at $u$ for every $u \in \mathbb{R}^k$.

We first prove assuming $g(u,x) \neq 0 \forall u \in \mathbbm{R}$.
Now consider $u_1 >u_2$ for some $u_1,u_2 \in \mathbb{R}$.\\

\begin{align*}
    \log g(u_1,\vect{x}) - \log g(u_2, \vect{x}) \overset{(a)}{\leq} \int_{u = u_2}^{u_1} \frac{\overline{g}_u(u,\vect{x})}{g(u,\vect{x})} \leq & \mathfrak{g}'(0) (u_1 -u_2) \\
    & \overset{(b)}{\leq} \mathfrak{g}(u_1-u_2)
\end{align*}

The first inequality $(a)$ follows since we use an upper derivative and the inequality $(b)$ follows from the fact that $\mathfrak{g}(.)$ is a convex function and $\mathfrak{g}(0)=0$ as defined in Section~\ref{sec:diff-privacy-gen}.

\begin{align*}
    \log g(u_1,\vect{x}) - \log g(u_2, \vect{x}) \geq \int_{u = u_2}^{u_1} \frac{\underline{g}_u(u,\vect{x})}{g(u,\vect{x})} \geq & -\mathfrak{g}'(0) (u_1 -u_2) \\
    & \overset{(c)}{\geq} - \mathfrak{g}(u_1-u_2)
\end{align*}

Inequality $(c)$ follows from the fact that $\mathfrak{g}(.)$ is a convex function. Thus, we get $g(u_1, \vect{x}) \leq e^{-\mathfrak{g}(|u_1-u_2|)} g(u_2,\vect{x})$ on taking exponent on both sides.
 
We now consider the other case that $g(u,\vect{x})=0$
for some $u \in \mathbbm{R}$. We now show that this implies $g(u,\vect{x})=0 \text{ } \forall u \in \mathbbm{R}$. Suppose not and there exists some $u_1$ s.t. $g(u_1,\vect{x})>0$ and consider the largest $u<u_1$ s.t. $g(u,\vect{x})=0$ (call it $u_0$)

Now, observe the following observe that we define the limits over extended reals i.e. $\mathbbm{R} \cup \{-\infty,\infty\}$. Observe that\\

\begin{align*}
    \lim\limits_{u \to u_1} g(u,\vect{x}) & = \lim\limits_{u \to u_0} g(u,\vect{x}) + \lim\limits_{\substack{d^l \to u_0\\d^u \to u_1}} (\log g(d^u,\vect{x}) - \log g(d^l,\vect{x}))\\
    & \leq \lim\limits_{u \to u_0} g(u,\vect{x}) + \int_{u = u_0}^{u_1} \frac{\overline{g}_u(u,\vect{x})}{g(u,\vect{x})}\\
    & {\leq} \lim\limits_{u \to u_0} g(u,\vect{x}) + \epsilon (u_1-u_0)\\
    & \overset{(a)}{=} -\infty + \epsilon(u_1-u_0) = -\infty
\end{align*}

$(a)$ follows from the fact that $g(u,\vect{x})$ is continuous in $u$ and $g(u_0,\vect{x})=0$.

Since, $g(u,\vect{x})$ is a continuous function, this implies that $g(u_1,\vect{x})=0$. Since we can make this argument for every $u_1 \in \mathbbm{R}$, we argue that $g(u,\vect{x})=0 \text{ }\forall u \in \mathbbm{R}$ and thus satisfy the condition of $\mathfrak{g}(.)$-geographic DP.

\end{proof}

\subsubsection{Proof of Lemma \ref{lemma:arbitrary_aprroximation}}{\label{sec:proof_arbitrary_approximation}}

To prove this lemma, we first show another lemma showing a necessary and sufficient condition for $\epsilon$-geographic differential privacy of distributions $\{P_u\}_{u \in \mathbb{R}}$ assuming that it has a probability density function.

We now restate and prove Lemma \ref{lemma:lp_formulation}
\begin{lemma*}
    Given any mechanism $\textbf{P} \in \mathcal{P}^{(\epsilon)}_{\textbf{R}^k}$ (satisfying $\epsilon$-geographic differential privacy), we can construct a sequence of mechanisms $\textbf{P}^{(\eta)} \in \mathcal{P}^{(\epsilon)}_{\textbf{R}^k}$ with a bounded probability distribution function for every $u \in \mathbbm{R}$ such that $\textbf{P}^{(\eta)}$ is an arbitrary cost approximation of mechanism $\textbf{P} \in \mathcal{P}^{(\epsilon)}_{\textbf{R}^k}$.
    \end{lemma*}


 \begin{proof}

For any mechanism $\textbf{P}$ belonging to the set $\mathcal{P}^{(\epsilon)}_{\mathbb{R}^k}$, we can construct alternative mechanisms $\textbf{P}^{(\eta)}$ within the same set, aiming to provide an arbitrary cost approximation of the original mechanism.
 
 For this, we borrow ideas from theory of mollifier in \cite{99c43556-16f5-3ed5-90f5-6c2bbf257ea3}. We now give a construction ${P}^{(\eta)}_u$ with a probability density function $\rho_V(u,.): \mathbb{R}^k \rightarrow \mathbb{R}$ for every $u \in \mathbb{R}$.

For $\eta >0$ and $\vect{a}\in R^k$, define the $\eta$-ball of $\vect{a}$, denoted by $\Aball{\vect{a}}{\eta}$, to be
    
    \begin{align*}
        \Aball{A}{\eta}\defeq \left\{ \vect{a}'\in R^k: |\vect{a}-\vect{a}'\|_{\infty}<\eta \right\},
    \end{align*}
    and use $\Aballsize{\vect{a}}{\eta}$ to denote the Lebesgue measure (corresponding to the volume) of $\Aball{A}{\eta}$.  Note that there exists a real number, denoted by $\ballsize{\eta}\in \mathbb{R}$, satisfying $\Aballsize{\vect{a}}{\eta}=\ballsize{\eta}$ for any $\vect{a}\in \mathbb{R}^k$ as this volume is independent of $\vect{a}$.

    Fix $\eta>0$, we define another probability measure ${P}^{(\eta)}_u$ which has density $\rhovaname(u,\vect{a})$ for all $\vect{a}\in \mathbb{R}^k$.
    \begin{align*}
        \rhova{u,\vect{a}}\defeq \frac{{P}_u({\Aball{\vect{a}}{\eta}}) }{\ballsize{\eta}}.
    \end{align*}
    $\rho_V(u,.)$ defines a valid probability measure because
    \begin{align*}
        \int_{\mathbb{R}^k} \rho_V(u,\vect{a}) d\vect{a}
        =& \int_{\mathbb{R}^k} \frac{ P_u{(\Aball{\vect{a}}{\eta}}) }{\ballsize{\eta}} d\vect{a} \\
        =& \frac{ 1 }{\ballsize{\eta}} \int_{\vect{a}\in \mathbb{R}^k}\int_{\vect{b}\in \Aball{\vect{a}}{\eta}}  P_u({d\vect{b}})  d\vect{a} \\
        =& \frac{ 1 }{\ballsize{\eta}} \int_{\vect{b}\in \mathbb{R}^k} P_u({d\vect{b}}) \int_{\vect{a}\in \Aball{\vect{b}}{\eta}}    d\vect{a} \\
        =& \frac{ 1 }{\ballsize{\eta}} \int_{\vect{b}\in \mathbb{R}^k} P_u({d\vect{b}}) \ballsize{\eta} = 1.
    \end{align*}


To show $\textbf{P}^{(\eta)} \in \mathcal{P}^{(\epsilon)}_{\mathbb{R}^k}$, we now show

$$ \rho_V(u_1, \vect{a}) = \frac{{P}_{u_1}(\Aball{\vect{a}}{\eta})}{\ballsize{\eta}} \leq \frac{{P}_{u_2}(\Aball{\vect{a}}{\eta}) e^{\epsilon |u_1-u_2|}}{\ballsize{\eta}} = e^{\epsilon |u_1-u_2|} \rho_V(u_2,\vect{a}) $$

This implies for any measurable set $S$ in $\mathbb{R}^k$, we have ${P}^{(\eta)}_{u_1}(S) \leq e^{\epsilon |u_1-u_2|}{P}^{(\eta)}_{u_2}(S)$ and thus satisfies $\epsilon$-geographic differential privacy and also implies that that $\rho_V(u,\vect{a})$ is continuous in $u$ for every $\vect{a} \in \mathbb{R}^k$.

Also observe that $\rho_V(u,\vect{a})$ is always bounded above by $\frac{1}{|B_{\eta}|}$ as ${P}_{u}(.)$ is a probability measure. Observe that since, we showed that $\rho_V(.,\vect{a})$ is \textit{Riemann integrable} and is non-negative for every $\vect{a} \in \mathbb{R}^k$ and $\rho_V(u,\vect{a})$ is continuous in $u$, we showed that $\rho_V(.,.)$ is Riemann integrable over $\mathbbm{R} \times \mathbbm{R}^k$. \footnote{Observe that this integral may go to $\infty$ but is always defined.}

In the next set of equations \eqref{eq:costetaend-1} and \eqref{eq:costetaend-2}, we upper and lower bound $\mathbb{E}_{\vect{a} \sim {P}^{(\eta)}_u} \left[\min\limits_{a \in \texttt{Set}(\vect{a})} \mathfrak{h}(|u-a|)\right]$ and then apply limit $\eta \to 0$ on both the upper and lower bounds in Equations \eqref{eqn:supremum-convergence-1} and \eqref{eqn:supremum-convergence-2} to show that $\textbf{P}^{(\eta)}$ arbitrarily approximates $\textbf{P}$ in Equation \eqref{eq:arbitrary-approx}.

    


    
    \begin{align}
         \mathbb{E}_{\vect{a} \sim {P}^{(\eta)}_u} \left[\min\limits_{a \in \texttt{Set}(\vect{a})} \mathfrak{h}(|u-a|)\right]
        =&
        \int_{\vect{a}\in \mathbb{R}^k} \min_{a\in \texttt{Set}(\vect{a})} \mathfrak{h}(|u-a|) \rho_V(u,\vect{a}) d\vect{a}  \label{eq:costetabegin} \\
        =&
        \int_{\vect{a}\in \mathbb{R}^k} \min_{a\in \texttt{Set}(\vect{a})} \mathfrak{h}(|u-a|) \frac{P_u(\Aball{\vect{a}}{\eta})}{\ballsize{\eta}} d\vect{a}   \\
        =&
        \int_{\vect{a}\in \mathbb{R}^k} \int_{\vect{b} \in \Aball{\vect{a}}{\eta}} \min_{a\in \texttt{Set}(\vect{a})} \mathfrak{h}(|u-a|) \frac{P_u(d\vect{b})}{\ballsize{\eta}} d\vect{a}   \\
        =&
        \frac{ 1 }{\ballsize{\eta}}  \int_{\vect{b}\in \mathbb{R}^k}  \int_{\vect{a} \in \Aball{\vect{b}}{\eta}}  P_u(d\vect{b})  \min_{a\in \texttt{Set}(\vect{a})} \mathfrak{h}(|u-a|) d\vect{a}   \\
        \leq&
        \frac{ 1 }{\ballsize{\eta}}  \int_{\vect{b}\in \mathbb{R}^k}  \int_{\vect{a} \in \Aball{\vect{b}}{\eta}}   \left(\min_{a\in \texttt{Set}(\vect{b})} \mathfrak{h}(|u-a|+\eta)\right) d\vect{a} P_u(d\vect{b})\label{ineq:costaddeta} \\
        =&
        \frac{ 1 }{\ballsize{\eta}} \int_{\vect{b}\in \mathbb{R}^k}  \int_{\vect{a} \in \Aball{\vect{b}}{\eta}}  \left(\min_{a\in \texttt{Set}(\vect{b})} \mathfrak{h}(|u-a|+\eta)\right) d\vect{a} P_u(d\vect{b})\label{eq:mct}\\
        =&
          \int_{\vect{b}\in \mathbb{R}^k} \left(\min_{a\in \texttt{Set}(\vect{b})} \mathfrak{h}(|u-a|+\eta) \right) P_u(d\vect{b})  \\
        =& \mathbb{E}_{\vect{a} \sim P_u} \left[\min\limits_{a \in \texttt{Set}(\vect{a})} \mathfrak{h}(|u-a|+\eta)\right] \label{eq:costetaend-1}
    \end{align}
    where Inequality \ref{ineq:costaddeta} holds because for any $\vect{b}\in \mathbb{R}^k$ and $\vect{a}\in \Aball{\vect{b}}{\eta}$, we have
    \begin{align*}
        & \min_{a\in \texttt{Set}(\vect{a})}  (|u-a|)
        \overset{(a)}\leq
        \min_{a\in \texttt{Set}(\vect{b})} (|u-a|) + \eta.\\
        \overset{(b)}{\implies} & \min_{a\in \texttt{Set}(\vect{a})}  \mathfrak{h}(|u-a|)
        \leq
        \min_{a\in \texttt{Set}(\vect{b})} \mathfrak{h}(|u-a|+ \eta)
    \end{align*}    

$(a)$ follows from triangle inequality and the fact that $|a_i-b_i| \leq \eta \text{ }\forall i \in [k]$. $(b)$ follows since $\mathfrak{h}$ is monotonic.

Equality $\eqref{eq:mct}$ follows from monotone convergence theorem since $\min_{a\in \texttt{Set}(\vect{b})} \mathfrak{h}(|u-a|+\eta)$ is monotonic in $\eta$ as $\mathfrak{h}(.)$ is monotonic.

We now show that $\mathbb{E}_{\vect{a} \sim {P}^{(\eta)}_u} \left[\min\limits_{a \in \texttt{Set}(\vect{a})} \mathfrak{h}(|u-a|)\right] \geq \mathbb{E}_{\vect{a} \sim P_u} \left[\min\limits_{a \in \texttt{Set}(\vect{a})} \mathfrak{h}(\left[|u-a|-\eta\right]_{+})\right] $. Note that in the following expression $[x]_{+}$ denotes $\max(x,0)$.

\begin{align}
         \mathbb{E}_{\vect{a} \sim {P}^{(\eta)}_u} \left[\min\limits_{a \in \texttt{Set}(\vect{a})} \mathfrak{h}(|u-a|)\right]
        =&
        \int_{\vect{a}\in \mathbb{R}^k} \min_{a\in \texttt{Set}(\vect{a})} \mathfrak{h}(|u-a|) \rho_V(u,\vect{a}) d\vect{a}  \label{eq:costetabegin-2} \\
        =&
        \int_{\vect{a}\in \mathbb{R}^k} \min_{a\in \texttt{Set}(\vect{a})} \mathfrak{h}(|u-a|) \frac{P_u(\Aball{\vect{a}}{\eta})}{\ballsize{\eta}} d\vect{a}   \\
        =&
        \int_{\vect{a}\in \mathbb{R}^k} \int_{\vect{b} \in \Aball{\vect{a}}{\eta}} \min_{a\in \texttt{Set}(\vect{a})} \mathfrak{h}(|u-a|) \frac{P_u(d\vect{b})}{\ballsize{\eta}} d\vect{a}   \\
        =&
        \frac{ 1 }{\ballsize{\eta}} \int_{\vect{b}\in \mathbb{R}^k}  \int_{\vect{a} \in \Aball{\vect{b}}{\eta}}  P_u(d\vect{b})  \min_{a\in \texttt{Set}(\vect{a})} \mathfrak{h}(|u-a|) d\vect{a}   \\
        \geq&
        \frac{ 1 }{\ballsize{\eta}}  \int\limits_{\vect{b}\in \mathbb{R}^k}  \int\limits_{\vect{a} \in \Aball{\vect{b}}{\eta}}   \min_{a\in \texttt{Set}(\vect{b})} \mathfrak{h}\left(\left[|u-a|-\eta\right]_{+}\right) d\vect{a} P_u(d\vect{b})\label{ineq:costaddeta-2} \\
        =&
        \frac{ 1 }{\ballsize{\eta}} \int\limits_{\vect{b}\in \mathbb{R}^k}  \int\limits_{\vect{a} \in \Aball{\vect{b}}{\eta}}   \min_{a\in \texttt{Set}(\vect{b})} \mathfrak{h}\left(\left[|u-a|-\eta\right]_{+}\right) d\vect{a} P_u(d\vect{b})\label{eq:mct-2}\\
        =&
          \int_{\vect{b}\in \mathbb{R}^k} \min_{a\in \texttt{Set}(\vect{b})} \mathfrak{h}\left(\left[|u-a|-\eta\right]_{+}\right) P_u(d\vect{b})  \\
        =& \mathbb{E}_{\vect{a} \sim P_u} \left[\min_{a\in \texttt{Set}(\vect{a})} \mathfrak{h}\left(\left[|u-a|-\eta\right]_{+}\right)\right] \label{eq:costetaend-2}
    \end{align}
    where Inequality \ref{ineq:costaddeta-2} holds because for any $\vect{b}\in \mathbb{R}^k$ and $\vect{a}\in \Aball{\vect{b}}{\eta}$, we have
    \begin{align*}
        & \min_{a\in \texttt{Set}(\vect{a})}  (|u-a|)
        \overset{(a)}\geq
        \min_{a\in \texttt{Set}(\vect{b})} \left[(|u-a|) - \eta\right]_{+}.\\
        \overset{(b)}{\implies} & \min_{a\in \texttt{Set}(\vect{a})}  \mathfrak{h}(|u-a|)
        \geq
        \min_{a\in \texttt{Set}(\vect{b})} \mathfrak{h}\left(\left[|u-a|-\eta\right]_{+}\right)
    \end{align*}    

$(a)$ follows from triangle inequality and the fact that $|a_i-b_i| \leq \eta \text{ }\forall i \in [k]$ and the fact $(b)$ follows since $\mathfrak{h}$ is monotonic and $\mathfrak{h}(0)=0$

Now observe that $\mathbb{E}_{\vect{a} \sim P_u} \left[\min\limits_{a\in \texttt{Set}(\vect{a})} \mathfrak{h}\left(\left[|u-a|-\eta\right]_{+}\right)\right]$
and $\mathbb{E}_{\vect{a} \sim P_u} \left[\min\limits_{a \in \texttt{Set}(\vect{a})} \mathfrak{h}(|u-a|+\eta)\right]$ is continuous in $\eta$ for every $u$. Observe that each of these functions is monotonic in $\eta$. 

Now we bound $\mathbb{E}_{\vect{a} \sim P_u} \left[\min\limits_{a \in \texttt{Set}(\vect{a})} \mathfrak{h}(|u-a|+\eta)\right]$ below by some linear function on \\$\mathbb{E}_{\vect{a} \sim P_u} \left[\min\limits_{a \in \texttt{Set}(\vect{a})} \mathfrak{h}(|u-a|)\right]$ for every $u \in \mathbb{R}$ and use it to show the convergence of the supremum in \eqref{eqn:supremum-convergence-1} and a similar approach for $\mathbb{E}_{\vect{a} \sim P_u} \left[\min\limits_{a \in \texttt{Set}(\vect{a})} \mathfrak{h}(|u-a|+\eta)\right]$. Let us denote the Lipschitz constant of $\log \mathfrak{h}(.)$ in $[\mathcal{B},\infty)$ by $M$. We now split $\mathbb{R}^k$ into two sets $\vect{a}_1(u) = \{\vect{a}: \min\limits_{a \in \texttt{Set}(\vect{a})}|u-a| > \mathcal{B}\}$ [recall from definition of $\mathcal{B}$ from Property \ref{item:second-property-h}] and $\vect{a}_2(u) = \{\vect{a}: \min\limits_{a \in \texttt{Set}(\vect{a})}|u-a| \leq \mathcal{B}\}$. 
With this, we now define 
\begin{equation}{\label{eq:kappa_defn}}
    \kappa(\eta): = \sup_{\mathfrak{v} \in [0,\mathcal{B}]} \mathfrak{h}(\mathfrak{v}+\eta) - \mathfrak{h}(\mathfrak{v})
\end{equation}

Observe that this supremum is finite for every $\eta$ since, the interval is closed and bounded and the function is continuous. We can now apply Dini's theorem \cite[Theorem 7.13]{rudin1976principles} to conclude $\kappa(\eta)$ goes to 0 as $\eta$ tends to 0. This argument goes as follows.

\begin{itemize}
    \item Observe that $\mathfrak{h}(\mathfrak{v}+\eta)$ converges point-wise in $\mathfrak{v}$ monotonically to $\mathfrak{h}(\mathfrak{v})$
    \item Since, $[0,\mathcal{B}]$ is a compact set(closed and bounded), we can argue that $\mathfrak{h}(\mathfrak{v}+\eta)$ uniformly converges to $\mathfrak{h}(\mathfrak{v})$ for $\mathfrak{v} \in [0,\mathcal{B}]$ implying convergence of $\kappa(\eta)$ to 0 as $\eta \to 0$.
\end{itemize}


Now observe for every $\vect{a} \in \vect{a}_1(u)$, we get the following

\begin{align}{\label{eq:ineq-firstpart}}
    \min\limits_{a \in \texttt{Set}(\vect{a})} \mathfrak{h}\left(|u-a|+\eta\right) = \mathfrak{h}\left(\min\limits_{a \in \texttt{Set}(\vect{a})} \mathfrak{h}(|u-a|+\eta\right) & \overset{(a)}{\leq} \mathfrak{h}\left(\min\limits_{a \in \texttt{Set}(\vect{a})}|u-a|\right) + \kappa(\eta)\nonumber\\
    & = \min\limits_{a \in \texttt{Set}(\vect{a})} \mathfrak{h}(|u-a|) + \kappa(\eta)
\end{align}

Observe that $(a)$ follows from Definition \ref{eq:kappa_defn} since $\min\limits_{a \in \texttt{Set}(\vect{a})}|u-a| \leq \mathcal{B}$ and the equality follows since $\mathfrak{h}(.)$ is monotonic thus, $\mathfrak{h}(\min(a,b)) = \min \left(\mathfrak{h}(a),\mathfrak{h}(b)\right)$

Similarly, for every $\vect{a} \in \vect{a}_2(u)$, we get the following

\begin{align}{\label{eq:ineq-secondpart}}
    \min\limits_{a \in \texttt{Set}(\vect{a})} \mathfrak{h}(|u-a|+\eta)= \mathfrak{h}\left(\min\limits_{a \in \texttt{Set}(\vect{a})} (|u-a|+\eta)\right) & \overset{(b)}{\leq} \mathfrak{h}\left(\min\limits_{a \in \texttt{Set}(\vect{a})}|u-a|\right)\cdot e^{M\eta}\nonumber\\
    & = \min\limits_{a \in \texttt{Set}(\vect{a})} \mathfrak{h}(|u-a|)\cdot e^{M\eta}
\end{align}

$(b)$ follows since, $\log \mathfrak{h}(z)$ is Lipschitz-continous in $[\mathcal{B},\infty)$, thus 
$\mathfrak{h}(z+\eta) \leq \mathfrak{h}(z) e^{M\eta}$ for $z \geq \mathcal{B}$.

Combining equations \eqref{eq:ineq-firstpart} and \eqref{eq:ineq-secondpart} coupled with the fact that $\vect{a}_1(u) \cup \vect{a}_2(u) = \mathbb{R}^k$, we get the following result for every $\vect{a} \in \mathbb{R}^k$

\begin{align}
    \min\limits_{a \in \texttt{Set}(\vect{a})} \mathfrak{h}(|u-a|+\eta) & \leq \max\left(\min\limits_{a \in \texttt{Set}(\vect{a})} \mathfrak{h}(|u-a|)\cdot e^{M\eta}, \min\limits_{a \in \texttt{Set}(\vect{a})} \mathfrak{h}(|u-a|) + \kappa(\eta)\right)\\
    & \leq \min\limits_{a \in \texttt{Set}(\vect{a})} \mathfrak{h}(|u-a|)\cdot e^{M\eta} + \kappa(\eta)
\end{align}

 With this, we now observe that

\begin{align}{\label{eqn:supremum-convergence-1}}
    & \sup_{u \in \mathbb{R}}\left(\mathbb{E}_{\vect{a} \sim P_u} \left[\min\limits_{a \in \texttt{Set}(\vect{a})} \mathfrak{h}(|u-a|+\eta)\right]\right) \leq \sup_{u \in \mathbb{R}} \left(\mathbb{E}_{\vect{a} \sim P_u} \left[\min\limits_{a \in \texttt{Set}(\vect{a})} \mathfrak{h}(|u-a|)\right]\right)e^{M\eta} + \kappa(\eta)\nonumber\\
    \overset{(c)}{\implies}  & \lim\limits_{\eta \to 0}
    \left(\sup_{u \in \mathbb{R}}\left(\mathbb{E}_{\vect{a} \sim P_u} \left[\min\limits_{a \in \texttt{Set}(\vect{a})} \mathfrak{h}(|u-a|+\eta)\right]\right)\right) \leq \sup_{u \in \mathbb{R}}\left(\mathbb{E}_{\vect{a} \sim P_u} \left[\min\limits_{a \in \texttt{Set}(\vect{a})} \mathfrak{h}(|u-a|)\right]\right)
\end{align}

$(c)$ follows since $\lim_{\eta\to 0} \kappa(\eta) = 0$

Using a very similar approach, we can show that 

\begin{align}{\label{eqn:supremum-convergence-2}}
    \lim\limits_{\eta \to 0}\left(\sup_{u \in \mathbb{R}}\left(\mathbb{E}_{\vect{a} \sim P_u} \left[\min\limits_{a \in \texttt{Set}(\vect{a})} \mathfrak{h}(\left[|u-a|-\eta\right]_{+}\right]\right)\right) \geq \sup_{u \in \mathbb{R}}\left(\mathbb{E}_{\vect{a} \sim P_u} \left[\min\limits_{a \in \texttt{Set}(\vect{a})} \mathfrak{h}(|u-a|)\right]\right)
\end{align}

Now, combining equations \eqref{eqn:supremum-convergence-1}, \eqref{eqn:supremum-convergence-2} and bounds in equation \eqref{eq:costetaend-2} and \eqref{eq:costetaend-1}, we obtain 

\begin{equation}{\label{eq:arbitrary-approx}}
    \lim_{\eta \to 0}\sup_{u \in \mathbb{R}}\left(\mathbb{E}_{\vect{a} \sim {P}^{(\eta)}_u} \left[\min\limits_{a \in \texttt{Set}(\vect{a})} \mathfrak{h}(|u-a|)\right]\right) = \sup_{u \in \mathbb{R}}\left(\mathbb{E}_{\vect{a} \sim P_u} \left[\min\limits_{a \in \texttt{Set}(\vect{a})} \mathfrak{h}(|u-a|)\right]\right)
\end{equation}

Thus, we show that the distributions $\textbf{P}^{(\eta)} \in \mathcal{P}^{(\epsilon)}_{\mathbb{R}^k}$ with a valid probability density function defined by $\rho_V(u,.)$ arbitrarily approximates $\textbf{P}$ for every distribution $\textbf{P} \in \mathcal{P}^{(\epsilon)}_{\mathbb{R}^k}$.



We now construct the DILP $\mathcal{O}$ in \eqref{orig_primal} where $g(u,.): \mathbb{R}^k \rightarrow \mathbb{R}$ denotes the probability density corresponding to ${P}^{(\eta)}_u$ for every $u \in \mathbb{R}$. The last 2 conditions in DILP $\mathcal{O}$ follow as a consequence of Lemma \ref{lemma:diff-privacy-condition}. 



 \end{proof}

\subsection{Dual fitting: Proof of Lemma \ref{lemma:bounding_f}}\label{sec-appendix:primal-feasible}

\begin{restatable} 
[detailed proof of Lemma \ref{lemma:bounding_f}] {lemma*}{simrstepa}\label{lemma-complete:bounding_f}
    $\hat{\textbf{P}}^{\mathcal{L}_{\epsilon}}$  satisfies $\epsilon$-geographic differential-privacy constraints i.e. $\hat{P}^{\mathcal{L}_{\epsilon}} \in \mathcal{P}^{(\epsilon)}_{\mathbb{R}^k}$ and thus, we have $f(\epsilon,k) \leq \text{cost}(\hat{\textbf{P}}^{\mathcal{L}_{\epsilon}}) = \hat{f}(\epsilon,k)$
     
\end{restatable}

\begin{proof}
     We denote $\vect{a}+z$ as $\vect{a}+ z\mathbbm{1}_k$ (i.e.) adding $z \in U$ to every component of $\vect{a} \in U^k$ 


    \begin{align*}
        f(\epsilon,k) & \overset{(a)}{\leq} \sup_{u \in \mathbb{R}} \underset{\vect{a} \sim \hat{P}^{\mathcal{L}_{\epsilon}}_u}{\mathbb{E}} \left[\min_{a \in \texttt{Set}(\vect{a})} \mathfrak{h}(|u-a|) \right] \\
        & \overset{(b)}{=} \sup_{u \in \mathbb{R}} \underset{S_u \sim \mathcal{L}_{\epsilon}(u)}{\mathbb{E}} \left(\left[\min_{a \in \texttt{Set}(\hat{\vect{a}})} \mathfrak{h}(|u-a|) \right]\Biggl\vert_{\hat{\vect{a}}=\argmin\limits_{\vect{a} \in \mathbb{R}^k} \mathbb{E}_{y \sim \mathcal{L}_{\epsilon}(0)}\left[\min\limits_{a \in \texttt{Set}(\vect{a})} \mathfrak{h} (|y -a |)\right] + S_u}\right)\\
        & \overset{(c)}{=} \sup_{u \in \mathbb{R}} \underset{S_u \sim \mathcal{L}_{\epsilon}(u)}{\mathbb{E}} \left(\left[\min_{a \in \texttt{Set}(\hat{\vect{b}})} \mathfrak{h}(|(u-S_u)-a|) \right]\Biggl\vert_{\hat{\vect{b}}=\argmin\limits_{\vect{a} \in \mathbb{R}^k} \mathbb{E}_{y \sim \mathcal{L}_{\epsilon}(0)}\left[\min\limits_{a \in \texttt{Set}(\vect{a})} \mathfrak{h} (|y -a|)\right] }\right)\\
        & \overset{(d)}{=} \sup_{u \in \mathbb{R}} \underset{y \sim \mathcal{L}_{\epsilon}(0)}{\mathbb{E}} \left[\min_{a \in \texttt{Set}(\hat{\vect{b}})} \mathfrak{h}(|y-a|) \right]\Biggl\vert_{\hat{\vect{b}}=\argmin\limits_{\vect{a} \in \mathbb{R}^k} \mathbb{E}_{y \sim \mathcal{L}_{\epsilon}(0)}\left[\min\limits_{a \in \texttt{Set}(\vect{a})} \mathfrak{h} (|y-a|)\right] }\\
        & \overset{(e)}{=}  \min_{\vect{a} \in \mathbb{R}^k} \mathbb{E}_{y \sim \mathcal{L}_{\epsilon}(0)} \left[\min_{a \in \texttt{Set}(\vect{a})} \mathfrak{h} (|y-a|)\right] = \hat{f}(\epsilon,k)
    \end{align*}


$(a)$ follows from the fact that $\hat{P}^{\mathcal{L}_{\epsilon}} \in \mathcal{P}^{(\epsilon)}_{\mathbb{R}^k}$ from the post processing theorem, refer to \cite{diff-privacy-book} since $S_u \sim \mathcal{L}_{\epsilon}(u)$ satisfies $\epsilon$-geographic differential privacy constraints.\footnote{Post processing theorem can be proven even for $\epsilon$-geographic differential privacy similarly}. 
$(b)$ follows from the definition in Equation \eqref{eq:hat_f_defn} and $(c)$ follows on substituting $\hat{\vect{b}} = \hat{\vect{a}}-S_u$ and $(d)$ follows on substituting $u-S_u$ by $y$ and the fact that $\hat{\vect{b}}$ is independent of $S_u$ and $(e)$ follows since $\hat{\vect{b}}$ is minimised over the same objective function and independent of $u$.

\end{proof}
\subsection{Dual fitting: Proofs of Lemmas \ref{lemma:diff_eqn_soln} and Lemma \ref{lemma:dual_achievable}}\label{sec-appendix:dual-soln-construction}





Recall from Section \ref{subsec-comb:dual-soln-construction} the following assignment to functions $\delta^{(c)}(.)$ and $\lambda^{(c)}(.)$

\begin{equation*}
    \delta^{(c)}(r) = (\zeta/2) e^{-\zeta|r|} \text{ and } \lambda^{(c)}(r) = \hat{\lambda}\cdot (\zeta/2) e^{-\zeta|r|} \text{ } \forall r \in \mathbb{R}
\end{equation*}


Recall that we considered the following Differential Equation \eqref{eqn:diff_eqn_nu} in $\hat{\nu}(.)$. 

\begin{equation*}
    -\left[\min_{a \in \texttt{Set}(\vect{v})} \mathfrak{h}(|r-a|)\right] \delta^{(c)}(r) + \lambda^{(c)}(r) + \frac{d\hat{\nu}(r)}{dr} + \epsilon |\hat{\nu}(r)| = 0 ; \text{ with  $\hat{\nu}(v_{med}) = 0$}  
\end{equation*}

We show that there always exists a solution $\hat{\nu}(.)$ to \eqref{eqn:diff_eqn_nu} such that $\hat{\nu}(r)$ is non-negative for sufficiently large $u$ and non-positive for sufficiently small $r$ \footnote{We may still need to prove continuity of the bounds $U(.)$ and $L(.)$ which we discuss in Lemma \ref{lemma:dual_achievable}} to satisfy the last two constraints of DILP $\mathcal{E}$ in Section below. Observe that the solution could depend on $\vect{v}$. In the next subsection, we prove two technical lemmas which is used in the proof of Lemmas and Claims in Section \ref{subsec-appendix:diff_eqn_soln}.

\subsubsection{Technical lemmas used }



\begin{lemma}{\label{hat_f_lemma}}
   Consider any vector $\vect{q} \in \mathbb{R}^k$ and consider any $v_{med} \in \mathbb{R}$ satisfying $v_{med} \geq \text{Median}(\texttt{Set}(\vect{q}))$. Then the following holds true (recall the definition of $\hat{f}(\epsilon,k)$ from Equation \eqref{eq:hat_f_defn})

   \begin{equation}
       2 \int_{v_{med}}^{\infty} \left[ \min_{a \in \texttt{Set}(\vect{q})}  \mathfrak{h}(|t-a|)\right] \left(\epsilon/2\right) e^{-\epsilon (t-v_{med})} dt\geq \hat{f}(\epsilon,k).
   \end{equation}
\end{lemma}

\textbf{Notation}: In this proof we use $q_{[j:k]}$ to denote a vector in $\mathbb{R}^{k-j+1}$ constructed from all coordinates of $\vect{q}$ starting from $j^{th}$ coordinate to the $k^{th}$ coordinate of $\vect{q}$.

\begin{proof}

    Suppose the median of $\texttt{Set}(\vect{q})$ is denoted by $q_{med}$ and construct $\vect{v} = \vect{q}+ (v_{med} - q_{med}) \mathbbm{1}_k$ \footnote{$\mathbbm{1}_k$ is a vector in $\mathbb{R}^k$ with all elements unity}. W.L.O.G, we assume that components in $\vect{q}$ are sorted in ascending order. Consider the smallest index $j$ in $[k]$ s.t. $q_j > v_{med}$. Clearly, $j > \frac{k+1}{2}$ if $k$ is odd and $j>\frac{k}{2}$ if $k$ is even.

    We construct a vector $\tilde{\vect{q}} \in \mathbb{R}^k$ by defining $\tilde{q}_{[j:k]} = q_{[j:k]}$ and $\tilde{q}_{[1:k-j+1]} = 2v_{med} \mathbbm{1} - q_{[j:k]}$. For all other entries of $\tilde{\vect{q}}$, define it to be $v_{med}$. Observe that we choose a points from $q_{[j:k]}$ and choose other points by symmetrizing around $v_{med}$ and $j>\frac{k}{2} + 1$ implies $v_{med} \in \texttt{Set}(\tilde{\vect{q}})$ as $2(k-j+1) < k$. Now consider the following two exhaustive cases.


    Case - 1: \textbf{$j > \frac{k}{2} + 1$}.

    \begin{align}
        2 \int_{v_{med}}^{\infty} \left[ \min_{a \in \texttt{Set}(\vect{q})}  \mathfrak{h}(|t-a|) \right]&  \left(\epsilon/2\right) e^{-\epsilon (t-v_{med})} dt \\
       & \overset{(a)}{\geq} 2\int_{v_{med}}^{\infty} \left[ \min_{a \in \texttt{Set}(q_{[j:k]}) \cup \{v_{med}\}}  \mathfrak{h}(|t-a|)\right] \left(\epsilon/2\right) e^{-\epsilon (t-v_{med})} dt \\
        & \overset{(b)}{=} \int_{-\infty}^{\infty} \left[ \min_{a \in \texttt{Set}(\tilde{\vect{q}})}  \mathfrak{h}(|t-a|)\right] \left(\epsilon/2\right) e^{-\epsilon |t-v_{med}|} dt \\
        & = \mathbb{E}_{y \sim \mathcal{L}_{\epsilon}(v_{med})} \left[ \min_{a \in \texttt{Set}(\tilde{\vect{q}})}  \mathfrak{h}(|y-a|)\right] \overset{(c)}{\geq} \hat{f}(\epsilon,k).
    \end{align}

$(a)$ follows from the fact that $q_j$ is the smallest element in list larger than $v_{med}$. $(b)$ follows from the fact that $\tilde{\vect{q}}$
is constructed from a symmetric extension of $q_{[j:k]}$ about $v_{med}$ and $v_{med} \in \texttt{Set}(\tilde{\vect{q}})$. $(c)$ follows from the fact that it is minimised over all collection of $k$ points as defined in Equation \eqref{eq:hat_f_defn}.

Case -2: $j = \frac{k}{2}+1$. And thus, 

    \begin{align}
        2 \int_{v_{med}}^{\infty} \left[ \min_{a \in \texttt{Set}(\vect{q})}  \mathfrak{h}(|t-a|)\right]& \left(\epsilon/2\right) e^{-\epsilon (t-v_{med})} dt \\
        & \overset{(a)}{\geq} 2\int_{v_{med}}^{\infty} \left[ \min_{a \in \texttt{Set}(q_{[j:k]}) }  \mathfrak{h}(|t-a|)\right] \left(\epsilon/2\right) e^{-\epsilon (t-v_{med})} dt \\
        & \overset{(b)}{=} \int_{-\infty}^{\infty} \left[ \min_{a \in \texttt{Set}(\tilde{\vect{q}})}  \mathfrak{h}(|t-a|)\right] \left(\epsilon/2\right) e^{-\epsilon |t-v_{med}|} dt \\
        & = \mathbb{E}_{y \sim \mathcal{L}_{\epsilon}(v_{med})} \left[ \min_{a \in \texttt{Set}(\tilde{\vect{q}})}  \mathfrak{h}(|y-a|)\right] \overset{(c)}{\geq} \hat{f}(\epsilon,k).
    \end{align}

Note that $(a)$ follows from the fact that $q_{j} + q_{j-1} > 2v_{med}$ as $\text{Median}(\texttt{Set}(\vect{q})) > v_{med}$ and $\text{Median}(\texttt{Set}(\vect{q})) = \frac{q_j+q_{j-1}}{2}$. $(b)$ follows from the fact that $\tilde{\vect{q}}$ is constructed from a symmetric extension of $q_{[j:k]}$ about $v_{med}$. $(c)$ follows from the fact that it is minimised over all collection of $k$ points as defined in Equation \eqref{eq:hat_f_defn}.

\end{proof}

We now state the next lemma which say that we says that there exists a solution to Equation \eqref{eqn:diff_eqn_nu}. This lemma states that zeros of the solution to Equation \eqref{eqn:diff_eqn_nu} with negative derivative is upper bounded and states that zeros with positive derivative is lower bounded. 

\begin{lemma}{\label{diff_eqn_soln_first_lemma}}
    Choose any $0<\hat{\lambda} \leq \frac{\epsilon - \zeta}{\epsilon+\zeta} \hat{f}(\epsilon+ \zeta,k)$ for some $0<\zeta<\epsilon$. Then for every $\vect{v} \in \mathbbm{R}^k$, a unique $\mathcal{C}^1$ solution $\nu^{(c)}(.)$ exists to Differential Equation \eqref{eqn:diff_eqn_nu} satisfying the following two constraints for some constant $C$ independent of $\vect{v} \in \mathbb{R}^k$. Note that $v_{\floor{i}}$ denotes the $i^{th}$ largest component of $\vect{v}$ for every $i \in [k]$.

    \begin{itemize}
        \item $\set[\Big]{\mathfrak{r} \in \mathbb{R}: \nu^{(c)}(\mathfrak{r})=0 \text{ and } \frac{d\nu^{(c)}(\mathfrak{r})}{d\mathfrak{r}} < 0 } \text{ is upper bounded}$ by $C + v_{\floor{k}}$.
        \item $\set[\Big]{\mathfrak{r} \in \mathbb{R}: \nu^{(c)}(\mathfrak{r})=0 \text{ and } \frac{d\nu^{(c)}(\mathfrak{r})}{d\mathfrak{r}} > 0 } \text{ is lower bounded}$ by $-C + v_{\floor{1}}$.
    \end{itemize}
    
\end{lemma}

\begin{proof}
    Recall Differential Equation from Equation \eqref{eqn:diff_eqn_nu}, restate it below as follows.

    \begin{equation*}
    -\left[\min_{a \in \texttt{Set}(\vect{v})} \mathfrak{h}(|r-a|)\right] \delta^{(c)}(r) + \lambda^{(c)}(r) + \frac{d\hat{\nu}(r)}{dr} + \epsilon |\hat{\nu}(r)| = 0
    \end{equation*}

Also recall from Equation \eqref{eqn:nu_delta_defn} that we defined

$$\delta^{(c)}(r) = (\zeta/2) e^{-\zeta|r|} \text{ and } \lambda^{(c)}(r) = \hat{\lambda}\cdot (\zeta/2) e^{-\zeta|r|} \text{ } \forall r \in \mathbb{R}$$
    
    From \textit{global uniqueness and existence} condition of \textit{Cauchy-Lipschitz} theorem in \cite[Theorem 2]{existenceode}, we can show that there is a \textit{unique} solution $\hat{\nu}^{(c)}(.)$ of Equation \eqref{eqn:diff_eqn_nu} satisfying this initial value condition as \textit{Lipschitz} conditions as described in \cite[Theorem 2]{existenceode} is satisfied in Equation \eqref{eqn:diff_eqn_nu}. Also, observe that this function is \textit{continuously differentiable}.  
    
    We now prove the first point in the lemma.. 

        \textit{We first prove that $\lim_{u \to \infty} \mathfrak{h}(u)> \hat{\lambda}$}. First observe from Claim \ref{hat_f_lemma}) (substituting $\epsilon$ by $\epsilon+\zeta$) that

    \begin{align*}
    & \int_{v_{med}}^{\infty} \left[ \min_{a \in \texttt{Set}(\vect{q})}  \mathfrak{h}(|t-a|)\right] \left(\epsilon+ \zeta\right) e^{-(\epsilon+\zeta) (t-v_{med})} dt\geq \hat{f}(\epsilon+\zeta,k).\\
    \implies & \int_{v_{med}}^{\infty} \left[ \min_{a \in \texttt{Set}(\vect{q})}  \mathfrak{h}(|t-a|)-\hat{f}(\epsilon+\zeta,k)\right] (\epsilon+\zeta) e^{-(\epsilon+\zeta) (t-v_{med})} dt \geq 0
    \end{align*}

    This, implies that for some $t > v_{med}$ s.t. $\left[ \min_{a \in \texttt{Set}(\vect{q})}  \mathfrak{h}(|t-a|)-\hat{f}(\epsilon+\zeta,k)\right]>0$ and since, we know $\mathfrak{h}(.)$ is monotonic, we can argue that $\lim_{r \to \infty} \mathfrak{h}(r)> \hat{f}(\epsilon+\zeta,k)> \hat{\lambda}$.\footnote{The limits are defined over extended reals $\mathbb{R}\cup \{-\infty,+\infty\}$} Let us denote the the smallest $r$ s.t. $\mathfrak{h}(r) > \hat{\lambda}$ by $\mathfrak{u}^s$. Now observe for $r>(v_{\floor{k}}+ \mathfrak{r}^s)$, $\left[\min_{a \in \texttt{Set}(\vect{v})} \mathfrak{h}(|r-a|)\right]\delta^{(c)}(r) - \lambda^{(c)}(r)$ is non-negative.

    Suppose there exists $r_0 \in [v_{\floor{k}}+ \mathfrak{r}^s,\infty)$ satisfying, $\nu^{(c)}(r_0)= 0$ and $\nu^{(c)'}(r_0)<0$ thus, $\nu^{(c)'}(r_0)+\nu^{(c)}(r_0)<0$ , thus contradicting the fact that $\left[\min_{a \in \texttt{Set}(\vect{v})} \mathfrak{h}(|r-a|)\right]\delta^{(c)}(r) - \lambda^{(c)}(r)>0$ is non-negative in $r \in [v_{\floor{k}}+ \mathfrak{u}^s,\infty)$. Thus, we show that there does not exist $r_0 > v_{\floor{k}}+ \mathfrak{r}^s$ satisfying $\nu^{(c)}(u_0)= 0$ and $\nu^{(c)'}(r_0)<0$ thus proving the desired statement, since $\mathfrak{r}^s$ is a constant independent of $\vect{v}$.

    The second statement follows from a very similar argument.

\end{proof}


With these lemmas, we now prove the following lemma \ref{lemma:diff_eqn_soln} which shows that $\nu^{(c)}(.)$ is positive for sufficiently large $u$ and negative for sufficiently small $u$.

\subsubsection{Obtaining a feasible solution to DILP $\mathcal{E}$: Proof of Lemma \ref{lemma:diff_eqn_soln}}{\label{subsec-appendix:diff_eqn_soln}}


\begin{restatable} 
[detailed proof of Lemma \ref{lemma:diff_eqn_soln}] {lemma*}{simrstepa}
     Choose $\zeta<\epsilon$ and $0< \hat{\lambda} \leq \frac{\epsilon - \zeta}{\epsilon+\zeta} \hat{f}(\epsilon+ \zeta,k)$, then equation $\eqref{eqn:diff_eqn_nu}$ has a unique $\mathcal{C}^1$ solution $\nu^{(c)}(.)$ and there exists $U, L \in \mathbb{R}$ satisfying $\nu^{(c)}(r) \geq 0 \text{ }\forall r \geq U$ and $\nu^{(c)}(r) \leq 0 \text{ } \forall r \leq L$.
     
\end{restatable}

\begin{proof}

    


    Recall from lemma \ref{diff_eqn_soln_first_lemma}, that there exists a unique $\mathcal{C}^1$ solution (denoted by $\nu^{(c)}(.)$) to \eqref{eqn:diff_eqn_nu}. Uniqueness and existence of the solution can be shown from \textit{global uniqueness and existence} condition of \textit{Cauchy-Lipschitz} theorem in \cite[Theorem 2]{existenceode} as \textit{Lipschitz} conditions as described in \cite[Theorem 2]{existenceode} is satisfied in Equation \eqref{eqn:diff_eqn_nu}.

    We first prove the following Equation \eqref{eqn:non-existence2} in \ref{proof-eqn:non-existence2} which states that there cannot be an unbounded interval in $[v_{med},\infty)$ where $\nu^{(c)}(.)$ is non-positive. Next, we use this statement to prove the main lemma in \ref{proof-lemma}.


    \begin{equation}{\label{eqn:non-existence2}}
            \set[\Big]{ \mathfrak{r} \in [v_{med},\infty) \given \nu^{(c)}(r) \leq 0 \text{ }\forall r \geq \mathfrak{r}} = \phi
    \end{equation} 

    \begin{pfparts}
        
    \item[\namedlabel{proof-eqn:non-existence2}{Part - 1}.]

    We now assume the \textit{contradictory} statement that there exists $r^0 \in \mathbbm{R}$ s.t. ${\nu}^{(c)}(r)\leq 0$ for all $r>r^{0}$ and define $r^{\max} = \inf \set[\Big]{ \mathfrak{r} \in [v_{med},\infty) \given \nu^{(c)}(r) \leq 0 \text{ }\forall r \geq \mathfrak{r}}$. Since, $\nu^{(c)}(.)$ is continuous and $\nu^{(c)}(v_{med})=0$, we have $\nu^{(c)}(r^{max})=0$. Thus, the differential equation \eqref{eqn:diff_eqn_nu} can be rewritten in $[r^{\max},\infty)$ as follows (replacing $|\hat{\nu}(u)|$ by $-\hat{\nu}(u)$) 

    \begin{equation*}
    -\left[\min_{a \in \texttt{Set}(\vect{v})} \mathfrak{h}(|r-a|)\right] \delta^{(c)}(r) + \lambda^{(c)}(r) + \frac{d\hat{\nu}(r)}{dr} -\epsilon \hat{\nu}(r) = 0
    \end{equation*}





    Now, multiplying both sides by $e^{-\epsilon r}$ and applying the initial value condition of $\nu^{(c)}(v_{med})=0$, we have the solution to the Differential Equation in $[r^{\max},\infty)$ is given by 
    $${\nu}^{(c)}(r) = \left(\int_{r^{\max}}^{r} \left[-\lambda^{(c)}(t) +\left[\min_{a \in \texttt{Set}(\vect{v})} \mathfrak{h}(|t-a|) \right]\delta^{(c)}(t) \right] e^{-\epsilon t} dt\right)e^{\epsilon {r}}$$

However, since $\hat{\nu}^{(c)}(r)$ is non-positive in $[r^{\max},\infty)$, we have\footnote{Observe that a strict inequality follows since a contrary assumption would imply that $\nu^{(c)}(r)$ is identically zero in $[r^{med},\infty)$. This would imply that $-\lambda^{(c)}(t) +\left[\min\limits_{a \in \texttt{Set}(\vect{v})} \mathfrak{h}(|t-a|) \right]\delta^{(c)}(t)$ is zero in $[r_{max},\infty)$ implying, $\mathfrak{h}(|t-a|)=\hat{\lambda}$ which is a contradiction since, $\lim_{r \to \infty} \mathfrak{h}(r) > \hat{\lambda}$ as shown in proof of Claim \ref{diff_eqn_soln_first_lemma}}



\begin{align}
    & \int_{r^{\max}}^{\infty} \left[-\lambda^{(c)}(t) +\left[\min_{a \in \texttt{Set}(\vect{v})} \mathfrak{h}(|t-a|) \right]\delta^{(c)}(t) \right] e^{-\epsilon t} dt <0\\
    \implies & \int_{r^{\max}}^{\infty} \left[\min_{a \in \texttt{Set}(\vect{v})} \mathfrak{h}(|t-a|) \right] \left(\frac{{\zeta}}{2}\right) e^{-{\zeta}|t|} (\epsilon/2)e^{-\epsilon t} dt < \int_{r^{\max}}^{\infty} \lambda^{(c)}(t) (\epsilon/2)e^{-\epsilon t}dt \label{curr_eqn}
\end{align}

We now consider two exhaustive cases and treat them separately $r^{\max} > 0$ and $r^{\max}<0$. We denote $q_i=v_i+v_{med}-r^{\max} \forall i \in [k]$ and $\vect{q} = [q_1, q_2, \ldots, q_k]$

\textit{Case a}: Let us first consider the case where $r^{\max}>0$. We thus get from Equation \eqref{curr_eqn}

\begin{align}
    \overset{(a)}{\implies} & e^{-(\epsilon+{\zeta})(r^{\max}-v_{med})}. \int_{v_{med}}^{\infty} \left[\min_{\substack{a \in \texttt{Set}(\vect{q})}} \mathfrak{h}(|t-a|) \right] \left(\frac{{\zeta}}{2}\right) e^{-{\zeta}t} (\epsilon/2)e^{-\epsilon t} dt < (1/4)\hat{\lambda} e^{-(\epsilon+{\zeta})r^{\max}} \left(\frac{\epsilon{\zeta}}{\epsilon+{\zeta}}\right)\\
    \implies & 2\int_{v_{med}}^{\infty} \left[\min_{\substack{a \in \texttt{Set}(\vect{q}) }}\mathfrak{h}(|t-a|) \right]((\epsilon+{\zeta})/2)e^{-(\epsilon + {\zeta}) t} dt <\hat{\lambda}e^{-(\epsilon+{\zeta})v_{med}} \\
    {\implies} & 2\int_{v_{med}}^{\infty} \left[\min_{\substack{a \in \texttt{Set}(\vect{q}) }}\mathfrak{h}(|t-a|) \right]((\epsilon+{\zeta})/2)e^{-(\epsilon + {\zeta}) (t-v_{med})} dt <\hat{\lambda}\\
    \overset{(b)}{\implies} & \hat{f}(\epsilon+{\zeta},k) < \hat{\lambda}
\end{align}

which is a contradiction.

$(a)$ follows from the fact that $r^{\max}>0$, thus $|t|=t$ and applying change of variables. The argument for $(b)$ follows from Lemma \ref{hat_f_lemma}.



\textit{Case b}: Consider the case when $r^{\max}<0$. We can thus write Equation \eqref{curr_eqn}
\begin{align}
    {\implies} & \int_{r^{\max}}^{0} \left[\min_{a \in \texttt{Set}(\vect{v})} \mathfrak{h}(|t-a|) \right] \left(\frac{\epsilon{\zeta}}{4}\right) e^{{\zeta}t} e^{-\epsilon t} dt + \int_{0}^{\infty} \left[\min_{a \in \texttt{Set}(\vect{v})} \mathfrak{h}(|t-a|) \right] \left(\frac{{\zeta}\epsilon}{4}\right) e^{-{\zeta}t} e^{-\epsilon t} dt\\
    & < \int_{r^{\max}}^{0} \hat{\lambda} \left( \epsilon{\zeta}/4 \right) e^{{\zeta} t} e^{-{\epsilon} t} dt + \int_{0}^{\infty} \hat{\lambda} \left( \epsilon{\zeta}/4 \right) e^{-{\zeta} t} e^{-{\epsilon} t} dt \\
    \overset{(a)}{\implies} & e^{-(\epsilon-{\zeta}) r^{\max}}\int_{v_{med}}^{v_{med}-r^{\max}} \left[\min_{a \in \texttt{Set}(\vect{q})} \mathfrak{h}(|t-a|) \right] e^{-(\epsilon-{\zeta})(t-v_{med})}  dt  \\
    + & e^{-(\epsilon+{\zeta}) r^{\max}}\int_{v_{med}-r^{\max}}^{0} \left[\min_{a \in \texttt{Set}(\vect{q})} \mathfrak{h}(|t-a|) \right] e^{-(\epsilon+{\zeta})(t-v_{med})} dt <  \hat{\lambda}\left[ \frac{1}{(\epsilon+{\zeta})} + \frac{e^{-(\epsilon-{\zeta})r^{\max}} -1}{(\epsilon - {\zeta})}\right]\\
    \overset{(b)}{\implies} & e^{-(\epsilon-{\zeta}) r^{\max}}\int_{v_{med}}^{\infty} \left[\min_{a \in \texttt{Set}(\vect{q})} \mathfrak{h}(|t-a|) \right] e^{-(\epsilon+{\zeta})(t-v_{med})}  dt  < \hat{\lambda}\left[ \frac{1}{(\epsilon+{\zeta})} + \frac{e^{-(\epsilon-{\zeta})r^{\max}} -1}{(\epsilon -{\zeta})}\right]\\
    {\implies} & 2e^{-(\epsilon-{\zeta}) r^{\max}}\int_{v_{med}}^{\infty} \left[\min_{a \in \texttt{Set}(\vect{q})} \mathfrak{h}(|t-a|) \right] \left(\frac{\epsilon+{\zeta}}{2}\right)e^{-(\epsilon+{\zeta})(t-v_{med})} dt  < \hat{\lambda}\left[1 + \frac{(\epsilon + {\zeta})}{(\epsilon -{\zeta})} (e^{-(\epsilon-{\zeta})r^{\max}} -1)\right]\\
    {\implies} & \frac{e^{-(\epsilon-{\zeta}) r^{\max}}2} {\left[1 + \frac{(\epsilon + {\zeta})}{(\epsilon -{\zeta})} (e^{-(\epsilon-{\zeta})r^{\max}} -1)\right]}\int_{v_{med}}^{\infty} \left[\min_{a \in \texttt{Set}(\vect{q})} \mathfrak{h}(|t-a|) \right] \left(\frac{\epsilon+{\zeta}}{2}\right)e^{-(\epsilon+{\zeta})(t-v_{med})} dt  < \hat{\lambda}\\
    \overset{(c)}{\implies} & \left(\frac{\epsilon-{\zeta}}{\epsilon+{\zeta}}\right)2\int_{v_{med}}^{\infty} \left[\min_{a \in \texttt{Set}(\vect{q})} \mathfrak{h}(|t-a|) \right] \left(\frac{\epsilon+{\zeta}}{2}\right)e^{-(\epsilon+{\zeta})(t-v_{med})} dt  < \hat{\lambda}\\
    & \overset{(d)}{\implies} \left(\frac{\epsilon-{\zeta}}{\epsilon+{\zeta}}\right) \hat{f}(\epsilon+{\zeta},k)< \hat{\lambda}
\end{align}

which is a contradiction.

$(b)$ follows from the fact that $\epsilon-{\zeta} < \epsilon+{\zeta}$ and thus, $-(\epsilon-{\zeta})r^{\max} < -(\epsilon+{\zeta})r^{\max}$ since $r^{\max} <0$. Also, $-(\epsilon+{\zeta})(t-v_{med}) < -(\epsilon-{\zeta})(t-v_{med})$ for $t>v_{med}$.
$(c)$ follows from by infimising $\frac{e^{-(\epsilon-{\zeta}) r^{\max}}} {\left[1 + \frac{(\epsilon + {\zeta})}{(\epsilon -{\zeta})} (e^{-(\epsilon-{\zeta})r^{\max}} -1)\right]}$ over all $r^{\max}<0$ (happens as $r^{\max} \to -\infty$) and $(d)$ follows from Lemma \ref{hat_f_lemma}.

We thus prove Equation \eqref{eqn:non-existence2} by showing contradiction under both cases.





\item[\namedlabel{proof-lemma}{Part - 2}.]


Observe that Equation \eqref{eqn:existence1} follows from Claim \ref{diff_eqn_soln_first_lemma} and Equation \eqref{eqn:non-existence2} (restated below) which has been shown above.

\begin{equation}{\label{eqn:existence1}}
            \set[\Big]{\mathfrak{r} \in \mathbb{R}: \nu^{(c)}(\mathfrak{r})=0 \text{ and } \frac{d\nu^{(c)}(\mathfrak{r})}{d \mathfrak{r}} < 0 } \text{ is upper bounded by $C + v_{\floor{k}}$}. 
\end{equation}. 

\begin{equation*}
            \set[\Big]{ \mathfrak{r} \in \mathbbm{R}\given \nu^{(c)}(r) \leq 0 \text{ }\forall r>\mathfrak{r}} = \phi
\end{equation*}

    




Let us denote the set of all the closed intervals in $[v_{med},\infty)$, where $\nu^{(c)}(.)$ is non-negative by $\mathcal{I}^{pos}$\footnote{We can construct such intervals as $\nu^{(c)}$ is \textit{continuously differentiable} and there does not exist an interval where $\nu^{(c)}$ is identically zero.}. Now observe, that the Equation \eqref{eqn:existence1} implies that \textit{upper bound} of every interval in $\mathcal{I}^{pos}$ must be upper bounded by $C+ v_{\floor{k}}$. This implies the existence of an unbounded interval where $\nu^{(c)}(.)$ is non-negative or an unbounded interval where $\nu^{(c)}(.)$ is non-positive. However Equation \eqref{eqn:non-existence2} implies that there can not exist an unbounded interval where $\nu^{(c)}(.)$ is non-positive and thus, we show that there exists an unbounded interval where  $\nu^{(c)}(.)$ is non-negative. This proves the statement in Claim \ref{lemma:diff_eqn_soln} that there exists $U \in \mathbbm{R}$ s.t. $\nu^{(c)}(u)\geq 0 \forall u \geq U$. 

\end{pfparts}

Using a very similar approach, we can prove that $\nu^{(c)}(u) \leq 0$ for $u\leq L$ for some $L \in \mathbbm{R}$.




\end{proof}

We now prove the main lemma \ref{lemma:dual_achievable} which shows that the objective value of $\hat{f}(\epsilon,k)$ is achievable by some solution in the DILP $\mathcal{E}$. Observe that this feasible solution is constructed using $\delta^{(c)}(.), \lambda^{(c)}(.)$ and $\nu^{(c)}(.)$ that we defined to satisfy Equation \eqref{eqn:diff_eqn_nu} and ensured positivity and negativity of $\nu^{(c)}(.)$ for ``sufficiently'' large and small $u$ respectively for every $\vect{v} \in \mathbb{R}^k$. However, we also require a technical claim \ref{claim:continuity_bounds} to prove the existence of continuous bounds $U(.)$ and $L(.)$. For sake of brevity, we state and prove it in Appendix \ref{sec:continuous_bounds_claim}. 

\begin{restatable} 
[detailed proof of Lemma \ref{lemma:dual_achievable}] {lemma*}{simrstepa}
         $\text{opt}(\mathcal{E}) \geq \hat{f}(\epsilon,k)$
\end{restatable}


\begin{proof}
    Recall the functions $\lambda^{(c)}(.), \delta^{(c)}(.)$ defined in \eqref{eqn:nu_delta_defn}. Also for every $\vect{v} \in \mathbb{R}^k$, we obtain a function $\nu^{(c)}(.,\vect{v})$ [solution of Equation \eqref{eqn:diff_eqn_nu}] with bounds $U(\vect{v})$ and $L(\vect{v})$ satisfying $\nu^{(c)}(r,\vect{v}) \geq 0 \forall u\geq U(\vect{v})$ and $\nu^{(c)}(r,\vect{v}) \leq 0 \forall u\leq L(\vect{v})$. Now, we argue that this solution is feasible attaining an objective value of $\hat{\lambda}$ which we argue below.

    \begin{itemize}
        \item Observe that $\int_{u \in \mathbb{R}} \lambda^{(c)}(u) = \hat{\lambda}$ and $\int_{u \in \mathbb{R}} \delta^{(c)}(u) = 1$
        \item The second constraint is satisfied as $\nu^{(c)}(\vect{v},.)$ is a solution of Equation \eqref{eqn:diff_eqn_nu}.
        \item Bounds $U(\vect{v})$ and $L(\vect{v})$ exist from statement in Lemma \ref{lemma:diff_eqn_soln}. The proof of continuity of $U(.)$ and $L(.)$ is slightly technical and we formally prove this in Claim \ref{claim:continuity_bounds}. Observe that all the assumptions in Claim \ref{claim:continuity_bounds} is satisfied due to results from Lemma \ref{lemma:diff_eqn_soln} and Lemma \ref{diff_eqn_soln_first_lemma}.
        
    \end{itemize}

    Now observe that the objective value of this feasible solution in $\hat{\lambda}$ and the constructed solution is feasible for any $\hat{\lambda} \leq \frac{\epsilon-\zeta}{\epsilon+\zeta} \hat{f}(\epsilon+\zeta,k)$ and $\zeta>0$. Now, since $\hat{f}(\epsilon,k)$ is continuous in $\epsilon$, choosing $\zeta$ to be arbitrarily small enables us to get obtain the objective value of the solution arbitrarily close to $\hat{f}(\epsilon,k)$ and thus, $\text{opt}(\mathcal{E}) \geq \hat{f}(\epsilon,k)$.
    
\end{proof}

\subsection{Optimal Result Selection given Laplace noise when $\mathfrak{h}(.)$ is an identity function}{\label{subsec-appendix:server_response}}



\begin{theorem} [corresponds to Theorem \ref{thm:lap:sim}]  \label{cor:geofinal}

Recall the definition of $\hat{f}(\epsilon,k)$ from Equation \eqref{eq:hat_f_defn} and suppose $A^{*}= \{x_1,x_2,\dots,x_{\K}\}$ where $x_1\leq x_2\leq \dots\leq x_{\K}$ minimises the same.\footnote{Here, we denote $A^{*}$ by a set instead of a vector.} Then, we have  


\begin{enumerate}
    \item When $\K$ is odd, $x_{(\K+1)/2}=0$. For $i\in [(\K-1)/2]$, $x_i=-2\log((\K+3)/2-i)/\epsilon$. For $i\in [\K]\setminus [(\K+1)/2]$, $x_i=-x_{\K+1-i}$. That is, if $t\defeq (\K -1)/2$, then
    \begin{align*}
        \{x_1,\dots,x_{\K}\}=\left\{
        0, \pm \frac{2\log\left(\frac{t+1}{t}\right)}{\epsilon},
        \pm \frac{2\log\left(\frac{t+1}{t-1}\right)}{\epsilon},
        \dots,
        \pm \frac{2\log\left(\frac{t+1}{2}\right)}{\epsilon},
        \pm \frac{2\log\left(t+1\right)}{\epsilon}
        \right\}.
    \end{align*}
    \item When $\K$ is even, $x_{\K/2}=-\log(1+2/\K)/\epsilon$. For $i\in [\K/2-1]$, $x_i=x_{i+1}-(2\log(1+1/i))/\epsilon$. For $i\in [\K]\setminus [\K /2]$, $x_i=-x_{\K+1-i}$. That is, if $t\defeq \K /2$, then
    \begin{align*}
        \{x_1,\dots,x_{\K}\}=\left\{
         \pm \frac{\log\left(\frac{t+1}{t}\right)}{\epsilon},
        \pm \frac{\log\left(\frac{t(t+1)}{(t-1)^2}\right)}{\epsilon},
        \pm \frac{\log\left(\frac{t(t+1)}{(t-2)^2}\right)}{\epsilon},
        \dots,
        \pm \frac{\log\left(\frac{t(t+1)}{2^2}\right)}{\epsilon},
        \pm \frac{\log\left(t(t+1)\right)}{\epsilon}
        \right\}.
    \end{align*}
\end{enumerate}
\end{theorem}
We prove this theorem using the following lemmas.

\begin{lemma}   \label{lem:median}

If $\min\limits_{A\in \Kads}  \Exxlimits{x\sim \Lap}{ \min\limits_{a\in A}\disrzero{a-x} }= \Exxlimits{x\sim \Lap}{ \min\limits_{a\in A^{*}}\disrzero{a-x} }$ and $A^{*}=\{x_1,x_2,\dots,x_{\K}\}$ where $x_1\leq x_2\leq \dots\leq x_{\K}$, we define $y_i\defeq (x_i+x_{i+1})/2$ for $i\in [\K -1]$, $y_0\defeq -\infty$, and $y_{\K}\defeq +\infty$. Then,
\begin{align} \label{eq:aveopt}
    \Prxx{z\sim \Lap}{  y_{i-1}<z<x_{i}} = \Prxx{z\sim \Lap}{  x_{i}<z<y_{i}}, \quad \forall~~i\in [\K].
\end{align}

\end{lemma}

\begin{proof}

Note that for any $i\in [\K]$, for any $z\in (y_{i-1},y_{i})$, $\min_{a\in A^{*}}\disrzero{a-x} = \disrzero{x_i-x}$. Then we have
\begin{align*}
    \Exx{x\sim \Lap}{ \min_{a\in A^{*}}\disrzero{a-x} }
    =
    \sum_{i=1}^{\K} \Prxx{z\sim \Lap}{y_{i-1}<z<y_{i}} \Exx{z\sim \Lap}{ \disrzero{x_i-z} | y_{i-1}<z<y_{i}}. 
\end{align*}
If $x_i\neq \Exx{z\sim \Lap}{ z | y_{i-1}<z<y_{i}}$ for a certain $i\in [\K]$, we can change the value of $x_i$ so that Equation \ref{eq:aveopt} holds, in which $x_i$ is somehow a median of $\Lap$ inside the interval $(y_{i-1},y_i)$, and then $\Exx{z\sim \Lap}{ \disrzero{x_i-z} | y_{i-1}<z<y_{i}}$ strictly decreases, which implies $\Exx{x\sim \Lap}{ \min_{a\in A^{*}}\disrzero{a-x} }$ decreases. However, $A^{*}$ is the optimal point, so it is a contradiction. As a result, for $i\in [\K]$, we have Equations \ref{eq:aveopt}.

\end{proof}

\begin{lemma}   \label{lem:aveadj}

Using the same definition of $x_i$ and $y_i$ as Lemma \ref{lem:median}, we have
\begin{enumerate}
    \item for $i\in [\K -1]$, if $y_i\leq 0$, we have $x_i=y_i-\frac{\log(1+1/i)}{\epsilon}$,
    \item for $i\in [\K]\setminus \{1\}$, if $y_{i-1}\geq 0$, we have $x_i=y_{i-1}+\frac{\log(1+1/(\K-i+1))}{\epsilon}$.
\end{enumerate}

\end{lemma}

\begin{proof}

We first show $x_i=y_i-\frac{\log(1+1/i)}{\epsilon}$ when $y_i\leq 0$. We prove it by induction. When $i=1$, by Equation \ref{eq:aveopt}, we have $\Prxx{z\sim \Lap}{  z<x_{1}} = \Prxx{z\sim \Lap}{  x_{1}<z<y_{1}}$. If $y_1\leq 0$, we have
\begin{align*}
    \int_{-\infty}^{x_{1}} e^{\epsilon z} dz = \int_{x_{1}}^{y_{1}} e^{\epsilon z} dz,
\end{align*}
which is equivalent to
\begin{align*}
    e^{\epsilon x_1} = e^{\epsilon y_1} - e^{\epsilon x_1},
\end{align*}
so $y_1=x_1+\log(2)/\epsilon$.

If $y_i\leq 0$, assuming $x_{i-1}=y_{i-1}-\frac{\log(1+1/(i-1))}{\epsilon}$, then $x_i=y_{i-1}+\frac{\log(1+1/(i-1))}{\epsilon}$. By Equation \ref{eq:aveopt}, we have
\begin{align*}
    \int_{y_{i-1}}^{x_i} e^{\epsilon z} dz = \int_{x_{i}}^{y_{i}} e^{\epsilon z} dz,
\end{align*}
which is equivalent to
\begin{align*}
    \exp{(\epsilon x_{i})} - \exp{\left(\epsilon x_{i} - \frac{\log(1+1/(i-1))}{\epsilon}\right)} = \exp{(\epsilon y_{i})} - \exp{(\epsilon x_{i})},
\end{align*}
so $y_i=x_i+\frac{\log(1+1/i)}{\epsilon}$.

As a result, for $i\in [\K -1]$, if $y_i\leq 0$, we have $x_i=y_i-\frac{\log(1+1/i)}{\epsilon}$.

By symmetry, using similar arguments, we know that for $i\in [\K]\setminus \{1\}$, if $y_{i-1}\geq 0$, then $x_i=y_{i-1}+\frac{\log(1+1/(\K-i+1))}{\epsilon}$.

\end{proof}

\begin{lemma}   \label{lem:geobalance}

Using the same definition of $x_i$ and $y_i$ as Lemma \ref{lem:median}, if $y_i\leq 0\leq y_{i+1}$ where $i\in [\K -2]$, then $|i-(\K-i)|\leq 1$.

\end{lemma}

\begin{proof}

By Lemma \ref{lem:aveadj}, we have
\begin{align*}
    \begin{cases}
    x_i=y_i-\frac{\log(1+1/i)}{\epsilon},    \\
    x_{i+2}=y_{i+1}+\frac{\log(1+1/(\K-i-1))}{\epsilon},
    \end{cases}
\end{align*}
which implies
\begin{align*}
    y_i+\frac{\log(1+1/i)}{\epsilon} = x_{i+1} = y_{i+1}-\frac{\log(1+1/(\K-i-1))}{\epsilon}.
\end{align*}

Since $y_i\leq 0\leq y_{i+1}$, we have $x_{i+1}\in \left[\beta_1,\beta_2 \right]$,
where
\begin{align*}
    \beta_1 \defeq -\frac{\log(1+1/(\K-i-1))}{\epsilon},
    \quad
    \beta_2 \defeq \frac{\log(1+1/i)}{\epsilon}.
\end{align*}

Define 
\begin{align*}
    g_1(x)\defeq \int_{x-\frac{\log(1+1/i)}{\epsilon}}^{x} e^{-\epsilon |z|} dz,~~\text{and}~~
    g_2(x) \defeq  \int_{x}^{x+\frac{\log(1+1/(\K-i-1))}{\epsilon}} e^{-\epsilon |z|} dz,
\end{align*}
and define $h(x)\defeq g_1(x)/g_2(x)$. By Lemma \ref{lem:median}, we have $g_1(x_{i+1})=g_2(x_{i+1})$, which implies $h(x_{i+1})=1$.

For $\beta_1 \leq y_1<y_2\leq 0$, we have $g_1(y_1)= e^{y_1-y_2} g_1(y_2)$, and $g_2(y_1)> e^{y_1-y_2} g_2(y_2)$, so $h(y_1)< h(y_2)$. Similarly, for $0\leq y_1<y_2\leq \beta_2$, we have $g_1(y_1)< e^{y_2-y_1} g_1(y_2)$, and $g_2(y_1)= e^{y_2-y_1} g_2(y_2)$, so $h(y_1)< h(y_2)$. As a result, $h(x)$ is strictly increasing on $[\beta_1,\beta_2]$.

If $i>\K-i+1$, $h(\beta_2)<1$. If $i+1<\K-i$, $h(\beta_1)>1$. In both cases, for any $x\in [\beta_1,\beta_2]$, $h(x)\neq 1$. It contradicts to $h(x_{i+1})=1$ and $x_{i+1}\in [\beta_1,\beta_2]$. As a result, we have $|i-(\K -i)|\leq 1$.

\end{proof}

Combining Lemmas \ref{lem:aveadj} and \ref{lem:geobalance}, we have the following corollary \ref{cor:geofinal}.

And thus, Theorem \ref{thm:main1} is proved.

\renewcommand{\zhihao}[1]{}
\renewcommand{\sss}[1]{ }







\subsection{Proof of weak duality as stated in Theorem \ref{theorem:weak_duality_result}}{\label{sec:weak_duality}}

We now restate the weak duality theorem from Theorem \ref{theorem:weak_duality_result}.

\begin{restatable} 
[proof of Theorem \ref{theorem:weak_duality_result}] {theorem*}{simrstepa}
    $\text{opt}(\mathcal{O}) \geq \text{opt}(\mathcal{E})$.
\end{restatable}


To prove this lemma, we now redefine the optimization problems $\mathcal{O}$ and $\mathcal{E}$ in \eqref{orig_primal} and \eqref{orig_dual}. 

\zhihao{We optimize the functions to be rieman integrable, is it on any bounded set, or it also works for unbounded set?}

\sss{Observe that since $g(.)$ is non-negative, Reimann integral on a bounded set would imply the existence of integral on an unbounded set too, the integral may go to infinity but is always defined.}





\begin{equation}
\text{ $\mathcal{O}$ =}
\left\{
\begin{aligned}
    \inf_{g(.,.): \mathcal{I}^B (\mathbb{R} \times \mathbb{R}^k \rightarrow \mathbb{R}^{+}), \kappa \in \mathbb{R}} \quad & \kappa \\
    \textrm{s.t.} \quad & \kappa - \int\limits_{\vect{x} \in \mathbb{R}^k} \left[\min_{a \in \texttt{Set}(\vect{x})} \mathfrak{h}(|u-a|) \right] g(u,\vect{x}) d\left(\prod_{i=1}^{k} x_i\right) \geq 0 \text{ }\forall u\in \mathbb{R}\\
    & \int\limits_{\vect{x} \in \mathbb{R}^k} g(u,\vect{x}) d\left(\prod_{i=1}^{k} x_i\right) = 1 \text{ } \forall u\in \mathbb{R}\\
    &  \epsilon g(u,\vect{x}) + \underline{g}_{u}(u,\vect{x}) \geq 0; \text{ }\forall u\in \mathbb{R}; \vect{x} \in \mathbb{R}^k\\
    &  \epsilon g(u,\vect{x}) - \overline{g}_{u}(u,\vect{x}) \geq 0; \text{ }\forall u\in \mathbb{R}; \vect{x} \in \mathbb{R}^k
\end{aligned}
\right.
\end{equation}

\begin{equation}
\text{ {$\mathcal{E}$ =}}
\left\{
\begin{aligned}
    \sup_{\substack{{\nu}(.,.): \mathcal{C}^1 (\mathbb{R} \times \mathbb{R}^k \rightarrow \mathbb{R}) \\ 
    \lambda(.): \mathcal{C}^0(\mathbb{R} \rightarrow \mathbb{R}^{+})\\
    \delta(.): \mathcal{C}^0(\mathbb{R} \rightarrow \mathbb{R}^{+})}} & \int_{u\in \mathbb{R} } {\lambda}(u) du\\
    \textrm{ s.t. } \quad & \int_{u\in \mathbb{R}} {\delta}(u) du \leq 1\\
    & \hspace{-3 em} - \left[\min_{a \in \texttt{Set}(\vect{x})} \mathfrak{h}(|u-a|) \right] {\delta}(u) + {\lambda}(u) + \epsilon |{\nu}(u,\vect{x})| + {\nu}_u(u,\vect{x})   \leq 0 \text{ } \forall u\in \mathbb{R}, \forall \vect{x} \in \mathbb{R}^k \\
     & \exists U: \mathcal{C}^0(\mathbb{R}^k \rightarrow \mathbb{R})  \textrm { s.t. }\nu(u,\vect{x}) \geq 0 \text{ }\forall u\geq U(\vect{x}) \text{ } \forall \vect{x} \in (\mathbb{R})^k \\
    & \exists L: \mathcal{C}^0(\mathbb{R}^k \rightarrow \mathbb{R}) \textrm { s.t. } \nu(u,\vect{x}) \leq 0 \text{ } \forall u\leq L(\vect{x}) \text{ } \forall \vect{x} \in (\mathbb{R})^k \\    
\end{aligned}
\right.    
\end{equation}

We now define a DILP $\mathcal{E}^{int}$ which effectively splits the function $\nu(.,.)$ into a negative and a positive part and prove a lemma \ref{lemma1} bounding the optimal value of $\mathcal{E}$ by $\mathcal{E}^{int}$.

\begin{equation}
\text{ {$\mathcal{E}^{int}$ =}}
\left\{
\begin{aligned}
    \sup_{\substack{\delta(.): \mathcal{C}^0 (\mathbb{R} \rightarrow \mathbb{R}^{+}) ,{\nu^{(1)}}(.,.): \mathcal{C}^1 (\mathbb{R} \times \mathbb{R}^k \rightarrow \mathbb{R}^{+}) \\ 
    {\nu^{(2)}}(.,.): \mathcal{C}^1 (\mathbb{R} \times \mathbb{R}^k \rightarrow \mathbb{R}^{+}), 
    \lambda(.): \mathcal{C}^0(\mathbb{R} \rightarrow \mathbb{R}^{+})}} & \int_{u\in \mathbb{R} } {\lambda}(u) du \\
    \textrm{ s.t. } \quad & \int_{u\in \mathbb{R}} {\delta}(u) du\leq 1\\
    & \hspace{-2 em} - \left[\min_{a \in \texttt{Set}(\vect{x})} \mathfrak{h}(|u-a|) \right] {\delta}(u) + {\lambda}(u) + \epsilon \left(-{\nu}^{(1)}(u,\vect{x}) +{\nu}^{(2)}(u,\vect{x})\right)\\ 
    & \hspace{5em} + {\nu}^{(1)}_u(u,\vect{x}) +{\nu}^{(2)}_u(u,\vect{x}) \leq 0 \text{    } \forall u\in \mathbb{R}, \forall \vect{x} \in \mathbb{R}^k \\
    & \hspace{-2 em}\exists U: \mathcal{C}^0(\mathbb{R}^k \rightarrow \mathbb{R})  \textrm { s.t. }-\nu^{(1)}(u,\vect{x}) + \nu^{(2)}(u,\vect{x}) \geq 0 \text{ }\forall u\geq U(\vect{x}) \text{ } \forall \vect{x} \in (\mathbb{R})^k \\
    & \hspace{-2 em}\exists L: \mathcal{C}^0(\mathbb{R}^k \rightarrow \mathbb{R}) \textrm { s.t. } -\nu^{(1)}(u,\vect{x}) + \nu^{(2)}(u,\vect{x}) \leq 0 \text{ } \forall u\leq L(\vect{x}) \text{ } \forall \vect{x} \in (\mathbb{R})^k
\end{aligned}
\right.    
\end{equation}

We now prove the following set of lemmas and denote the optimal value of DILP given by $\mathcal{O}$ as $\text{opt}(\mathcal{O})$

\begin{lemma}{\label{lemma1}}
    $\text{opt}(\mathcal{E}^{int}) \geq \text{opt}(\mathcal{E})$
\end{lemma}

\begin{proof}
    This proof follows by restricting exactly one of the values $\nu^{(1)}(u,\vect{x})$ or $\nu^{(2)}(u,\vect{x})$ = 0. In this case, we may define $\nu(u,\vect{x}) = -\nu^{(1)}(u,\vect{x}) + \nu^{(2)}(u,\vect{x})$ and thus, $|\nu(u,\vect{x})| = \nu^{(1)}(u,\vect{x}) + \nu^{(2)}(u,\vect{x})$. Since, we restrict the optimization variables i.e. add an extra constraint, we get $\text{opt}(\mathcal{E}^{int}) \geq \text{opt}(\mathcal{E})$
\end{proof}


\begin{lemma}{\label{lemma2}}
    $\text{opt}(\mathcal{O}) \geq \text{opt}(\mathcal{E}^{int})$.
\end{lemma}

Before we start the proof, we give some key observations which allow us to prove the weak duality. Observe that the constraint of DILP $\mathcal{O}$ involves a linear constraint of derivative of $g(.,.)$. We use integration by parts in this proof to convert this into a constraint on the derivative of the dual variable $\nu(.)$ in the dual DILP $\mathcal{E}$. We state the equation for this below. \footnote{For sake of brevity, we give the inequality using derivatives and not upper and lower derivatives. A formal version of this is given in equation \eqref{integral_by_parts}.} 
\begin{align}
    & \int_{u \in \mathbb{R}} (\nu^{(1)}(u,\vect{x}) -\nu^{(2)}(u,\vect{x}) )g_u(u,\vect{x}) du \\
    = & \left[(\nu^{(1)}(u,\vect{x}) -\nu^{(2)}(u,\vect{x})) g(u,\vect{x})\right]_{u=-\infty}^{+\infty} - \int_{u \in \mathbb{R}} (\nu^{(1)}_u(u,\vect{x}) -\nu^{(2)}_u(u,\vect{x}) ) g(u,\vect{x}) du\\
    \leq & - \int_{u \in \mathbb{R}} (\nu^{(1)}(u,\vect{x}) -\nu^{(2)}(u,\vect{x}) ) g(u,\vect{x})
\end{align}

The last step follows from the last two constraints in DILP $\mathcal{E}$ i.e.  $\lim_{u \to \infty} \nu^{(1)}(u,\vect{x}) - \nu^{(2)}(u,\vect{x}) \leq 0$ and  $\lim_{u \to -\infty} \nu^{(1)}(u,\vect{x}) - \nu^{(2)}(u,\vect{x}) \geq 0$


We now move to the proof of the weak duality. Notice that this proof \ref{subsec-proof:formal_proof} bears resemblance to the weak duality proof, albeit with a clever utilization of integration by parts, Fatou's lemma, and the monotone convergence theorem. This approach allows for the exchange of limits and integrals in a strategic manner. 

To get an intuition on the proof steps, we first present an informal proof \ref{subsec-proof:informal_proof} where we assume functions are always integrable and exchange limits and integrals without giving explicit reasons. To look into its formal treatment refer to Section \ref{subsec-proof:formal_proof}

\subsubsection{Informal proof to Theorem \ref{theorem:weak_duality_result}}\label{subsec-proof:informal_proof}

As discussed above, in this proof we just give an intuition for the steps to gain an understanding without giving formal reasons for exchange of integrals and limits. We further assume that the lower and upper derivatives are integrable in this part. To look into its formal treatment refer to Proof in Section \ref{subsec-proof:formal_proof}.
\begin{proof}[Informal proof]

    Now consider any feasible solution $g(.,.)$ and $\kappa$ in the primal DILP $\mathcal{O}$ and feasible solution $\nu^{(1)}(.,.), \nu^{(2)}(.,.), \lambda(.) \text{ and } \delta(.)$ in the dual DILP $\mathcal{E}^{int}$.


    Now pre-multiply the first constraint in DILP $\mathcal{O}$ by $\delta(u)$ second constraint in DILP by $\lambda(u)$, the third constraint in DILP by $\nu^{(1)}(u,\vect{x})$ and fourth constraint in DILP $\mathcal{O}$ by $\nu^{(2)}(u,\vect{x})$

    \begin{align}
        & \int\limits_{u \in \mathbb{R}} \delta(u) \left[\kappa - \int\limits_{\vect{x} \in \mathbb{R}^k} \left(\min_{a \in \texttt{Set}(\vect{x})} \mathfrak{h}(|u-a|) \right) g(u,\vect{x}) d\left(\prod_{i=1}^{k} x_i\right) \right]du + \int\limits_{u \in \mathbb{R}} \int\limits_{\vect{x} \in \mathbb{R}^k} g(u,\vect{x}) \lambda(u) d\left(\prod_{i=1}^{k} x_i\right) du\nonumber\\
        & \hspace{-6 em}+ \int\limits_{u \in \mathbb{R}} \int\limits_{\vect{x} \in\mathbb{R}^k} (\epsilon g(u,\vect{x}) + \underline{g}_u(u,\vect{x})) \nu^{(1)}(u,\vect{x}) d\left(\prod_{i=1}^{k} x_i\right) du + \int\limits_{u \in \mathbb{R}} \int\limits_{\vect{x} \in\mathbb{R}^k} (\epsilon g(u,\vect{x}) -\overline{g}_u(u,\vect{x})) \nu^{(2)}(u,\vect{x}) d\left(\prod_{i=1}^{k} x_i\right) du \geq \int\limits_{u \in \mathbb{R}} \lambda(u) du\\
        \overset{(e)}{\implies} & \kappa  + \int\limits_{\vect{x} \in \mathbb{R}^k} \left[\int\limits_{u \in \mathbb{R}} \left(-\left[\min_{a \in \texttt{Set}(\vect{x})} \mathfrak{h}(|u-a|) \right] \delta(u) + \lambda(u) + \epsilon \nu^{(2)}(u,\vect{x})  + \epsilon \nu^{(1)}(u,\vect{x})\right)du \right] g(u,\vect{x}) d\left(\prod_{i=1}^{k} x_i\right) \nonumber\\
        & + \int_{\vect{x} \in \mathbb{R}^k} \left[\int_{u \in \mathbb{R}} \left(\nu^{(1)}(u,\vect{x}) (u,\vect{x}) \underline{g}_u(u,\vect{x}) - \nu^{(2)}(u,\vect{x}) \overline{g}_u(u,\vect{x})\right)  du \right] d\left(\prod_{i=1}^{k} x_i\right) \geq \int\limits_{u \in \mathbb{R}} \lambda(u) du \label{first_part_step}
    \end{align}

$(e)$ follows by exchanging integrals since each term is positive and the constraint$\int_{u} \delta(u) \leq 1$.

Now, we solve the for the term $\int_{u \in \mathbb{R}} \left[\nu^{(1)}(u,\vect{x}) \underline{g}_u(u,\vect{x}) du- \nu^{(2)}(u,\vect{x}) \overline{g}_u(u,\vect{x})\right] du$

\begin{align}
    & \int_{u \in \mathbb{R}} \left[\nu^{(1)}(u,\vect{x}) \underline{g}_u(u,\vect{x}) du- \nu^{(2)}(u,\vect{x}) \overline{g}_u(u,\vect{x})\right] du\\
    & \overset{(a)}{\leq} \left[\left(\nu^{(1)}(u,\vect{x}) - \nu^{(2)}(u,\vect{x})\right) g(u,\vect{x}) \right]_{u=-\infty}^{+\infty} - \left[\int\limits_{u \in \mathbb{R}} (\nu^{(1)}_u(u,\vect{x})-\nu^{(2)}_u(u,\vect{x})) {g}(u,\vect{x}) du\right]\\
    & \overset{(b)}{\leq} -\left[\int\limits_{u \in \mathbb{R}} \left(\nu^{(1)}_u(u,\vect{x}) - \nu^{(2)}_u(u,\vect{x})\right) {g}(u,\vect{x}) du\right] \label{integral_by_parts}
\end{align}

$(a)$ follows from integration of parts and the inequality from the fact that we take lower derivative and upper derivative of $g(.)$ respectively. $(b)$ follows from the third and fourth constraints in the LP $D^{int}$ on $\nu^{(1)} (u,\vect{x}) - \nu^{(2)} (u,\vect{x})$ as $u \to \infty$ or $u \to -\infty$.

Combining the inequalities in \eqref{first_part_step} and \eqref{integral_by_parts}, we get,

\begin{align}
    \implies & \kappa  + \int\limits_{\vect{x} \in \mathbb{R}^k} \left[\int\limits_{u \in \mathbb{R}} \left(\left[\min_{a \in \texttt{Set}(\vect{x})} \mathfrak{h}(|u-a|) \right] \delta(u) + \lambda(u) + \epsilon \nu^{(2)}(u,\vect{x})  + \epsilon \nu^{(1)}(u,\vect{x})\right)du \right] g(u,\vect{x}) d\left(\prod_{i=1}^{k} x_i\right) \nonumber\\
        &  -\int_{\vect{x} \in \mathbb{R}^k}\left[\int\limits_{u \in \mathbb{R}} \left(\nu^{(1)}_u(u,\vect{x}) - \nu^{(2)}_u(u,\vect{x})\right) {g}(u,\vect{x}) du\right]d\left(\prod_{i=1}^{k} x_i\right) \geq \int\limits_{u \in \mathbb{R}} \lambda(u) du \\
    \implies & \kappa + \int_{\vect{x} \in \mathbb{R}^k} \Biggl[ \int\limits_{u \in \mathbb{R}} \Biggl(\left(\min_{a \in \texttt{Set}(\vect{x})} \mathfrak{h}(|u-a|) \right) \delta(u) + \lambda(u) + \epsilon \nu^{(2)}(u,\vect{x})  + \epsilon \nu^{(1)}(u,\vect{x}) \\
    & \hspace{10 em} -\left(\nu^{(1)}_u(u,\vect{x}) - \nu^{(2)}_u(u,\vect{x})\right)\Biggr) du\Biggr]
    g(u,\vect{x}) d\left(\prod_{i=1}^{k} x_i\right) \geq \int\limits_{u \in \mathbb{R}} \lambda(u) du\\
    & \overset{(c)}{\implies} \kappa \geq \int\limits_{u \in \mathbb{R}} \lambda(u) du
\end{align}

$(c)$ follows from the first and second constraint in LP $\mathcal{E}^{int}$. Since this inequality is true for every feasible solution in the primal $\mathcal{O}$ and dual $\mathcal{E}^{int}$, we have the proof in the theorem.

\end{proof}

\subsubsection{A formal proof of Theorem \ref{theorem:weak_duality_result}}{\label{subsec-proof:formal_proof}}
\begin{proof}
    Now consider any feasible solution $g(.,.)$ and $\kappa$ in the primal DILP $\mathcal{O}$ and feasible solution $\nu^{(1)}(.,.), \nu^{(2)}(.,.), \lambda(.) \text{ and } \delta(.)$ in the dual DILP $\mathcal{E}^{int}$.


    Now pre-multiply the first constraint in DILP $\mathcal{O}$ by $\delta(u)$ second constraint in DILP by $\lambda(u)$, the third constraint in DILP by $\nu^{(1)}(u,\vect{x})$ and fourth constraint in DILP $\mathcal{O}$ by $\nu^{(2)}(u,\vect{x})$ and take the lower Riemann-Darboux integrals for the last two terms and the integral for the first two terms. \footnote{The integral may not be defined as $\underline{g}_u$ and $\overline{g}_u$ may not be Riemann integrable, however the first and second terms are integrable which follows since the product of integrable functions is integrable. }. This approach is very similar to the use of Lagrange multiplier in weak duality proof in linear programming. Also observe that some limits may take values in the extended reals i.e. $\mathbbm{R} \cup \{-\infty,\infty\}$ and thus, we define the limits in extended reals.

    \zhihao{Is $\underline{g}_u$ \textit{Reinamann integrable}?}
    \sss{New proof does not require it as we use lower \textit{Reinamann integral} which is always defined}. 
    
    Observe that $d^l,d^u \in \mathbb{R}$,$\vect{c}^l,\vect{c}^u \in \mathbb{R}^k$ and $c^l_i,c^u_i$ refers to its $i^{th}$ component of $\vect{c}^l$ and $\vect{c}^u \in \mathbb{R}^k$ respectively.

Thus, we get\footnote{$\lowint f(x) dx$ and $\upint f(x) dx$ denotes the lower and upper integral respectively.}

    \begin{align}
        & \hspace{-1.5 em}\lim_{\substack{d^l \to -\infty \\ d^u \to \infty}} \int\limits_{u = d^l}^{d^u} \delta(u) \left[\kappa - \int\limits_{\vect{x} \in \mathbb{R}^k} \left(\min_{a \in \texttt{Set}(\vect{x})} \mathfrak{h}(|u-a|) \right) g(u,\vect{x}) d\left(\prod_{i=1}^{k} x_i\right) \right]du \\
        & \hspace{18 em}+ \lim_{\substack{d^l \to -\infty \\ d^u \to \infty}} \int\limits_{u = d^l}^{d^u} \left(\int\limits_{\vect{x} \in \mathbb{R}^k } g(u,\vect{x})  d\left(\prod_{i=1}^{k} x_i\right) \right)\lambda(u) du\nonumber\\
        & \hspace{-3 em}+ \liminf_{d, \vect{c} \to \infty}\lowint_{\substack{u \in [d^l,d^u] \\ \vect{x} \in \prod\limits_{i=1}^{k}[c^l_i,c^u_i]}} \left((\epsilon g(u,\vect{x}) + \underline{g}_u(u,\vect{x})) \nu^{(1)}(u,\vect{x}) + (\epsilon g(u,\vect{x}) -\overline{g}_u(u,\vect{x})) \nu^{(2)}(u,\vect{x}) \right)d\left(\prod_{i=1}^{k} x_i\right) du \geq \int\limits_{u \in \mathbb{R}} \lambda(u) du\\
        & \hspace{-5 em} \text{ }\implies \lim_{\substack{d^l \to -\infty \\ d^u \to \infty}} \int\limits_{u =d^l}^{d^u} \delta(u) \lim_{\substack{\vect{c}^u \to \infty\\ \vect{c}^l \to -\infty}} \left[\kappa - \int\limits_{\vect{x} \in \prod\limits_{i=1}^{k}[c^l_i,c^u_i]} \left(\min_{a \in \texttt{Set}(\vect{x})} \mathfrak{h}(|u-a|) \right) g(u,\vect{x}) d\left(\prod_{i=1}^{k} x_i\right)\right]du \nonumber\\
        & \hspace{20 em} + \lim_{\substack{d^l \to -\infty \\ d^u \to \infty}} \int\limits_{u =d^l}^{d^u} \lim_{\substack{\vect{c}^u \to \infty\\ \vect{c}^l \to -\infty}} \left[\int\limits_{\vect{x} \in  \prod\limits_{i=1}^{k}[c^l_i,c^u_i]} g(u,\vect{x}) \lambda(u) d\left(\prod_{i=1}^{k} x_i\right) \right] du \nonumber\\
        & \hspace{-3 em}+ \liminf_{d, \vect{c} \to \infty}\lowint_{\substack{u \in [d^l,d^u] \\ \vect{x} \in \prod\limits_{i=1}^{k}[c^l_i,c^u_i]}} \left((\epsilon g(u,\vect{x}) + \underline{g}_u(u,\vect{x})) \nu^{(1)}(u,\vect{x}) + (\epsilon g(u,\vect{x}) -\overline{g}_u(u,\vect{x})) \nu^{(2)}(u,\vect{x}) \right)d\left(\prod_{i=1}^{k} x_i\right) du \geq \int\limits_{u \in \mathbb{R}} \lambda(u) du\\
        & \hspace{-1 em} \text{ }\overset{(c)}{\implies} \lim_{d\to \infty} \lim_{\vect{c}\to \infty} \int\limits_{u =d^l}^{d^u} \left(\kappa \delta(u) + \int\limits_{\vect{x} \in \prod\limits_{i=1}^{k}[c^l_i,c^u_i]} \left(\left(-\delta(u)  \min_{a \in \texttt{Set}(\vect{x})} \mathfrak{h}(|u-a|) + \lambda(u)\right)g(u,\vect{x}) \right)d\left(\prod_{i=1}^{k} x_i\right)\right) du \nonumber\\ 
        & \hspace{4 em} + \liminf_{d,\vect{c} \to \infty}\Biggl( \int_{\substack{u \in [d^l,d^u] \\ \vect{x} \in \prod\limits_{i=1}^{k}[c^l_i,c^u_i]}} \left(\epsilon g(u,\vect{x}) \left(\nu^{(1)}(u,\vect{x}) + \nu^{(2)}(u,\vect{x})\right)\right) d\left(\prod_{i=1}^{k} x_i\right) du\nonumber\\
        & \hspace{6 em}+ \lowint_{\substack{u \in [d^l,d^u] \\ \vect{x} \in \prod\limits_{i=1}^{k}[c^l_i,c^u_i]}} \left(\underline{g}_u(u,\vect{x}) \nu^{(1)}(u,\vect{x}) -\overline{g}_u(u,\vect{x}) \nu^{(2)}(u,\vect{x}) \right)d\left(\prod_{i=1}^{k} x_i\right) du\Biggr) \geq \int\limits_{u \in \mathbb{R}} \lambda(u) du\\
        & \hspace{-4.5 em}\overset{(d)}{\implies}   \kappa  + \limsup_{{d},\vect{c} \to \infty}\Biggl(\int_{\vect{x} \in \prod\limits_{i=1}^{k}[c^l_i,c^u_i]} \left[\int\limits_{u = d^l}^{d^u} \left(-\left[\min_{a \in \texttt{Set}(\vect{x})} \mathfrak{h}(|u-a|) \right] \delta(u) + \lambda(u) \right)du \right] g(u,\vect{x}) d\left(\prod_{i=1}^{k} x_i\right)\Biggr) \nonumber\\
        & + \liminf_{d,\vect{c} \to \infty} \Biggl( \int_{\substack{u \in [d^l,d^u] \\ \vect{x} \in \prod\limits_{i=1}^{k}[c^l_i,c^u_i]}} \left(\epsilon g(u,\vect{x}) \left(\nu^{(1)}(u,\vect{x}) + \nu^{(2)}(u,\vect{x})\right)\right) d\left(\prod_{i=1}^{k} x_i\right) du \nonumber\\ 
        &\hspace{4.5 em} + \lowint_{\vect{x} \in \prod\limits_{i=1}^{k}[c^l_i,c^u_i]} \left[\lowint_{u = d^l}^{d^u} \left(\nu^{(1)}(u,\vect{x}) \underline{g}_u(u,\vect{x}) - \nu^{(2)}(u,\vect{x}) \overline{g}_u(u,\vect{x})\right)  du \right] d\left(\prod_{i=1}^{k} x_i\right) \Biggr) \geq \int\limits_{u \in \mathbb{R}} \lambda(u) du \label{first_part_step}
    \end{align}


\zhihao{In (71), is $\kappa$ has a factor $\delta(u)$?}
\sss{Was a typo addressed now}

$(c)$ \footnote{$d, \to \infty$ is a shorthand notation for $d^u \to \infty$, $d^l \to -\infty$ and $\vect{c} \to \infty$ is a shorthand for $\vect{c}^u \to \infty$ and $\vect{c}^l \to -\infty$ } follows from the following observations.

\begin{itemize}
    \item For the first two integrals, apply monotone convergence theorem (MCT)  which allows us to exchange limit and integral. This is possible since the expression in $[.]$ is positive and non-decreasing as $\vect{c}^u \to \infty$ and $\vect{c}^l \to -\infty$ due to first constraint of DILP $\mathcal{O}$ and the fact that the $g(.,.)$ is non-negative.
\end{itemize}




\zhihao{For limit exchange, $g(.,.)$ is non-negative, but the expression in [.] may not be non-negative}

\sss{No, it is non-negative check condition 1 and condition 2 of DILP $\mathcal{O}$ and use the fact that g(.) and $\mathfrak{h}(.)$ is non-negative.}

$(d)$ follows from the following steps.

\begin{itemize}
    \item Apply the \textit{interated integral} inequality \cite[Proposition 3.9]{integral} that $\lowint_{A \times B} g(A,B) dA dB \leq \lowint_{A} \left[\lowint_{B} g(A,B) dB\right]dA$ for the last integral.

    \item We exchange integrals for the first term using \textit{Fubini's theorem} \cite[Theorem 3.10]{integral} as the functions are \textit{Riemann integrable} and use the fact that $\int_{u \in \mathbbm{R}} \delta(u) \leq 1$.
    
    \item We upper bound the iterated limits by a limsup as $\limsup_{m,n} a_{m,n} \geq \lim_m \lim_n a_{m,n}$.
    
\end{itemize}

\zhihao{It seems like, Proposition 3.9 only works for integral on bounded set. On the top of Page 19 in \cite{integral}, it says \emph{Let A' be a bounded subset of $R^m$ and A'' be a bounded subset of $R^n$}}

\sss{Addressed it I put limits outside and used Monotone Convergence Theorem and Fatou's lemma to exchange integrals and limits when needed.}

Now, we bound the term $\lowint_{u \in [d^l,d^u]} \left[\nu^{(1)}(u,\vect{x}) \underline{g}_u(u,\vect{x}) du- \nu^{(2)}(u,\vect{x}) \overline{g}_u(u,\vect{x})\right] du$.\footnote{Note that the second integral is defined as the product of \textit{Riemann integrable} functions is also \textit{integrable}}. Observe that,
\begin{align}
    & \lowint_{u \in  [d^l,d^u]} \left[\nu^{(1)}(u,\vect{x}) \underline{g}_u(u,\vect{x}) du- \nu^{(2)}(u,\vect{x}) \overline{g}_u(u,\vect{x})\right] du +  \int_{u \in [d^l,d^u]} \left(\nu^{(1)}_u(u,\vect{x}) - \nu^{(2)}_u(u,\vect{x})\right) {g}(u,\vect{x}) du\nonumber\\
    \overset{(f)}{\leq} & \left[\left(\nu^{(1)}(u,\vect{x}) - \nu^{(2)}(u,\vect{x})\right) g(u,\vect{x}) \right]_{u=d^l}^{d^u} \label{integral_by_parts}
\end{align}

\zhihao{why is (\ref{integral_by_parts})? And is the right side always zero? (It is $\left(\nu^{(1)}(u,\vect{x}) - \nu^{(2)}(u,\vect{x})\right) g(u,\vect{x})$ minus itself)}

\sss{Sorry it was a major typo addressed now}

$(f)$ follows from the chain rule of differentiation and the inequality from the fact that we take lower derivative of $-g(.)$ is upper derivative of $g(.)$

\zhihao{Is limsup/liminf sufficient for this argument?} 

\sss{Addressed it now}




Now consider (observe we take the -ve of the last term in Equation \eqref{integral_by_parts}). 

\begin{align}{\label{eq_lower_upper_limit_init}}
    \hspace{-5 em}& \limsup_{d,\vect{c} \to \infty} \left[\int_{\vect{x} \in \prod\limits_{i=1}^{k}[c^l_i,c^u_i]}-\left[\left(\nu^{(1)}(u,\vect{x}) - \nu^{(2)}(u,\vect{x})\right) g(u,\vect{x}) \right]_{u=d^l}^{d^u} d\left(\prod_{i=1}^{k} x_i\right)\right]\nonumber\\
    \overset{(i)}{\geq} &  \limsup_{\vect{c} \to \infty } \liminf_{d \to \infty}\left[\int_{\vect{x} \in \prod\limits_{i=1}^{k}[c^l_i,c^u_i]}\left(-\left[\left(\nu^{(1)}(u,\vect{x}) - \nu^{(2)}(u,\vect{x})\right) g(u,\vect{x}) \right]_{u=d^l}^{d^u} \right) d\left(\prod_{i=1}^{k} x_i\right) \right]\nonumber\\
    \overset{(g)}{\geq} & \limsup_{\vect{c} \to \infty} \int_{\vect{x} \in \prod\limits_{i=1}^{k}[c^l_i,c^u_i]} \liminf_{\substack{d^{u} \to \infty \\ d^l \to -\infty}} \left[-\left(\nu^{(1)}(u,\vect{x}) - \nu^{(2)}(u,\vect{x})\right) g(u,\vect{x}) \right]_{u=d^l}^{d^u} d\left(\prod_{i=1}^{k} x_i\right)\nonumber \\
    \hspace{-3 em} \overset{(h)}{\geq} &  \limsup_{\vect{c} \to \infty} \int_{\vect{x} \in \prod\limits_{i=1}^{k}[c^l_i,c^u_i]} \Bigl(\liminf_{d^{u} \to \infty } \left[\left(-\nu^{(1)}(d^u,\vect{x}) + \nu^{(2)}(d^u,\vect{x})\right) g(d^u,\vect{x})\right]  \\
    & \hspace{12 em}+ \liminf_{d^l \to -\infty} \left[\left(\nu^{(1)}(d^l,\vect{x}) - \nu^{(2)}(d^l,\vect{x})\right) g(d^l,\vect{x}) \right] \Bigr) d\left(\prod_{i=1}^{k} x_i\right) \nonumber\\
    & \overset{(m)}{\geq} 0
\end{align}

\zhihao{Fatou's Lemma only works for non-negative functions}

\sss{Addressed check it now, showed that it is positive for sufficiently large $d^u$ and sufficiently small $d^l$.}


\zhihao{For formulas (46)-(51), we may need to explicitly distinguish $\lim_{d\rightarrow \infty}\lim_{c\rightarrow \infty}$, $\lim_{c\rightarrow \infty}\lim_{d\rightarrow \infty}$, and $\lim_{(d,c)\rightarrow (\infty,\infty)}$, and specify conditions when changing the limit type. There is $\lim_{d\rightarrow \infty}\lim_{c\rightarrow \infty}$ in (46), and there is $\lim_{c\rightarrow \infty}\lim_{d\rightarrow \infty}$ in (51), but it seems like we didn't provide reason for exchanging limit. (We are using $\lim_{(d,c)\rightarrow (\infty,\infty)}$ in the middle, without justification) }

\sss{Provided a reason for the limit exchange in $(i)$ This directly follows from standard limit-exchange theorem as stated. In (46), we used the fact that $\limsup_{m,n} f(m,n) \geq \liminf_{m} \lim_{n} f(m,n)$}

\zhihao{(I'm not confident to say this) In (51), right side of inequality (h), the integral may not be well defined. For example, is it possible that $\liminf_{d^{u} \to \infty } \left[\left(-\nu^{(1)}(d^u,\vect{x}) + \nu^{(2)}(d^u,\vect{x})\right) g(d^u,\vect{x})\right] = +\infty$, while $\liminf_{d^l \to -\infty} \left[\left(\nu^{(1)}(d^l,\vect{x}) - \nu^{(2)}(d^l,\vect{x})\right) g(d^l,\vect{x}) \right] = +\infty$, meaning the integral is not absolutely convergent   }


\sss{You are absolutely right, however the limits could go to $\infty$ and integral may not be absolutely convergent but still the limit is always defined over extended reals where limits can take value from $\mathbbm{R}\cup \{\infty,-\infty\}$ check out this wiki link \url{https://en.wikipedia.org/wiki/Extended_real_number_line}. For example read up on this how limits are defined for divergent sequences here \url{https://connect.gonzaga.edu/asset/file/1502/ch02_3_infinite_limits_v00.html}. To conclude even if both limits go to $+\infty$, as per extended real system $\infty+ \infty>0$.}

\zhihao{Follow up: in (56), is it possible that the left side is $\kappa - \infty + \infty$?}

\sss{Made some changes in the structuring of the proof. Now both terms in (58) are lower bounded. First limit in (58) + $\kappa$ is lower bounded by the sum of first integral in (48) which is postive from the first constraint in DILP $\mathcal{O}$. Second limit in (58) is positive as lower bounded by second limit of (46) which is further non-negative due to constraint of DILP $\mathcal{O}$.}



$(i)$ follows from the fact that $\limsup_{m,n} f(m,n) \geq \limsup_m \limsup_n f(m,n) \geq \limsup_m \liminf_n f(m,n)$.

$(g)$ follows from the following arguments.

\begin{itemize}
    \item Observe that $\left[-\left(\nu^{(1)}(u,\vect{x}) - \nu^{(2)}(u,\vect{x})\right) g(u,\vect{x}) \right]_{u=d^l}^{d^u}$ is non-negative for every $\vect{x} \in \prod\limits_{i=1}^{k}[c^l_i,c^u_i]$ whenever, $d_u > \sup\{U(\vect{x})| \vect{x} \in \prod\limits_{i=1}^{k}[c^l_i,c^u_i]\}$ and $d_l < \inf\{L({\vect{x}})| \vect{x} \in \prod\limits_{i=1}^{k}[c^l_i,c^u_i]\}$  from third and fourth constraint of DILP $\mathcal{E}^{int}$. Note that value is finite as supremeum and infimum of \textit{continuous functions} is finite over a \textit{compact set} since $\prod\limits_{i=1}^{k}[c^l_i,c^u_i]$ is  closed and bounded. \zhihao{why it is compact?}\sss{$\prod\limits_{i=1}^{k}[c^l_i,c^u_i]$ is compact as it is closed and bounded subset. $\{c_i\}$ is a finite real number inside the limit.}
    \item Now apply \textit{Fatou's lemma} to conclude $(g)$.
\end{itemize}

$(h)$ follows from the fact that $\liminf_n (A_n+B_n) \geq \liminf_n A_n + \liminf_n B_n$ and the last inequality $(m)$ follows from the following 2 statements.

\begin{itemize}
    \item $\liminf_{d^u \to \infty} \left(-\nu^{(1)}(d^u,\vect{x}) + \nu^{(2)}(d^u,\vect{x})\right) \geq 0$ and $\liminf_{d^l \to -\infty} \left(\nu^{(1)}(d^l,\vect{x}) - \nu^{(2)}(d^l,\vect{x})\right) \geq 0$ follows from third and fourth constraint of DILP $\mathcal{E}^{int}$.

    \item $g(.,.)$ is a bounded function follows from the constraint in DILP $\mathcal{O}$. 

    
\end{itemize}

\zhihao{For the last equation, we don't have $\int_u g(u,x)du = 1$}
\sss{Fixed it by just putting a constraint in DILP  that $g(.)$ is bounded.}


Observe that \eqref{eq_lower_upper_limit_init} implies (follows from $\limsup_n A_n = - \liminf_n -A_n$) that 

\begin{equation}{\label{eq_lower_upper_limit_final}}
    \liminf_{d,\vect{c} \to \infty} \left[\int_{\vect{x} \in \prod\limits_{i=1}^{k}[c^l_i,c^u_i]}\left[\left(\nu^{(1)}(u,\vect{x}) - \nu^{(2)}(u,\vect{x})\right) g(u,\vect{x}) \right]_{u=d^l}^{d^u} d\left(\prod_{i=1}^{k} x_i\right)\right] \leq 0
\end{equation}



Using inequality in \eqref{integral_by_parts}  and the fact that $\liminf_n (a_n+b_n) \leq \liminf a_n + \limsup_n b_n$, we obtain 
\begin{align}
    & \liminf_{d,\vect{c} \to \infty} \Biggl(\lowint_{\vect{x} \in \prod\limits_{i=1}^{k}[c^l_i,c^u_i]} \left[\lowint_{u = d^l}^{d^u} \left(\nu^{(1)}(u,\vect{x}) \underline{g}_u(u,\vect{x}) - \nu^{(2)}(u,\vect{x}) \overline{g}_u(u,\vect{x})\right)  du \right] d\left(\prod_{i=1}^{k} x_i\right) \nonumber \\ 
    & \hspace{6 em} + \int_{\substack{u \in [d^l,d^u] \\ \vect{x} \in \prod\limits_{i=1}^{k}[c^l_i,c^u_i]}} \left(\epsilon g(u,\vect{x}) \left(\nu^{(1)}(u,\vect{x}) + \nu^{(2)}(u,\vect{x})\right)\right) d\left(\prod_{i=1}^{k} x_i\right) \Biggr)du \\
    {\leq} & \liminf_{d,\vect{c} \to \infty} \left[\int_{\vect{x} \in \prod\limits_{i=1}^{k}[c^l_i,c^u_i]}\left[\left(\nu^{(1)}(u,\vect{x}) - \nu^{(2)}(u,\vect{x})\right) g(u,\vect{x}) \right]_{u=d^l}^{d^u} d\left(\prod_{i=1}^{k} x_i\right)\right] \nonumber\\
    & \hspace{6 em}+  \limsup_{d, \vect{c} \to \infty}\Biggr(\int_{\vect{x} \in \prod\limits_{i=1}^{k}[c^l_i,c^u_i]}\left[-\int_{u \in [d^l,d^u]} \left(\nu^{(1)}_u(u,\vect{x}) - \nu^{(2)}_u(u,\vect{x})\right) {g}(u,\vect{x}) du\right]d\left(\prod_{i=1}^{k} x_i\right) \nonumber\\
    & \hspace{10 em} + \int_{\substack{u \in [d^l,d^u] \\ \vect{x} \in \prod\limits_{i=1}^{k}[c^l_i,c^u_i]}} \left(\epsilon g(u,\vect{x}) \left(\nu^{(1)}(u,\vect{x}) + \nu^{(2)}(u,\vect{x})\right)\right) d\left(\prod_{i=1}^{k} x_i\right) du\Biggr)\\
    \overset{(k)}{\leq} &  \limsup_{d, \vect{c} \to \infty}\Biggr(\int_{\vect{x} \in \prod\limits_{i=1}^{k}[c^l_i,c^u_i]}\left[-\int_{u \in [d^l,d^u]} \left(\nu^{(1)}_u(u,\vect{x}) - \nu^{(2)}_u(u,\vect{x})\right) {g}(u,\vect{x}) du\right]d\left(\prod_{i=1}^{k} x_i\right) \nonumber\\
    & \hspace{10 em} + \int_{\substack{u \in [d^l,d^u] \\ \vect{x} \in \prod\limits_{i=1}^{k}[c^l_i,c^u_i]}} \left(\epsilon g(u,\vect{x}) \left(\nu^{(1)}(u,\vect{x}) + \nu^{(2)}(u,\vect{x})\right)\right) d\left(\prod_{i=1}^{k} x_i\right) du\Biggr)\label{temp_equation1}
\end{align}

$(k)$ follows on applying Equation \eqref{eq_lower_upper_limit_final}

Now combining \eqref{temp_equation1} and \eqref{first_part_step}, we obtain
\begin{align}
    & \hspace{-4 em}  \kappa  + \limsup_{d, \vect{c} \to \infty}\Biggl(\int_{\vect{x} \in \prod\limits_{i=1}^{k}[c^l_i,c^u_i]} \left[\int_{u \in [d^l,d^u]} \left(-\left[\min_{a \in \texttt{Set}(\vect{x})} \mathfrak{h}(|u-a|) \right] \delta(u) + \lambda(u) \right)du \right] g(u,\vect{x}) d\left(\prod_{i=1}^{k} x_i\right)\Biggr) \nonumber\\
        &  + \limsup_{d, \vect{c} \to \infty}\Biggr(\int_{\vect{x} \in \prod\limits_{i=1}^{k}[c^l_i,c^u_i]}\left[-\int_{u \in [d^l,d^u]} \left(\nu^{(1)}_u(u,\vect{x}) - \nu^{(2)}_u(u,\vect{x})\right) {g}(u,\vect{x}) du\right]d\left(\prod_{i=1}^{k} x_i\right) \nonumber\\
        & \hspace{10 em} + \int_{\substack{u \in [d^l,d^u] \\ \vect{x} \in \prod\limits_{i=1}^{k}[c^l_i,c^u_i]}} \left(\epsilon g(u,\vect{x}) \left(\nu^{(1)}(u,\vect{x}) + \nu^{(2)}(u,\vect{x})\right)\right) d\left(\prod_{i=1}^{k} x_i\right) du\Biggr) \geq \int\limits_{u \in \mathbbm{R}} \lambda(u) du \\
    & \hspace{-6 em}  \overset{(i)}{\implies} \kappa  + \lim_{d, \vect{c} \to \infty}\Biggl(\int_{\vect{x} \in \prod\limits_{i=1}^{k}[c^l_i,c^u_i]} \left[\int_{u \in [d^l,d^u]} \left(-\left[\min_{a \in \texttt{Set}(\vect{x})} \mathfrak{h}(|u-a|) \right] \delta(u) + \lambda(u)\right)du \right] g(u,\vect{x}) d\left(\prod_{i=1}^{k} x_i\right)\Biggr) \nonumber\\
        &  + \limsup_{d, \vect{c} \to \infty}\Biggr(\int_{\vect{x} \in \prod\limits_{i=1}^{k}[c^l_i,c^u_i]}\left[-\int_{u \in [d^l,d^u]} \left(\nu^{(1)}_u(u,\vect{x}) - \nu^{(2)}_u(u,\vect{x})\right) {g}(u,\vect{x}) du\right]d\left(\prod_{i=1}^{k} x_i\right) \Biggr)\nonumber \\
        & \hspace{6 em} + \int_{\vect{x} \in \prod\limits_{i=1}^{k}[c^l_i,c^u_i]} \int_{u \in [d^l,d^u]} \left(\epsilon g(u,\vect{x}) \left(\nu^{(1)}(u,\vect{x}) + \nu^{(2)}(u,\vect{x})\right)\right) d\left(\prod_{i=1}^{k} x_i\right) du\Biggr) \geq \int\limits_{u \in \mathbbm{R}} \lambda(u) du \\
    {\implies} & \kappa + \limsup_{d, \vect{c} \to \infty}\Biggl(\int_{\vect{x} \in \prod\limits_{i=1}^{k}[c^l_i,c^u_i]} \Biggl[ \int_{u \in [d^l,d^u]} \Biggl(-\left(\min_{a \in \texttt{Set}(\vect{x})} \mathfrak{h}(|u-a|) \right) \delta(u) + \lambda(u) + \epsilon \nu^{(2)}(u,\vect{x})  + \epsilon \nu^{(1)}(u,\vect{x})\Biggr) \\
    & \hspace{10 em} -\left(\nu^{(1)}_u(u,\vect{x}) - \nu^{(2)}_u(u,\vect{x})\right)\Biggr) du\Biggr]
    g(u,\vect{x}) d\left(\prod_{i=1}^{k} x_i\right)\Biggr) \geq \int\limits_{u \in \mathbb{R}} \lambda(u) du\\
    & \overset{(j)}{\implies} \kappa \geq \int\limits_{u \in \mathbb{R}} \lambda(u) du
\end{align}

\zhihao{as we discussed, (i) need another way to write}


The limsup in first term of $(i)$ can be replaced by a limit as integral \eqref{eqn:integral_temp} is defined in extended reals ($\mathbb{R} \cup \{\infty,-\infty\}$) which follows from the following reasons.\footnote{Intuitively it shows that $\infty-\infty$ scenario cannot arrive in this integral}

\begin{align}{\label{eqn:integral_temp}}
    & \Biggl(\int_{\vect{x} \in \mathbb{R}^k} \left[\int_{u \in \mathbb{R}} \left(-\left[\min_{a \in \texttt{Set}(\vect{x})} \mathfrak{h}(|u-a|) \right] \delta(u) + \lambda(u) \right)du \right] g(u,\vect{x}) d\left(\prod_{i=1}^{k} x_i\right)\Biggr)\\
    = & - \int_{\vect{x} \in \mathbb{R}^k} \int_{u \in \mathbb{R}} \left[\min_{a \in \texttt{Set}(\vect{x})} \mathfrak{h}(|u-a|) \right] \delta(u) g(u,\vect{x}) du \text{ }d\left(\prod_{i=1}^{k} x_i\right) \label{eqn:first_term}\\ 
    & \hspace{18 em}+ \int_{\vect{x} \in \mathbb{R}^k} \int_{u \in \mathbb{R}} \lambda(u) g(u,\vect{x}) du \text{ }d\left(\prod_{i=1}^{k} x_i\right) \label{eqn:second_term}
\end{align}

\begin{itemize}

    \item Observe that $\int_u \delta(u) \leq 1$ and $\int_{\vect{x} \in \mathbb{R}^k} \min\limits_{a \in \texttt{Set}(\vect{x})} \mathfrak{h}(|u-a|) g(u,\vect{x}) d\left(\prod_{i=1}^{k} x_i\right) \leq \kappa \text{ } \forall u \in \mathbb{R}$ (first constraint of DILP $\mathcal{O}$) implying that term 
    \eqref{eqn:first_term} is lower bounded by $-\kappa$.

    \item Thus, we may conclude that the following integral \eqref{eqn:integral_temp} is defined in extended reals as the first term \eqref{eqn:first_term} is lower bounded by $-\kappa$ and second term \eqref{eqn:second_term} is lower bounded by 0.

\end{itemize}


$(j)$ follows from the first and second constraint in DILP $\mathcal{E}^{int}$. Since this inequality is true for every feasible solution in the primal $\mathcal{O}$ and dual $\mathcal{E}^{int}$, we have the proof in the theorem.
\end{proof}

We now prove Theorem \ref{theorem:weak_duality_result}.
\begin{proof}
Combining Claim \ref{lemma1} and \ref{lemma2}, we obtain that $\text{opt}(\mathcal{O}) \geq \text{opt}(\mathcal{E})$.
\end{proof}

\subsection{Claim to prove the existence of continuous bounds $U(\vect{v})$ and $L(\vect{v})$ for the feasible solution in $\mathcal{E}$ in the proof of Lemma \ref{lemma:dual_achievable} }\label{sec:continuous_bounds_claim}

\begin{claim}{\label{claim:continuity_bounds}}
    Consider a  function $\nu(.,.): \mathcal{C}^1(\mathbbm{R} \times \mathbbm{R}^k \rightarrow \mathbbm{R})$ s.t. zeros of $\nu(.,\vect{v})$ i.e. $\{\mathfrak{u}:\nu(\mathfrak{u},\vect{v})=0\}$ is a \textit{countable} set for every $\vect{v} \in \mathbbm{R}^k$. 
    
    It also satisfies the following two conditions for some constant $C$ independent of $\vect{v}$. Note that $v_{\floor{i}}$ denotes the $i^{th}$ largest component of $\vect{v}$ for every $i \in [k]$.
    \begin{itemize}
        \item $\set[\Big]{\mathfrak{u} \in \mathbb{R}: \nu(\mathfrak{u},\vect{v})=0 \text{ and } \nu_u(\mathfrak{u},\vect{v}) < 0 }$ is upper bounded by $C+ v_{\floor{k}}$ for every $\vect{v} \in \mathbbm{R}^k$. \zhihao{what is $v_{\floor{k}}$?}
        \sss{Check the notation subsection. We defined $v_{\floor{k}}$ to be $k^{th}$ largest component (or basically the largest component of $\vect{v}$.)}
        \item $\forall \vect{v} \in \mathbbm{R}^k \text{ } \exists \text{ }\mathcal{U}$ s.t $\nu(u,\vect{v}) \geq 0 \text{ }\forall u >\mathcal{U}$.
    \end{itemize}
Then there exists $U: \mathcal{C}^0(\mathbbm{R}^k \rightarrow \mathbbm{R})$ s.t. $\nu(u,\vect{v}) \geq 0 \text{ }\forall u\geq U(\vect{v}) \text{ }\forall \vect{v} \in \mathbbm{R}^k$.
\end{claim}

\begin{proof}
    We prove this statement as follows. For every $\vect{v} \in \mathbbm{R}^k$, denote the largest zero of $\nu(.,\vect{v})$ by $p(\vect{v})$. Now construct $U(\vect{v}) = \max(p(\vect{v}), C+ v_{\floor{k}})$. Observe that $\nu(u,\vect{v}) \geq 0 \text{ }\forall u\geq U(\vect{v}) \text{ }\forall \vect{v} \in \mathbbm{R}^k$ follows from the second assumption.

    To prove continuity at $\vect{v} \in \mathbbm{R}^k$, we aim to show the following statement. 

    \begin{equation}{\label{eqn:continuity}}
        \forall \epsilon >0 \text{ } \exists \delta >0 \text{ s.t. } \forall \vect{z} \in \mathbbm{R}^k; \text{ }||\vect{z}-\vect{v}||_1 \leq \delta \implies |U(\vect{z}) - U(\vect{v})| \leq \epsilon.
    \end{equation}

   We now show $U(\vect{v})$ is continuous for every $\vect{v} \in \mathbbm{R}^k$. We consider two exhaustive cases below.

    \textit{Case 1}: $p(\vect{v}) \geq C+ v_{\floor{k}}$ and thus, $U(\vect{v}) = p(\vect{v})$. Recall from the second assumption that $\nu(u,\vect{v})$ goes from negative to positive at $u=p(\vect{v})$.

     We now discuss the construction of $\delta$ below for a \textit{sufficiently small} $\epsilon$ as follows. 

    For some $\epsilon$ sufficiently small, we now discuss the construction of an interval around $p_{\vect{v}}$ as $[\mathcal{A}_\vect{v},\mathcal{B}_\vect{v}] = [p_{\vect{v}}-\frac{\epsilon}{2}, p_{\vect{v}}+\frac{\epsilon}{2}]$. Since, $\epsilon$ is sufficiently small, observe that $\nu(\mathcal{A}_{\vect{v}},\vect{v}) = -\zeta_1$ and $\nu(\mathcal{B}_{\vect{v}},\vect{v}) = \zeta_2$ for some $\zeta_1,\zeta_2>0$.

    Observe that $\nu(\mathcal{A}_{\vect{v}},\vect{v})$ must be continuous in $\vect{v}$ and thus there must exist $\delta_1>0$ s.t. $\forall \vect{z} \in \mathbbm{R}^k; \text{ }||\vect{z}-\vect{v}|| \leq \delta_1 \implies |\nu(\mathcal{A}_{\vect{v}},\vect{v}) - \nu(\mathcal{A}_{\vect{v}},\vect{z})| \leq \frac{\zeta_1}{2}$ and similarly, choose $\delta_2>0$ s.t. $\forall \vect{z} \in \mathbbm{R}^k; \text{ }||\vect{z}-\vect{v}_1|| \leq \delta_2 \implies |\nu(\mathcal{B}_{\vect{v}},\vect{v}) - \nu(\mathcal{B}_{\vect{v}},\vect{z})| \leq \frac{\zeta_2}{2}$.\footnote{Note that we can do this since, the continuity result holds true for every $\epsilon>0$ and we can choose any $\epsilon$ we want. } 

     \textbf{We choose $\delta = \min(\delta_1,\delta_2,\frac{\epsilon}{2})$ and consider any $\vect{z} \in \mathbbm{R}^k$ satisfying $||\vect{z}-\vect{v}||_1< \delta$ and show \eqref{eqn:continuity}}

    This  implies $\nu(\mathcal{A}_{\vect{v}},\vect{z}) < -\frac{\zeta_1}{2}<0$ and $\nu(\mathcal{B}_{\vect{v}},\vect{z}) > \frac{\zeta_2}{2}>0$ since, $\nu(\mathcal{A}_{\vect{v}},\vect{v}) = -\zeta_1$ and $\nu(\mathcal{B}_{\vect{v}},\vect{v}) = \zeta_2$. Thus, from \textit{intermediate value and LMVT theorem}, there must exist some zero of $\nu(.,\vect{z})$ (call it $u_0$) with $u_0 \in [\mathcal{A}_{\vect{v}},\mathcal{B}_{\vect{v}}]$ with $\nu(.)$ going from negative to positive i.e. $\nu_u(u_0,\vect{z}) \geq 0$ and $\nu(u_0,\vect{z})=0$

    Also observe that, 

    \begin{align}{\label{eqn1:temp}}
        \mathcal{B}_{\vect{v}} & \overset{(a)}{\geq} C+ v_{\floor{k}}  + \frac{\epsilon}{2}\nonumber\\
        & \overset{(b)}{>} C+ z_{\floor{k}} + \frac{\epsilon}{2} - \delta \geq C+ z_{\floor{k}}
    \end{align}

    $(a)$ follows from $p(\vect{v}) \geq C+ v_{\floor{k}}$ and the definition of $\mathcal{B}_{\vect{v}}$. $(b)$ follows from the fact $||\vect{v}-\vect{z}||\leq \delta$. Observe that the last inequality follows from  
    the fact that $\delta< \frac{\epsilon}{2}$.

     Equation \eqref{eqn1:temp} and the first condition in Claim \ref{claim:continuity_bounds}implies that $U(u,\vect{z})$ cannot have a zero in $u$ with $U(u,\vect{z})$ going from positive to negative on the right side of interval $[\mathcal{A}_{\vect{v}},\mathcal{B}_{\vect{v}}]$. This implies that the largest zero of $\nu(u,\vect{z})$ in interval $[\mathcal{A}_{\vect{v}},\mathcal{B}_{\vect{v}}]$ must be the largest zero and thus, it must be $p(\vect{z})$. Also, $p(\vect{z}) > \mathcal{A}_{\vect{v}} > C+z_{\floor{k}}$ and thus, $U(\vect{z}) = p(\vect{z})$

    We thus, have $|U(\vect{z})-U(\vect{v})| = |p(\vect{z})-p(\vect{v})| \leq \epsilon$, since $\max(|\mathcal{A}_{\vect{v}} - p(\vect{v})|,|\mathcal{B}_{\vect{v}} - p(\vect{v})|) < \epsilon$ thus proving the desired result in \eqref{eqn:continuity}

    \textit{Case 2}: $p(\vect{v}) < C+ v_{\floor{k}}$, thus $U(\vect{v}) = C+ v_{\floor{k}}$.

    Since, $p(\vect{v})$ was the largest zero of $\nu(.,\vect{v})$ with $\nu_u(p(\vect{v}), \vect{v}) >0$, observe that $\nu(C+ v_{\floor{k}}, \vect{v}) = \zeta >0$. Now since $\nu(u,\vect{v})$ is continuous, there must exist $\delta_1>0$ s.t. $\forall \hat{u} \in \mathbbm{R}$ satisfying $|\hat{u}-(C+v_{\floor{k}})| < \delta_1$, we have  $|\nu(\hat{u},\vect{v})- \nu(\hat{u},\vect{v})| < \frac{\zeta}{4}$. Similarly, there must exist $\delta_2>0$ s.t. $ \forall \vect{z} \in \mathbb{R}$ satisfying $||\vect{v}-\vect{z}||_1 < \delta_2$, we have $|\nu(C+z_{\floor{k}},\vect{v}) - \nu(C+z_{\floor{k}},\vect{z})| \leq \frac{\zeta}{4}$. \footnote{This holds since continuity result holds for every $\epsilon>0$.}

    \textbf{Now we choose $\delta = \min(\delta_1,\delta_2,\epsilon)$ and consider any $\vect{z} \in \mathbb{R}^k$ satisfying $||\vect{z}-\vect{v}||_1 \leq \delta$ to show \eqref{eqn:continuity}.}

    Observe that, since, $|(C+v_{\floor{k}}) - (C+ {z}_{\floor{k}})| \leq \delta \leq \delta_1$, we must have $\nu(C+z_{\floor{k}},\vect{v}) = \nu(C+v_{\floor{k}},\vect{v}) + (\nu(C+z_{\floor{k}},\vect{v})- \nu(C+ {v}_{\floor{k}},\vect{v})) \geq \zeta -\frac{\zeta}{4} = \frac{3\zeta}{4}$. Similarly, since $||\vect{v}-\vect{z}|| \leq \delta \leq \delta_2$, we must have $\nu(C+z_{\floor{k}},\vect{z}) = \nu(C+z_{\floor{k}},\vect{v}) + (\nu(C+z_{\floor{k}},\vect{v})-\nu(C+z_{\floor{k}},\vect{v})) \geq \frac{3\zeta}{4}-\frac{\zeta}{4} = \frac{\zeta}{2}>0$.

    Since, $\nu(C+z_{\floor{k}},\vect{z})>0$ and since the second condition in lemma \ref{claim:continuity_bounds} says that there can be no zero of $\nu(u,\vect{z})$ in $u$ with $\nu(u,\vect{z})$ going from positive to negative beyond $C+z_{\floor{k}}$ we have the largest zero of $\nu(.,\vect{v})$ should be smaller than $C+z_{\floor{k}}$ and thus $U(\vect{z}) = C+z_{\floor{k}}$.

    Now observe that $|U(\vect{z})-U(\vect{y})| = |(C+z_{\floor{k}})-(C+y_{\floor{k}})| \leq \delta \leq \epsilon$, thus proving the desired statement in \eqref{eqn:continuity}.








\end{proof}

\renewcommand{\zhihao}[1]{\iftoggle{COMMENTS}{{\color{brown}[ZZ: \textsf{#1}]}}{}}
\renewcommand{\sss}[1]{ {\color{orange} [ Sahasrajit: {#1} ]  } }


\end{document}